\documentclass[a4paper,12pt]{article}

\usepackage{amsmath,amssymb,amsthm}
\allowdisplaybreaks
\usepackage{abstract}
\usepackage{titling}
\usepackage{tabularx}
\usepackage[japanese,english]{babel}
\usepackage[table]{xcolor}
\usepackage{adjustbox}
\usepackage{subcaption}
\usepackage[margin=1.25in]{geometry}
\usepackage{graphicx}
\usepackage{mathtools}
\usepackage{pdflscape}
\usepackage{colortbl}
\usepackage{booktabs}
\usepackage{comment}
\usepackage{amsfonts}
\usepackage{bbm,bm}
\usepackage{physics}
\usepackage{float}
\usepackage[bottom]{footmisc}
\usepackage{cancel}
\usepackage{cases}
\usepackage{ascmac}
\usepackage{ulem}
\usepackage{stmaryrd}
\usepackage{setspace}
\usepackage[labelfont=bf]{caption}
\usepackage{setspace}
\usepackage{hyperref}
\hypersetup{
    colorlinks=true,
    linkcolor=orange,
    filecolor=orange,
    urlcolor=orange,
    citecolor=orange,
}

\usepackage{physics} 
\usepackage{titlesec}
\titleformat{\section}
  {\normalfont\large\bfseries}{\thesection}{1em}{}
  \titleformat{\subsection}
  {\normalfont\normalsize\bfseries}{\thesubsection}{1em}{}

\newcommand{\indep}{\mathop{\perp\!\!\!\!\perp}}
\newtheorem*{Estimation and Inference}{Estimation and Inference}
\newtheorem*{Applicability of DID-IV designs in practice}{Applicability of DID-IV designs in practice}
\newtheorem*{Repeated cross sections}{Repeated cross sections}

\newtheorem{Proof of Theorem}{Proof of Theorem}
\newtheorem{Proof of Lemma}{Proof of Lemma}
\newtheorem{proof of}{proof of}

\newcommand{\1}{\mbox{1}\hspace{-0.25em}\mbox{l}}
\newcommand{\E}{\mathbb{E}}
\newcommand{\Prb}{\mathbb{P}}

\theoremstyle{definition}
\newtheorem{Assumption}{Assumption}
\newtheorem*{Discussion in section 5.4}{Discussion in section 5.4}
\newtheorem*{Remark*}{Remark}

\newtheorem{Remark}{Remark}
\newtheorem{Theorem}{Theorem}
\newtheorem{Corollary}{Corollary}

\newtheorem{Lemma}{Lemma}
\newtheorem{Algorithm}{Algorithm}

\usepackage[round]{natbib}
\usepackage{multibib}
\newcites{online}{References}

\title{Instrumented Difference-in-Differences with Heterogeneous Treatment Effects\thanks{This paper won the Kanematsu Prize from the Research Institution for Economics and Business Administration, Kobe University, in 2023. It also won the Best Young Scholar Award at SETA (Symposium on Econometric Theory and Applications) Conference, Taipei, in 2024. I am grateful to my advisors, Daiji Kawaguchi and Ryo Okui, for their continued guidance and support. I am also thankful to Alberto Abadie, Louis-Daniel Pape, Shakeeb Khan, Pedro H.C. Sant'Anna, Masamune Iwasawa, Shoya Ishimaru, Shosei Sakaguchi, Takahide Yanagi, Edward Vytlacil, Jonas Skjold Raaschou-Pedersen, and the three anonymous referees of the Kanematsu Prize for their insightful comments. I gratefully acknowledge the support provided by the JSPS KAKENHI Grant JP 24KJ0817.
All errors are my own.}}
\author{Sho Miyaji\thanks{Department of Economics, Yale University, 28 Hillhouse Avenue, New Haven, CT 06511, USA; Email: \href{mailto:sho.miyaji@yale.edu}{sho.miyaji@yale.edu}.}}

\date{\today}

\begin{document}
\onehalfspacing
\maketitle
\begin{abstract}
Many studies exploit variation in the timing of policy adoption across units as an instrument for treatment. This paper formalizes the underlying identification strategy as a staggered instrumented difference-in-differences (DID-IV). In this design, a Wald-DID estimand, which scales the DID estimand of the outcome by the DID estimand of the treatment, captures the cohort specific local average treatment effect on the treated (CLATT). We propose a credible estimation method in this design that is robust to treatment effect heterogeneity. We also provide simultaneous inference procedures based on the multiplier bootstrap. We illustrate the empirical relevance of our findings, estimating returns to schooling in the United Kingdom. In this application, the two-way fixed effects instrumental variable regression, the conventional approach to implement DID-IV designs, yields a negative estimate. By contrast, our estimation method indicates a substantial gain from schooling.\par
\end{abstract}

\section{Introduction}\label{sec1}
To identify the effect of a treatment on an outcome, many studies exploit variation in the timing of policy adoption across units as an instrument for treatment. For example, to estimate the effect of Medicaid coverage on mortality among low-income adults, \cite{Miller2021-mj} exploit variation in the timing of the adoption of the Medicaid expansions across states as an instrument for Medicaid coverage. Similarly, to estimate the causal link between parents' and children's education attainment, \cite{Black2005-aw} exploit variation in the timing of the implementation of school reforms across municipalities as an instrument for parents' education attainment. Notably, the identification strategy underlying these studies is similar to difference-in-differences (DID) designs, in that variation in the timing of a policy shock is used to identify treatment effects. At the same time, it differs from DID designs in that this variation is used to construct an instrument rather than a treatment.\par
In this paper, we formalize the underlying identification strategy as a staggered instrumented difference-in-differences (DID-IV). We establish the target parameter and the identifying assumptions in this design, and develop a credible estimation method that is robust to treatment effect heterogeneity. We illustrate the empirical relevance of our findings with the setting of \cite{Oreopoulos2006-bn}, estimating returns to schooling in the United Kingdom. In this application, we demonstrate that the choice of estimation method
matters in practice.\par
We start by defining the target parameter and the identifying assumptions
in this design. We consider multiple period settings in which units
become exposed to the instrument at different points in time and remain
exposed thereafter. In this setting, we partition units into mutually
exclusive and exhaustive cohorts based on the initial exposure date of
the instrument. Our target parameter is the cohort specific local average
treatment effect on the treated (CLATT): this parameter measures the
treatment effects for units that belong to cohort $e$ and are induced
into treatment by the instrument (the compliers) in a given period
following the initial exposure. To identify each CLATT, we introduce the
Wald-DID estimand, which scales the DID estimand of the outcome by the
DID estimand of the treatment. We construct two such estimands, one using
the never exposed cohort as the control group and the other using the
not-yet-exposed cohorts. We show that each estimand identifies the CLATT
in the corresponding post-exposed periods under our staggered DID-IV
design. The key identifying assumptions are a monotonicity assumption,
parallel trends assumptions in the treatment and the outcome, and a no
anticipation assumption, which requires that units do not adjust their
treatment status before the initial exposure to the instrument.\par
In practice, when researchers implicitly rely on a staggered DID-IV
design, they typically implement it using two-way fixed effects
instrumental variable (TWFEIV) regressions (e.g., \cite{Johnson2019-kb},
\cite{Lundborg2014-gm}, \cite{Black2005-aw}, \cite{Akerman2015-hh},
\cite{Bhuller2013-ki}, and \cite{Miller2019-ok}). In a companion paper (\cite{Miyaji2023-tw}),
however, we show that in more than two periods, the TWFEIV estimand
generally fails to summarize treatment effects if the effect of the
instrument on the treatment or on the outcome is not stable over
time.\footnote{See also \cite{De_Chaisemartin2020-dw}, who decompose the numerator and denominator of the TWFEIV estimand separately, and point out that this estimand does not necessarily admit a causal interpretation in Section 3.4 of their Web Appendix.}\par
Given this problem, we next propose a credible estimation method in
staggered DID-IV designs that is robust to treatment effect
heterogeneity. The estimation procedure consists of two steps: we subset
the data that contain only two cohorts and two periods and then, in each
data set, we run the TWFEIV regression. We call this a stacked two stage
least squares (STS) regression. We establish the asymptotic normality of
the resulting estimators and show how to construct simultaneous
confidence bands for the CLATTs, using the multiplier bootstrap. To
summarize the treatment effects in this design, we further propose a
weighting scheme following \cite{Callaway2021-wl}, where each weight
reflects the share of compliers in cohort $e$ in a given post-exposed
period. Finally, we discuss how to assess the plausibility of the
identifying assumptions in practice.\par
We extend our DID-IV framework along four dimensions. First, we relax the no anticipation assumption and allow for limited
anticipation as in \cite{Callaway2021-wl}. Second, we show that our framework can be applied to settings with a non-binary, ordered treatment. Third, we consider repeated cross sections. Fourth, we incorporate discrete covariates into our estimation method.\par
We illustrate our findings with the setting of \cite{Oreopoulos2006-bn},
who estimates returns to schooling in the UK, exploiting variation in the
timing of the implementation of school reforms between Britain and
Northern Ireland as an instrument for education attainment. In this
application, we first assess the plausibility of the DID-IV
identification strategy implicitly adopted by \cite{Oreopoulos2006-bn}.
We then estimate the TWFEIV regression in the author's setting. We find
that the TWFEIV estimate is strictly negative and not significantly
different from zero. Next, we use our estimation method to reassess
returns to schooling in the UK. We find that our STS estimates are
positive in each relative period after the school reform, and our
weighting scheme yields a summary estimate of $0.240$, which is more
plausible than the negative TWFEIV estimate. Once we adjust for
differences in age composition using the saturated STS regression, the
summary estimate falls to $0.164$, in line with conventional estimates of
the return to schooling (\cite{GUNDERSON202039}). Finally, we conduct
weak instrument diagnostics and confirm that these conclusions are not
driven by weak identification.\par
At the end of this paper, we describe the relationship between DID-IV
and Fuzzy DID considered in \cite{De_Chaisemartin2018-xe} (hereafter,
``dCDH''). Whereas our staggered DID-IV design builds on the two-period, two-group
DID-IV designs studied in \cite{chasemartin2010-ch} and
\cite{Hudson2017-tm},
dCDH formalize these designs differently, calling them Fuzzy DID. Both designs consider the same setting and the same
Wald-DID estimand, but they differ in their target parameters and identifying
assumptions. In Section \ref{sec6}, we clarify the implications of these differences
for the restrictions on treatment adoption behavior, the
interpretation of the target parameter, and the use of the Wald-DID
estimand. We also illustrate these implications in our empirical
application. This comparison would help empirical researchers choose between the two designs in a given application.\par
The rest of the paper is organized as follows. The following subsection discusses the related literature. Section \ref{sec2} formalizes the target parameter and identifying assumptions in staggered DID-IV designs. Section \ref{sec4} presents our estimation method. Section \ref{sec3} contains extensions. Section \ref{sec5} presents our empirical application. Section \ref{sec6} compares DID-IV with Fuzzy DID. Section \ref{sec7} concludes. All proofs are given in the Appendix.
\subsection*{Related literature}\label{sec1.1}
Our paper is related to the recent DID-IV literature (\cite{chasemartin2010-ch}; \cite{Hudson2017-tm}; \cite{De_Chaisemartin2018-xe}; \cite{Miyaji2023-tw}; \cite{dechaisemartin2024differenceindifferencesestimatorstreatmentscontinuously}; \cite{chen2025efficientdifferenceindifferenceseventstudy}; \cite{helmers2025judge}; \cite{Miyaji2025-ddd}; \cite{raaschoupedersen2026doubly}).\footnote{While \cite{Ye2023-ju} also consider the ``instrumented difference-in-differences'', their target parameter is the average treatment effect (ATE), and they impose strong assumptions for the Wald-DID estimand to identify this parameter. We therefore view \cite{Ye2023-ju} as distinct from the recent DID-IV literature.} In this literature, \cite{chasemartin2010-ch} is the first to formalize two-period, two-group ($2 \times 2$) DID-IV designs, where the instrument (a policy shock) is assigned only in the second period, and units are partitioned into exposed and unexposed groups. In this setting, \cite{chasemartin2010-ch} shows that, under a monotonicity assumption and parallel trends assumptions in the treatment and the outcome, the Wald-DID estimand identifies the local average treatment effect on the treated (LATET).\footnote{\cite{Blundell565} also consider a DID-IV setting with a binary treatment and a binary instrument in Subsection~F.2 of Section~IV, but they focus on the switcher local average treatment effect on the treated (SLATET), as in \cite{De_Chaisemartin2018-xe}, and impose a parallel trends assumption on untreated outcomes rather than on unexposed outcomes.} Following \cite{chasemartin2010-ch}, \cite{Hudson2017-tm} also study $2 \times 2$ DID-IV designs with a non-binary, ordered treatment.\par
This paper contributes to the DID-IV literature in three ways. First, we extend $2 \times 2$ DID-IV designs to multiple period settings with the staggered adoption of the instrument across units. Second, we provide a credible estimation method in staggered DID-IV designs in the presence of heterogeneous treatment effects. Finally, we investigate the detailed connections between DID-IV and the Fuzzy DID considered in \cite{De_Chaisemartin2018-xe} and discuss their implications for empirical practice.\par
Recently, \cite{chen2025efficientdifferenceindifferenceseventstudy} extend our DID-IV framework by allowing for covariates, \cite{helmers2025judge} consider an extension to the continuous treatment case, and \cite{raaschoupedersen2026doubly} develops doubly robust estimators in staggered DID-IV designs with covariates. \cite{Miyaji2025-ddd} also extends the DID-IV framework to triple differences settings. Independently of our work, \cite{dechaisemartin2024differenceindifferencesestimatorstreatmentscontinuously} consider a continuous treatment with a continuous, static, and potentially non-staggered instrument.\par
Our paper is also related to the recent DID literature in two ways. First, this paper provides an alternative identification strategy to DID designs, addressing settings in which the parallel trends assumption is unlikely to hold or no credible control group exists. Second, in line with the DID literature, this paper develops a reliable estimation method for DID-IV designs in the presence of heterogeneous treatment effects. Recently, several studies have shown that two-way fixed effects regressions, or their dynamic specifications, generally fail to summarize the treatment effects under heterogeneous treatment effects (\cite{Athey2022-uo}; \cite{Borusyak2021-jv}; \cite{Callaway2021-wl}; \cite{De_Chaisemartin2020-dw}; \cite{Goodman-Bacon2021-ej}; \cite{Imai2021-dn}; \cite{Sun2021-rp}; \cite{callaway2024differenceindifferencescontinuoustreatment}). Some of these studies propose credible estimation methods that deliver a sensible estimand and are robust to treatment heterogeneity. In the same spirit, this paper proposes an alternative to the TWFEIV regression and illustrates its usefulness through an empirical application.\par

\section{Staggered DID-IV designs}\label{sec2}
In this section, we formalize DID-IV designs in multiple period settings
where the adoption of the instrument is staggered across units
(e.g., \cite{Black2005-aw}; \cite{Bhuller2013-ki}; \cite{Lundborg2014-gm}; \cite{Meghir2018-bk}). We call this a staggered DID-IV design, and establish the target parameter
and identifying assumptions.

\subsection{Set up}\label{sec2.1}
We introduce the notation we use throughout this paper. We consider a panel data setting with $T$ periods and $N$ units. For any random variable $R$, let $\mathcal{S}(R)$ denote its support. For each $i \in \{1,\dots ,N\}$ and $t \in \{1,\dots,T\}$, let $Y_{i,t}$ denote the outcome, and $D_{i,t}\in \{0,1\}$ denote the treatment status: $D_{i,t}=1$ if unit $i$ receives the treatment in period $t$, and $D_{i,t}=0$ otherwise. Let $Z_{i,t}\in \{0,1\}$ denote the instrument status: $Z_{i,t}=1$ if unit $i$ is exposed to the instrument in period $t$, and $Z_{i,t}=0$ otherwise. Let $D_i=(D_{i,1},\dots,D_{i,T})$ and $Z_i=(Z_{i,1},\dots,Z_{i,T})$ denote the treatment path and the instrument path for unit $i$, respectively. Throughout Sections \ref{sec2}--\ref{sec3}, we assume that $\{(Y_{i,t},D_{i,t},Z_{i,t})_{t=1}^{T}\}_{i=1}^{N}$ are independent and identically distributed (i.i.d.). \par
We impose the following assumption about the assignment process of the instrument.\par
\begin{Assumption}[Staggered adoption for $Z_{i,t}$]
\label{sec3as1}
$Z_{i,1}=0$ for all $i$. For all $s,t \in \{1,\dots,T\}$ with $s < t$,
$Z_{i,s} \leq Z_{i,t}$.
\end{Assumption}
Assumption \ref{sec3as1} requires that no unit is exposed to the
instrument in period $t=1$ and that, once exposed, units remain exposed. In the DID literature, several studies impose a similar
assumption on the treatment and call it ``staggered treatment adoption''
(e.g., \cite{Athey2022-uo}; \cite{Callaway2021-wl}; \cite{Sun2021-rp}).\par
Given Assumption \ref{sec3as1}, one can uniquely characterize the
instrument path of unit $i$ by the initial exposure date of the
instrument, denoted as $E_{i}=\min\{t: Z_{i,t}=1\}$. If unit $i$ is
never exposed to the instrument, we define $E_{i}=\infty$. Based on the
initial exposure date $E_i$, one can uniquely partition units into
mutually exclusive and exhaustive cohorts $e \in \{2,3,\dots,T,\infty\}$.
Let $\bar{e}=\max_{i=1,\dots,N}E_i$ denote the largest cohort value in
the dataset. Let $\mathcal{E}=\mathcal{S}(E_i)\setminus \{\bar{e}\}
\subseteq \{2,3,\dots,T\}$ denote the support of $E_i$ excluding
$\bar{e}$. In contrast to the staggered assignment of the instrument, we allow a
general adoption process for the treatment: the treatment can potentially turn on
and off repeatedly over time.\par
Next, we introduce the potential outcomes framework. Let $Y_{i,t}(d,z)$ denote the potential outcome in period $t$ when unit $i$ receives the treatment path $d \in \mathcal{S}(D)$ and the instrument path $z \in \mathcal{S}(Z)$. Similarly, let $D_{i,t}(z)$ denote the potential treatment status in period $t$ when unit $i$ receives the instrument path $z \in \mathcal{S}(Z)$.\par
Assumption \ref{sec3as1} allows us to rewrite $D_{i,t}(z)$ in terms of
the initial exposure date $E_i$. Let $D_{i,t}^{e}$ denote the potential
treatment status in period $t$ if unit $i$ is first exposed to the
instrument in period $e$. Let $D_{i,t}^{\infty}$ denote the potential
treatment status in period $t$ if unit $i$ is never exposed to the
instrument. We call $D_{i,t}^{\infty}$ the ``never exposed treatment''. We can then write the observed treatment status for unit $i$ in period $t$ as $D_{i,t}=D_{i,t}^{\infty}+\sum_{2 \leq e \leq T}(D_{i,t}^{e}-D_{i,t}^{\infty})\cdot \mathbf{1}\{E_i=e\}$.\par
Next, we introduce the group variable that describes the type of unit
$i$ in period $t$, based on the response of its potential treatment
status to the instrument path. Let
$G_{i,e,t} \equiv (D_{i,t}^{\infty},D_{i,t}^{e})$ be the group variable
in period $t$ for unit $i$ and the initial exposure date $e$. Following
the terminology in \cite{Imbens1994-qy}, we define
$G_{i,e,t}=(0,0) \equiv NT_{e,t}$ as the never-takers,
$G_{i,e,t}=(1,1) \equiv AT_{e,t}$ as the always-takers,
$G_{i,e,t}=(0,1) \equiv CM_{e,t}$ as the compliers, and
$G_{i,e,t}=(1,0) \equiv DF_{e,t}$ as the defiers.\par
Here, we assume that potential outcomes $Y_{i,t}(d,z)$ depend only on the current treatment status $d_t$ and the instrument path $z$.\footnote{In the DID literature, several studies adopt this assumption in non-staggered treatment settings (e.g., \cite{De_Chaisemartin2020-dw}; \cite{Imai2021-dn}).}
\begin{Assumption}[No carryover assumption]
\label{sec3as2}
$Y_{i,t}(d,z)=Y_{i,t}(d_t,z)$ for all $z \in \mathcal{S}(Z)$, $d\in \mathcal{S}(D)$, and $t\in \{1,\dots,T\}$, where $d=(d_1,\dots,d_T)$ is a generic element of the treatment path $D_i$, with $d_t$ its period-$t$ component.
\end{Assumption}
Finally, we assume that the instrument path affects potential outcomes only through the treatment. This assumption is common in the IV literature (e.g., \cite{Imbens1994-qy}; \cite{Abadie2003-ry}).
\begin{Assumption}[Exclusion restriction]
\label{sec3as3}
$Y_{i,t}(d,z)=Y_{i,t}(d)$ for all $z \in \mathcal{S}(Z)$, $d \in \mathcal{S}(D)$, and $t \in \{1,\dots,T\}$.
\end{Assumption}
Given Assumptions \ref{sec3as2} and \ref{sec3as3}, we can write the
potential outcome $Y_{i,t}(d,z)$ as $Y_{i,t}(d_t)$, and the observed
outcome as
$Y_{i,t}=Y_{i,t}(D_{i,t})=D_{i,t}Y_{i,t}(1)+(1-D_{i,t})Y_{i,t}(0)$.
Here, we introduce the outcome in period $t$ if unit $i$ is assigned to
the instrument path $z \in \mathcal{S}(Z)$: $Y_{i,t}(D_{i,t}(z)) \equiv D_{i,t}(z)Y_{i,t}(1)+(1-D_{i,t}(z))Y_{i,t}(0)$.
We can similarly write the outcomes for cohort $e$ and cohort $\infty$ as $Y_{i,t}(D_{i,t}^{e})$ and $Y_{i,t}(D_{i,t}^{\infty})$, respectively. We refer to
$Y_{i,t}(D_{i,t}^{\infty})$ as the ``never exposed outcome''.\par
Henceforth, we maintain Assumptions \ref{sec3as1}--\ref{sec3as3}. In the next section, we define the target parameter in staggered DID-IV designs.

\subsection{Target parameter}\label{sec2.2}
In staggered DID-IV designs, our target parameter is the cohort specific local average treatment effect on the treated (CLATT) in period $t$:
\begin{align*}
CLATT_{e,t}=\E[Y_{i,t}(1)-Y_{i,t}(0)|E_i=e,CM_{e,t}].
\end{align*}\par
This parameter measures the treatment effects in period $t$ for those who belong to cohort $e$ and are the compliers $CM_{e,t}$, that is, units induced into treatment by the instrument in that period. It is conditional on cohort $e$ because, as in DID designs, the counterfactual is constructed from the unexposed cohorts, so that only the treatment effects on the exposed cohort are identified. It is also conditional on the compliers because, as in IV designs, treatment effects are identified only for units whose treatment status is shifted by the instrument. Note that the CLATT can potentially vary across cohorts and over time because it depends on cohort $e$, period $t$, and the compliers $CM_{e,t}$.

\subsection{Identifying assumptions}\label{sec2.3}
In this section, we establish the identifying assumptions in staggered
DID-IV designs. To identify each $CLATT_{e,t}$, we first
consider the following estimand:
\begin{align*}
    w^{DID}_{e,t}=\frac{\E[Y_{i,t}-Y_{i,e-1}|E_i=e]-\E[Y_{i,t}-Y_{i,e-1}|E_i=\infty]}{\E[D_{i,t}-D_{i,e-1}|E_i=e]-\E[D_{i,t}-D_{i,e-1}|E_i=\infty]},
\end{align*}
for $e \in \mathcal{E}$ and $t \in \{e,\dots,T\}$. This estimand is the
ratio between the DID estimand of the outcome and the DID estimand of the treatment, where
the pre-exposed period is $e-1$ and the control group is the never
exposed cohort. Following the terminology in
\cite{De_Chaisemartin2018-xe}, we call this the Wald-DID estimand. For this estimand,
we assume that the largest cohort value in the dataset is
$\bar{e}=\infty$.\par
In some applications, however, there may be no never exposed cohort, or its
sample size may be too small to serve as a reliable control group. Moreover,
the parallel trends assumptions based on the never exposed cohort may be
viewed as less credible in certain empirical settings. In such cases, it is
natural to instead use the not-yet-exposed cohorts as the control group. We
therefore also consider the following alternative Wald-DID estimand:
\begin{align*}
    w^{DID,ny}_{e,t}=\frac{\E[Y_{i,t}-Y_{i,e-1}|E_i=e]-\E[Y_{i,t}-Y_{i,e-1}|Z_{i,t}=0]}{\E[D_{i,t}-D_{i,e-1}|E_i=e]-\E[D_{i,t}-D_{i,e-1}|Z_{i,t}=0]},
\end{align*}
for $e \in \mathcal{E}$ and $t \in \{e,\dots,T\}$ such that $t < \bar{e}$.
In this Wald-DID estimand, the control group is the not-yet-exposed cohorts,
that is, the cohorts not exposed to the instrument by period $t$.\par
We now present the identifying assumptions under which these estimands capture the corresponding $CLATT_{e,t}$. We begin with the assumptions common to both estimands.
\begin{Assumption}[Monotonicity assumption]
\label{sec3as4}
For all $e\in \mathcal{E}$ and all $t \in \{e,\dots,T\}$, $\Prb(D_{i,t}^{e} \geq D_{i,t}^{\infty})=1$ or $\Prb(D_{i,t}^{e} \leq D_{i,t}^{\infty})=1$.
\end{Assumption}
Assumption \ref{sec3as4} implies that the group variable
$G_{i,e,t} \equiv (D_{i,t}^{\infty},D_{i,t}^{e})$ can take three values
with non-zero probability for all $e \in \mathcal{E}$ and all
$t \in \{e,\dots,T\}$. Hereafter, we consider the type of the
monotonicity assumption that rules out the existence of the defiers
$DF_{e,t}$ for all $t \geq e$ in any cohort $e \in \mathcal{E}$.\par
\begin{Assumption}[No anticipation in the first stage]
\label{sec3as5}
$D_{i,t}^{e}=D_{i,t}^{\infty}$ for all $e\in \mathcal{S}(E_i)\setminus\{\infty\}$ and $t \in \{1,\dots,T\}$ such that $t < e$.
\end{Assumption}
Assumption \ref{sec3as5} requires that the potential treatment status in any
period before the initial exposure to the instrument is equal to the never
exposed treatment, ruling out anticipatory behavior in the first stage. In the DID literature, an analogous no
anticipation assumption is commonly imposed on the treatment
(e.g., \cite{Callaway2021-wl}; \cite{Sun2021-rp}).\par
Next, we present the parallel trends assumptions and the relevance
conditions, which depend on the choice of the control group. We first
state them for the estimand $w^{DID}_{e,t}$.\par
\begin{Assumption}[Parallel trends assumption in the treatment based on a never exposed cohort]
\label{sec3as6}
For each $e\in \mathcal{E}$ and $t\in \{2,\dots,T\}$ such that $t \geq e$,
\begin{align*}
\E[D_{i,t}^{\infty}-D_{i,t-1}^{\infty}|E_i=e]=\E[D_{i,t}^{\infty}-D_{i,t-1}^{\infty}|E_i=\infty].
\end{align*}
\end{Assumption}

\begin{Assumption}[Parallel trends assumption in the outcome based on a never exposed cohort]
\label{sec3as7}
For each $e \in \mathcal{E}$ and $t\in \{2,\dots,T\}$ such that $t \geq e$,
\begin{align*}
\E[Y_{i,t}(D_{i,t}^{\infty})-Y_{i,t-1}(D_{i,t-1}^{\infty})|E_i=e]=\E[Y_{i,t}(D_{i,t}^{\infty})-Y_{i,t-1}(D_{i,t-1}^{\infty})|E_i=\infty].
\end{align*}
\end{Assumption}
Assumptions \ref{sec3as6} and \ref{sec3as7} extend the parallel trends assumptions in $2 \times 2$ DID-IV designs (\cite{chasemartin2010-ch}; \cite{Hudson2017-tm}) to multiple period settings. These assumptions require
that, in the absence of exposure to the instrument, the average
evolution of the treatment and the outcome would have followed the
same path between cohort $e$ and the never exposed cohort $\infty$.
\begin{Assumption}[Relevance condition based on a never exposed cohort]
\label{sec3asrelevance}
For each $e \in \mathcal{E}$ and $t\in \{e,\dots,T\}$, $\E[D_{i,t}-D_{i,e-1}|E_i=e]-\E[D_{i,t}-D_{i,e-1}|E_i=\infty] > 0$.
\end{Assumption}
Assumption \ref{sec3asrelevance} is a relevance condition, ensuring that the Wald-DID estimand $w^{DID}_{e,t}$ is well defined.\footnote{If we instead consider the direction of Assumption \ref{sec3as4} that rules out the compliers $CM_{e,t}$, the relevance condition should be $\E[D_{i,t}-D_{i,e-1}|E_i=e]-\E[D_{i,t}-D_{i,e-1}|E_i=\infty] < 0$. The same applies to Assumption \ref{sec3as11}.}\par
For the estimand $w^{DID,ny}_{e,t}$, we instead impose the following
counterparts of Assumptions \ref{sec3as6}--\ref{sec3asrelevance}, replacing
the never exposed cohort with the not-yet-exposed cohorts.\par
\begin{Assumption}[Parallel trends assumption in the treatment based on not-yet-exposed cohorts]
\label{sec3as9}
For each $e\in \mathcal{E}$ and each $(s,t) \in \{2,\dots,T\} \times \{2,\dots,T\}$ such that $e \leq t \leq s < \bar{e}$,
\begin{align*}
\E[D_{i,t}^{\infty}-D_{i,t-1}^{\infty}|E_i=e]=\E[D_{i,t}^{\infty}-D_{i,t-1}^{\infty}|Z_{i,s}=0].
\end{align*}
\end{Assumption}
\begin{Assumption}[Parallel trends assumption in the outcome based on not-yet-exposed cohorts]
\label{sec3as10}
For each $e\in \mathcal{E}$ and each $(s,t) \in \{2,\dots,T\} \times \{2,\dots,T\}$ such that $e \leq t \leq s < \bar{e}$,
\begin{align*}
\E[Y_{i,t}(D_{i,t}^{\infty})-Y_{i,t-1}(D_{i,t-1}^{\infty})|E_i=e]=\E[Y_{i,t}(D_{i,t}^{\infty})-Y_{i,t-1}(D_{i,t-1}^{\infty})|Z_{i,s}=0].
\end{align*}
\end{Assumption}
\begin{Assumption}[Relevance condition based on not-yet-exposed cohorts]
\label{sec3as11}
For each $e\in \mathcal{E}$ and $t\in \{e,\dots,T\}$ such that $t < \bar{e}$, $\E[D_{i,t}-D_{i,e-1}|E_i=e]-\E[D_{i,t}-D_{i,e-1}|Z_{i,t}=0] > 0$.
\end{Assumption}
\begin{Theorem}
    \label{sec3.3.theorem1}
    Suppose Assumptions~\ref{sec3as1}--\ref{sec3as5} hold.
    \begin{itemize}
        \item[(i)] If Assumptions~\ref{sec3as6}--\ref{sec3asrelevance} hold, then for each $e \in \mathcal{E}$ and $t \in \{e,\dots,T\}$, we have $w^{DID}_{e,t}=CLATT_{e,t}$.
        \item[(ii)] If Assumptions~\ref{sec3as9}--\ref{sec3as11} hold, then for each $e \in \mathcal{E}$ and $t \in \{e,\dots,T\}$ such that $t < \bar{e}$, we have $w^{DID,ny}_{e,t}=CLATT_{e,t}$.
    \end{itemize}
    \begin{proof}
        See Appendix.
    \end{proof}
\end{Theorem}
Theorem \ref{sec3.3.theorem1} is the main identification result of this paper. It shows that each Wald-DID estimand captures the corresponding $CLATT_{e,t}$ under the staggered DID-IV identification strategy (Assumptions \ref{sec3as1}--\ref{sec3as5}, and \ref{sec3as6}--\ref{sec3asrelevance} or \ref{sec3as9}--\ref{sec3as11}).\par
We now explain the intuition behind this identification strategy. As we already noted, our Wald-DID estimands scale the DID estimand of the outcome by the DID estimand of the treatment. Accordingly, under staggered DID-IV designs, we rely on the common DID identification strategy in both the numerator and the denominator. First, in the denominator, we use the DID identification strategy to identify the direct effect of the instrument on the treatment: we recover the expectation of the unexposed treatment in cohort $e$ at period $t$, using the no anticipation assumption and the parallel trends assumption in the treatment. Next, in the numerator, we use the DID identification strategy to identify the indirect effect of the instrument on the outcome through the treatment (the exclusion restriction): we recover the expectation of the unexposed outcome in cohort $e$ at period $t$, using the no anticipation assumption and the parallel trends assumption in the outcome. Finally, by scaling the indirect effect by the direct effect, we can capture the $CLATT_{e,t}$ under the monotonicity assumption.\par
\section{Estimation and inference}\label{sec4}
In this section, we propose a credible estimation method in staggered DID-IV designs that is robust to treatment effect heterogeneity. First, we propose an estimator for each CLATT and provide simultaneous inference
procedures based on the multiplier bootstrap. Following \cite{Callaway2021-wl}, we next propose a weighting scheme to construct the summary causal parameters from the CLATTs. Finally, we discuss how to assess the plausibility of the identifying assumptions in practice.\par
Throughout this section, let $\mathcal{I}_{+}=\{(e,t): e\in\mathcal{E},\ t \in \{2,\dots,T\},\ t \geq e\}$ denote the set of all pairs $(e,t)$ for which $CLATT_{e,t}$ is identified when the control group is the never exposed cohort, and let $\mathcal{I}_{+}^{ny}=\{(e,t): e\in\mathcal{E},\ t \in \{2,\dots,T\},\ e \leq t < \bar{e}\}$ denote its counterpart when the control group is the not-yet-exposed cohorts.\par
\begin{Remark}
In practice, when researchers implicitly rely on a staggered DID-IV design, they commonly implement this design via two-way fixed effects instrumental variable (TWFEIV) regressions:
\begin{align*}
    &Y_{i,t}=\phi_{i.}+\lambda_{t.}+\beta_{IV} D_{i,t}+v_{i,t},\\
    &D_{i,t}=\gamma_{i.}+\zeta_{t.}+\pi Z_{i,t}+\eta_{i,t}.
\end{align*}\par
In a companion paper (\cite{Miyaji2023-tw}), however, we show that in more than two periods, the TWFEIV estimand generally fails to aggregate the treatment effects under staggered DID-IV designs: it is a weighted average of all possible CLATTs, but some weights can be negative if the effect of the instrument on the treatment or the outcome is not stable over time. This result motivates the estimation method and the weighting scheme
proposed in this section.\par
\end{Remark}
\subsection{Stacked two stage least squares estimator}\label{sec4.1}
Our estimator for each $CLATT_{e,t}$ is the sample analog of the corresponding Wald-DID estimand:
\begin{align*}
\widehat{CLATT}_{e,t}=\frac{\E_N[Y_{i,t}-Y_{i,e-1}|E_i=e]-\E_N[Y_{i,t}-Y_{i,e-1}|E_i \in U]}{\E_N[D_{i,t}-D_{i,e-1}|E_i=e]-\E_N[D_{i,t}-D_{i,e-1}|E_i \in U]}.
\end{align*}
Here, $\E_N[\cdot|\cdot]$ denotes the sample analog of the conditional expectation, and the control set is $U=\{\infty\}$ for $w^{DID}_{e,t}$ and $U=\{e' \in \mathcal{S}(E_i) : e' > t\}$ for $w^{DID,ny}_{e,t}$. In the following, let $CLATT_{t \geq e}$ and $\widehat{CLATT}_{t \geq e}$ denote the vectors that stack $CLATT_{e,t}$ and $\widehat{CLATT}_{e,t}$ over $(e,t) \in \mathcal{I}_{+}$ when the control group is the never exposed cohort, and over $(e,t) \in \mathcal{I}_{+}^{ny}$ when the control group is the not-yet-exposed cohorts.\par
In practice, $\widehat{CLATT}_{e,t}$ can be computed by a simple regression-based method. The procedure consists of two steps. First, we create the data set for each $CLATT_{e,t}$, which includes the units that are either in cohort $e$ or in the control set $U$, observed in periods $e-1$ and $t$. Second, in each data set, we run the following IV regression:
\begin{align*}
Y_{i,s}=\beta^{e,t}_{0}+\beta^{e,t}_{1}\mathbf{1}\{E_i=e\}+\beta^{e,t}_{2}\mathbf{1}\{s=t\}+\beta^{e,t}_{IV}D_{i,s}+\epsilon^{e,t}_{i,s},\hspace{3mm}s \in \{e-1,t\},
\end{align*}
where the group indicator $\mathbf{1}\{E_i=e\}$ and the post-period
indicator $\mathbf{1}\{s=t\}$ are the included instruments, and the
interaction of the two is the excluded instrument. The resulting IV estimator $\hat{\beta}^{e,t}_{IV}$ coincides with $\widehat{CLATT}_{e,t}$. From the estimation procedure, we call this a stacked two stage least squares (STS) estimator.\footnote{Our STS estimator is related to the DID estimators in staggered DID designs proposed by \cite{Callaway2021-wl} and \cite{Sun2021-rp}. The DID estimator of the treatment and the outcome in our STS estimator corresponds to that of \cite{Sun2021-rp}, and coincides with that of \cite{Callaway2021-wl} in the absence of covariates.}\par
Theorem \ref{sec4.1thm1} below establishes the joint limiting distribution of $\widehat{CLATT}_{t \geq e}$ for each choice of the control group.
\begin{Theorem}
\label{sec4.1thm1}
(i) Suppose Assumptions~\ref{sec3as1}--\ref{sec3asrelevance} hold and $U=\{\infty\}$. Assume also that \(\E[Y_{i,s}^2]<\infty\) and \(\E[D_{i,s}^2]<\infty\) for all $s \in \{1,\dots,T\}$. Then, for each $(e,t) \in \mathcal{I}_{+}$, we have:
\begin{align}
\label{sec4.1thm1_expansion}
\sqrt{N}(\widehat{CLATT}_{e,t}-CLATT_{e,t})=\frac{1}{\sqrt{N}}\sum_{i=1}^{N}\psi_{i,e,t}+o_p(1),
\end{align}
where the influence function $\psi_{i,e,t}$ is given in the Appendix. Furthermore, as $N \rightarrow \infty$, we have:
\begin{align}
\label{sec4.1thm1_clt}
\sqrt{N}(\widehat{CLATT}_{t \geq e}-CLATT_{t \geq e}) \xrightarrow{d} \mathcal{N}(0,\Sigma),
\end{align}
where $\Sigma=\E[\Psi_{i,t \geq e}\Psi'_{i,t \geq e}]$ and $\Psi_{i,t \geq e}$ is the vector that stacks $\psi_{i,e,t}$ over $(e,t) \in \mathcal{I}_{+}$.\par
(ii) Suppose Assumptions~\ref{sec3as1}--\ref{sec3as5} and \ref{sec3as9}--\ref{sec3as11} hold, $U=\{e' \in \mathcal{S}(E_i) : e' > t\}$, and the moment conditions in (i) hold. Then, the same results hold with $\mathcal{I}_{+}$ replaced by $\mathcal{I}_{+}^{ny}$.
\begin{proof}
    See Appendix.
\end{proof}
\end{Theorem}
Based on Theorem \ref{sec4.1thm1}, one can conduct inference on each $CLATT_{e,t}$ separately, using the sample analog of $\Sigma$. The resulting pointwise confidence intervals, however, do not account for multiple testing across the pairs $(e,t) \in \mathcal{I}_{+}$.\par
In practice, we can construct the simultaneous confidence band for the CLATTs by the multiplier bootstrap, as in \cite{Callaway2021-wl}. The algorithm is as follows.\par
\begin{Algorithm}
\label{bootstrap_algorithm}\mbox{}\\
(i) Draw i.i.d.\ multiplier weights $\{V_i\}_{i=1}^{N}$ with zero mean, unit variance, and finite third moment, independent of the original sample $\{(Y_{i,t},D_{i,t},Z_{i,t})_{t=1}^{T}\}_{i=1}^{N}$,\footnote{A popular choice of the multiplier weights is the two-point distribution suggested by \cite{Mammen1993} (see, e.g., \cite{Callaway2021-wl}): $V_i=1-\kappa$ with probability $\kappa/\sqrt{5}$ and $V_i=\kappa$ with probability $1-\kappa/\sqrt{5}$, where $\kappa=(\sqrt{5}+1)/2$.} and compute the bootstrap draw $\widehat{CLATT}^{\ast}_{t \geq e}=\widehat{CLATT}_{t \geq e}+\frac{1}{N}\sum_{i=1}^{N}V_i \hat{\Psi}_{i,t \geq e}$, where $\hat{\Psi}_{i,t \geq e}$ is the sample analog of $\Psi_{i,t \geq e}$ in Theorem \ref{sec4.1thm1}. For each $(e,t)\in\mathcal{I}_+$, denote the corresponding component by $\widehat{CLATT}^{\ast}_{e,t}$, and store $\widehat{R}^{\ast}_{e,t}=\sqrt{N}\bigl(\widehat{CLATT}^{\ast}_{e,t}-\widehat{CLATT}_{e,t}\bigr)$.\\
(ii) Repeat step (i) until $B$ bootstrap draws are obtained. Denote them by $\{\widehat{R}^{\ast(b)}_{e,t}\}_{b=1}^{B}$.\\
(iii) For each $(e,t)\in\mathcal{I}_+$, compute $\widehat{\sigma}_{e,t}=\bigl(q_{0.75}(e,t)-q_{0.25}(e,t)\bigr)/(z_{0.75}-z_{0.25})$, where $q_{\tau}(e,t)$ is the empirical $\tau$-quantile of $\{\widehat{R}^{\ast(b)}_{e,t}\}_{b=1}^{B}$ and $z_{\tau}$ is the $\tau$-quantile of the standard normal distribution.\\
(iv) For each $b=1,\ldots,B$, compute $M^{\ast(b)}=\max_{(e,t)\in\mathcal{I}_+}\bigl|\widehat{R}^{\ast(b)}_{e,t}/\widehat{\sigma}_{e,t}\bigr|$. Let $\widehat{c}_{1-\alpha}$ be the empirical $(1-\alpha)$-quantile of $\{M^{\ast(b)}\}_{b=1}^{B}$.\\
(v) For each $(e,t)\in\mathcal{I}_+$, construct $\widehat{C}(e,t)=\bigl[\widehat{CLATT}_{e,t}-\widehat{c}_{1-\alpha}\widehat{\sigma}_{e,t}/\sqrt{N},\ \widehat{CLATT}_{e,t}+\widehat{c}_{1-\alpha}\widehat{\sigma}_{e,t}/\sqrt{N}\bigr]$.
\end{Algorithm}\par
Here, we focus on the case in which the control group is the never exposed cohort; the same construction applies to the not-yet-exposed control group, with $\mathcal{I}_{+}$ replaced by $\mathcal{I}_{+}^{ny}$.\par
Theorem~\ref{bootstrap_coverage} below establishes the asymptotic coverage of the confidence band constructed by Algorithm~\ref{bootstrap_algorithm}.
\begin{Theorem}
\label{bootstrap_coverage}
Under the assumptions of Theorem \ref{sec4.1thm1}(i), suppose also that \(\E[\psi_{i,e,t}^{2}]>0\) for all \((e,t)\in\mathcal I_+\) and that \(B=B_N\to\infty\). Then, for any \(0<\alpha<1\), as \(N\to\infty\),
\[
\Prb\left(
CLATT_{e,t}\in \widehat C(e,t)
\ \text{for all } (e,t)\in\mathcal I_+
\right)
\to
1-\alpha,
\]
where \(\widehat C(e,t)\) is defined in Algorithm~\ref{bootstrap_algorithm}.
\begin{proof}
    See Appendix.
\end{proof}
\end{Theorem}
\begin{Remark}
Our bootstrap procedure can accommodate cluster-based sampling uncertainty: one only needs to draw the multiplier weights $V_i$ at the cluster level, provided that the number of clusters is large; see also Remark 10 in \cite{Callaway2021-wl}.
\end{Remark}
\begin{Remark}\label{remark:weakIV}
If the instrument is weak, the above normal approximation can be poor even in large samples (\cite{andrews2019weak}). We therefore recommend conducting weak instrument diagnostics in practice. Since each STS regression contains a single treatment and a single excluded instrument, researchers can report the heteroskedasticity-robust first stage $F$ statistic and compare it with the rule-of-thumb threshold of $10$. Researchers can also report the Anderson--Rubin confidence sets, which remain valid regardless of the strength of the instrument (\cite{andrews2019weak}). When reporting these sets across cohorts and periods, one can account for the multiple testing, using the Bonferroni correction.
\end{Remark}\par
\subsection{Weighting scheme}\label{sec4.4}
In staggered DID-IV designs, we can also construct the summary causal parameters from the CLATTs, using the weighting scheme proposed by \cite{Callaway2021-wl}: $\theta^{IV}=\sum_{e}\sum_{t=1}^T w(e,t)\cdot CLATT_{e,t}$, where $w(e,t)$ is the weighting function assigned to each $CLATT_{e,t}$.\par
Following \cite{Callaway2021-wl}, who propose various summary measures in staggered DID designs, we construct their counterparts by weighting each $CLATT_{e,t}$ by the relative size of the effect of the instrument on the treatment in the corresponding period.\footnote{Our weighting scheme based on the first stage effects is motivated by the TWFEIV regression: the TWFEIV estimand also aggregates the underlying causal parameters with weights that reflect the first stage effects, but some of these weights can be negative (\cite{Miyaji2023-tw}). In this sense, our weighting scheme serves as an alternative to the TWFEIV regression when constructing a summary causal parameter. However, we should also emphasize that it is only one natural choice, and the resulting parameter does not necessarily have a direct economic interpretation. We therefore recommend that empirical researchers choose the weights based on the economic question in their application, rather than automatically adopting our weighting scheme.}
For instance, to aggregate the dynamic effects in cohort $e$ over time, we propose the following summary measure $\theta_{sel}^{IV}(e)$, which corresponds to $\theta_{sel}(e)$ in \cite{Callaway2021-wl}:
\begin{align*}
\theta_{sel}^{IV}(e)=\sum_{t=e}^{T}\frac{CAET^{1}_{e,t}}{\sum_{s=e}^{T} CAET^{1}_{e,s}}CLATT_{e,t},
\end{align*}
where $CAET^{1}_{e,t}=\E[D_{i,t}-D_{i,t}^{\infty}|E_i=e]$ is the average effect of the instrument on the treatment in period $t$ in cohort $e$, which we call the cohort specific average exposed effect on the treated in the first stage. If treatment is binary, under the monotonicity assumption (Assumption \ref{sec3as4}), each $CAET^{1}_{e,t}$ equals the share of the compliers in cohort $e$ in period $t$: $CAET^{1}_{e,t}=\Pr(CM_{e,t}|E_i=e)$. The weight assigned to each $CLATT_{e,t}$ thus reflects the relative share of the compliers across the post-exposed periods.\par
By similar arguments, we can specify the weighting functions for the other summary measures; Table \ref{table5} in Appendix \ref{OnlineWeights} summarizes the specific expressions for the weights assigned to each $CLATT_{e,t}$ in each summary causal parameter.\par
In practice, we can estimate $\theta^{IV}$ by the plug-in estimator $\hat{\theta}^{IV}=\sum_{e}\sum_{t=1}^T \hat{w}(e,t)\cdot \widehat{CLATT}_{e,t}$, where $\hat{w}(e,t)$ is the sample analog of $w(e,t)$ and $\widehat{CLATT}_{e,t}$ is the STS estimator in Section \ref{sec4.1}. In Appendix \ref{OnlineWeights}, we derive the influence function of $\hat{\theta}^{IV}$. One can then conduct inference on $\theta^{IV}$, using the sample analog of its asymptotic variance or the multiplier bootstrap in Section \ref{sec4.1}.\footnote{The multiplier bootstrap is attractive when researchers examine several summary causal parameters at once, as it accounts for the multiple testing across the parameters.}\par

\subsection{Assessing the identifying assumptions}\label{sec4.5}
In this section, we discuss how researchers can assess the plausibility of the identifying assumptions in practice.
\subsubsection*{Pre-trends tests}
In most DID applications, researchers assess the plausibility of the parallel trends assumption by examining the evolution of the outcome between treatment and control groups in pre-treatment periods, known as the pre-trends test (e.g., \cite{Roth2023-ig}; \cite{baker2026did}). In staggered DID-IV designs, researchers can conduct the same type of tests for the parallel trends assumptions in the treatment and the outcome.\par
Specifically, under the no anticipation assumption (as well as the no carryover assumption and the exclusion restriction for \eqref{pretrend1outcome}), one can assess the two parallel trends assumptions by testing the following null hypotheses, each of which is equivalent to the corresponding parallel trends condition in pre-exposed periods:
\begin{align}
\E[D_{i,t}-D_{i,t-1}|E_i=e] &{}-\E[D_{i,t}-D_{i,t-1}|E_i \in U]=0 \label{pretrend1treatment}\\
\;\Leftrightarrow\; \E[D_{i,t}^{\infty}-D_{i,t-1}^{\infty}|E_i=e] &{}-\E[D_{i,t}^{\infty}-D_{i,t-1}^{\infty}|E_i \in U]=0, \notag\\
\E[Y_{i,t}-Y_{i,t-1}|E_i=e] &{}-\E[Y_{i,t}-Y_{i,t-1}|E_i \in U]=0 \label{pretrend1outcome}\\
\;\Leftrightarrow\; \E[Y_{i,t}(D_{i,t}^{\infty})-Y_{i,t-1}(D_{i,t-1}^{\infty})|E_i=e] &{}-\E[Y_{i,t}(D_{i,t}^{\infty})-Y_{i,t-1}(D_{i,t-1}^{\infty})|E_i \in U]=0, \notag
\end{align}
for each $e \in \mathcal{E}$ and $t<e$.\par
In practice, one can implement these tests by applying the estimation and inference methods developed in staggered DID designs (\cite{Callaway2021-wl}; \cite{Sun2021-rp}; \cite{Borusyak2021-jv}) to the first stage and the reduced form: $D_{i,t}$ or $Y_{i,t}$ is the outcome and $Z_{i,t}$ is the treatment.\par
Note that while these tests are informative, researchers should also be aware of some caveats. First, the null hypotheses \eqref{pretrend1treatment} and \eqref{pretrend1outcome} are joint implications of the parallel trends assumptions in pre-exposed periods and the no anticipation assumption (as well as the no carryover assumption and the exclusion restriction for \eqref{pretrend1outcome}). A rejection may therefore reflect anticipation, non-parallel trends, or both. Second, the parallel trends assumptions in Section \ref{sec2.3} restrict the post-exposed periods, whereas the null hypotheses above concern the pre-exposed periods. The tests therefore do not directly examine these assumptions (\cite{baker2026did}).\footnote{Another caveat is that the tests may have low power against relevant violations, and that the inference conducted after passing a pre-test can be distorted (\cite{Roth2022}).}\par
These caveats, however, do not overturn the value of the tests: if the treatment and the outcome evolve in parallel across cohorts up to the initial exposure, it is more plausible that they would have continued to do so afterwards. We therefore recommend implementing the pre-trends tests in both the first stage and the reduced form.\par
\subsubsection*{Other identifying assumptions}
The remaining assumptions (the no carryover assumption, the exclusion restriction, the monotonicity assumption, and the no anticipation assumption) are not directly testable, and researchers should justify them with institutional knowledge in each application. Nevertheless, among these, the monotonicity assumption has a testable implication. Since it rules out defiers, the first stage effects are non-negative in the post-exposed periods: $\E[D_{i,t}-D_{i,t}^{\infty}|E_i=e] \geq 0$ for each $e \in \mathcal{E}$ and $t \geq e$. Because these effects are identified by the DID estimands of the treatment under the no anticipation assumption and the parallel trends assumption in the treatment, one can examine this implication in the data.\par
\section{Extensions}\label{sec3}
This section contains four extensions. First, we relax the no anticipation
assumption and allow for
limited anticipation. Second, we consider non-binary, ordered treatments.
Third, we consider repeated cross section data. Finally, we incorporate
discrete covariates into our estimation method. For more details and proofs,
see Online Appendix Section \ref{ApeB}.
\subsubsection*{Limited anticipation}\label{sec3.anticipation}
In some applications, the no anticipation assumption may be implausible, for instance, when the policy underlying the instrument is announced in advance. In Online Appendix Section \ref{ApeB0}, we therefore replace Assumption \ref{sec3as5} with a limited anticipation assumption, in the spirit of \cite{Callaway2021-wl}.
\begin{Assumption}[Limited anticipation in the first stage]
\label{sec4as1}
There is a known $\delta \geq 0$ such that $D_{i,t}^{e}=D_{i,t}^{\infty}$ for all $e\in \mathcal{S}(E_i)\setminus\{\infty\}$ and $t \in \{1,\dots,T\}$ such that $t < e-\delta$.
\end{Assumption}

This assumption allows units to adjust their treatment status up to $\delta$ periods before the initial exposure to the instrument, and nests the no anticipation assumption as the special case $\delta=0$.\par
Under Assumption \ref{sec4as1}, our identification results in Section \ref{sec2.3} carry over with two modifications. First, the pre-exposed period in each Wald-DID estimand shifts from $e-1$ to $e-\delta-1$. Second, the monotonicity assumption, the parallel trends assumptions, and the relevance conditions are extended to cover the anticipation periods $t \geq e-\delta$. Under these modified assumptions, the resulting Wald-DID estimands identify $CLATT_{e,t}$ for $t \geq e-\delta$: up to $t=T-\delta$ when the control group is the never exposed cohort, and for $t < \bar{e}-\delta$ when the control group is the not-yet-exposed cohorts (see Theorem \ref{sec4_anticipation_theorem1} in Online Appendix Section \ref{ApeB0}).\footnote{Our approach maintains point identification by assuming a known anticipation
horizon $\delta$ as in \citet{Callaway2021-wl}. An alternative direction is to remain agnostic about the anticipation horizon and pursue partial identification: in a standard DID setup, \citet{Fenaroli2026} derives sharp bounds on the ATT while jointly relaxing the no anticipation assumption and the parallel trends assumption; see also \citet{Gong2022}. We leave
the extension of these approaches to DID-IV settings for future work.}\par

\subsubsection*{Non-binary, ordered treatment}\label{sec3.nonbinary}
Up to now, we have considered only the case of a binary treatment. However, the same idea can be applied when treatment takes a finite number of ordered values: $D_{i,t} \in \{0,1,\dots,J\}$. In this case, our target parameter is the cohort specific average causal response on the treated (CACRT) in period $t$:
\begin{align*}
CACRT_{e,t} \equiv \sum_{j=1}^{J}w^{e}_{t,j} \cdot \E[Y_{i,t}(j)-Y_{i,t}(j-1)|E_i=e, D_{i,t}^{e} \geq j > D_{i,t}^{\infty}],
\end{align*}
where the weight is $w^{e}_{t,j}=\dfrac{\Prb(D_{i,t}^{e} \geq j > D_{i,t}^{\infty}|E_i=e)}{\sum_{j=1}^{J} \Prb(D_{i,t}^{e} \geq j > D_{i,t}^{\infty}|E_i=e)}$.\par
This parameter is a weighted average of the effect of a one-unit increase in the treatment on the outcome, for those who belong to cohort $e$ and are induced to increase the treatment in period $t$ by the instrument. It is similar to the average causal response (ACR) considered in \cite{Angrist1995-ij}, with the difference that each weight $w^{e}_{t,j}$ and the associated causal parameter are conditional on cohort $E_i=e$.\par
In Online Appendix Section \ref{ApeB1}, we show that the Wald-DID estimands in Section \ref{sec2.3} identify the corresponding $CACRT_{e,t}$ under the same assumptions as in Theorem \ref{sec3.3.theorem1}. Our estimation procedure and weighting scheme in Section \ref{sec4} also remain the same in this setting.\par
\subsubsection*{Repeated cross sections}\label{sec3.rcs}
In some applications, researchers only have access to repeated cross section data, in which each unit is observed only once.\footnote{It also includes the case where researchers use cross section data and exploit a policy shock across cohorts as an instrument for treatment, as in \cite{Duflo2001-nh}.} In Online Appendix Section \ref{ApeB2}, we present the corresponding identifying assumptions and establish the counterparts of Theorems \ref{sec3.3.theorem1} and \ref{sec4.1thm1} in this setting.

\subsubsection*{Saturated STS regression}\label{sec3.saturated}
In empirical applications, researchers often wish to adjust for covariates. In practice, this is typically done by adding the covariates linearly to the TWFEIV regression; however, this covariate-adjusted TWFEIV regression does not generally admit a causal interpretation under DID-IV designs (\cite{Miyaji2023-tw}). Our estimation method can instead accommodate discrete covariates by saturating each STS regression in Section \ref{sec4.1} with the covariate indicators.\footnote{For general (possibly continuous) covariates, see \cite{chen2025efficientdifferenceindifferenceseventstudy}, who extend their semiparametric efficiency results to DID-IV settings, and \cite{raaschoupedersen2026doubly}, who proposes doubly robust and double machine learning estimators in staggered DID-IV designs.}\par
Specifically, let $X_i$ denote discrete covariates with finite support $\mathcal X$. Here, we assume that the covariates are time-invariant and not affected by either the instrument or the treatment. Then, in each data set for $CLATT_{e,t}$ in Section \ref{sec4.1}, we can run the following saturated IV regression:
\small
\begin{align*}
Y_{i,s}
  = \beta^{e,t} D_{i,s}
  + \sum_{x\in\mathcal X}\1\{X_i=x\}
    \left(\alpha_{x}+\gamma_{x}\1\{E_i=e\}+\lambda_{x}\1\{s=t\}\right)
  + u_{i,s},\hspace{3mm}s \in \{e-1,t\},
\end{align*}
\normalsize
where the covariate indicators and their interactions with $\1\{E_i=e\}$ and $\1\{s=t\}$ are the included instruments, and the interaction $\1\{E_i=e\}\1\{s=t\}$ is the excluded instrument.\par 
Theorem \ref{theorem8_saturate} below shows that the saturated IV coefficient $\beta^{e,t}$ aggregates the covariate-specific Wald-DID estimands.
\begin{Theorem}[Saturated STS regression]
\label{theorem8_saturate}
Let $S_i=\1\{E_i \in \{e\} \cup U\}$ indicate the units in the corresponding $2 \times 2$ subsample. Let $R_{i,s}$ be the residual from the linear projection of $\1\{E_i=e\}\1\{s=t\}$ on the included instruments in this subsample, and let $\omega_x=\E[R_{i,s}^{2}\mid X_i=x, S_i=1]$. For each $x \in \mathcal{X}$, let $\beta(x)=\Delta_Y(x)/\Delta_D(x)$, where $\Delta_Y(x) = \E[Y_{i,t}-Y_{i,e-1}\mid E_i=e, X_i=x]-\E[Y_{i,t}-Y_{i,e-1}\mid E_i \in U, X_i=x]$ and $\Delta_D(x)$ is defined analogously for the treatment. Assume $\Prb(X_i=x \mid S_i=1) > 0$ and $\omega_{x} > 0$ for all $x \in \mathcal{X}$. Assume also
$\Delta_D(x)\neq 0$ for all $x\in\mathcal X$ and $\sum_{x\in\mathcal X}\Prb(X_i=x \mid S_i=1)\,\omega_x\Delta_D(x) \neq 0$. Then, we have:
  \begin{align*}
      \beta^{e,t}=\frac{\sum_{x\in\mathcal X}\Prb(X_i=x \mid S_i=1)\,\omega_x \Delta_D(x)\,\beta(x)}
       {\sum_{x\in\mathcal X}\Prb(X_i=x \mid S_i=1)\,\omega_x\Delta_D(x)}.
  \end{align*}
  \begin{proof}
      See Online Appendix Section \ref{OnlineApeSat}.
  \end{proof}
\end{Theorem}\par
Note that, under the conditional counterparts of the identifying assumptions in Section \ref{sec2.3}, each $\beta(x)$ captures the corresponding $CLATT_{e,t}$ conditional on $X_i=x$, by the same argument as in the proof of Theorem \ref{sec3.3.theorem1}.\footnote{For a formal treatment of DID-IV designs with covariates, see Appendix B of \cite{chen2025efficientdifferenceindifferenceseventstudy}.} The saturated IV coefficient $\beta^{e,t}$ admits a causal interpretation if, in addition, the direction of the monotonicity is common across all $x \in \mathcal{X}$: under this condition, the weight assigned to each $\beta(x)$ is guaranteed to be non-negative.\footnote{An alternative approach is to use the interactions of the instrument with the covariate indicators as excluded instruments, as in \cite{Angrist1995-ij}. This ``saturate and weight'' approach does not require the direction of the monotonicity to be common across all $x \in \mathcal{X}$, but may suffer from many instruments bias (\cite{Blandhol2022-lk}).}\par
In Online Appendix Section \ref{OnlineApeSat}, we establish the joint limiting distribution of the saturated STS estimators, so that one can construct the simultaneous confidence band as in Algorithm~\ref{bootstrap_algorithm}.\par

\section{Application}\label{sec5}
We illustrate the empirical relevance of our findings in the setting of \cite{Oreopoulos2006-bn}, estimating returns to schooling in the United Kingdom. We first assess the identifying assumptions in staggered DID-IV designs implicitly imposed by \cite{Oreopoulos2006-bn}. We then estimate the TWFEIV regression in the author's setting, and apply our proposed method and weighting scheme to estimate the target parameter and its summary measure. Next, we estimate the saturated STS regressions with age-bin indicators to account for the difference in age composition. Finally, we conduct weak instrument diagnostics.\par
\subsection{Setting and data}\label{sec5.1}
\cite{Oreopoulos2006-bn} estimates returns to schooling using a major education reform in the UK that raised the minimum school-leaving age from 14 to 15. Specifically, \cite{Oreopoulos2006-bn} exploits variation resulting from the different timing of implementation of school reforms between Britain (England, Scotland, and Wales) and Northern Ireland as an instrument for education attainment: the school-leaving age increased in Britain in 1947, but the increase was not implemented until 1957 in Northern Ireland.

The data we use come from the full micro dataset made available with
\cite{Oreopoulos2008} (\cite{Oreopoulos2019-data}), which combines the British General Household Surveys and the Northern Ireland Continuous Household Surveys.\footnote{\label{fn:survey_coverage}\cite{Oreopoulos2006-bn} was followed by a corrigendum, \cite{Oreopoulos2008}, which corrects the UK results and distributes revised code and an expanded UK data set. \cite{Oreopoulos2008} describes the survey coverage as 1984--2006 for Britain and 1985--2006 for Northern Ireland. After excluding observations with missing earnings or education attainment, however, we retain waves 1984--1996, 1998, and 2000--2006 for Britain, and 1986--1991, 1996, 1998, and 1999 for Northern Ireland.}
The sample consists of individuals in Britain and Northern Ireland who were aged
$28$--$64$ at the time of the survey and reached age $14$ between $1936$ and
$1965$.\footnote{
While \cite{Oreopoulos2006-bn} uses ages $32$--$64$, we
follow \cite{Oreopoulos2008}, which uses ages $28$--$64$. We exclude the cohort aged $14$ in 1935 because
earnings are missing for that cohort in Northern Ireland.
} In \cite{Oreopoulos2006-bn}, a cohort (a year when aged $14$) plays the role of time, as it determines the exposure to the policy change. The data thus form repeated cross sections across cohorts (see Section~\ref{sec3}).\par
While \cite{Oreopoulos2006-bn} first aggregates the data into cell means and runs regressions weighted by cell size, we use the full micro dataset directly. We do so for two reasons. First, this choice does not affect the point estimates.\footnote{When the cells are defined by all the variables entering the regression, the weighted cell-mean regression yields the same point estimates as the unweighted regression on the full micro sample.} Second, the micro data allow us to compute standard errors from individual-level observations.

In Section III.C\footnote{In the first part of his analysis, \cite{Oreopoulos2006-bn} adopts regression discontinuity designs and analyzes Britain and Northern Ireland separately; due to the imprecision of the standard errors, he then moves to ``a difference-in-differences and instrumental-variables analysis by combining the two sets of U.K. data''. In Section II.C, \cite{Oreopoulos2006-bn} also runs an IV regression with a Northern Ireland fixed effect, but includes a quartic polynomial in birth cohort rather than birth cohort fixed effects. We therefore focus on the specification in Section III.C.}, \cite{Oreopoulos2006-bn} (see also Table 4 in
\cite{Oreopoulos2008}) estimates the following two-way fixed effects instrumental variable regression with the education reform as an excluded instrument for education attainment:
 \begin{align*}
&Y_{i}=\mu_{NI.}+\delta_{.t}+\beta_{IV} D_{i}+X_{i}'\theta_Y+\epsilon_{i},\\
&D_{i}=\gamma_{NI.}+\zeta_{.t}+\pi Z_{i}+X_{i}'\theta_D+\eta_{i}.
\end{align*}
 Here, the dependent variables $Y_{i}$ and $D_{i}$ are log annual earnings and education attainment (age left full-time education) for unit $i$, respectively. Both equations include a birth cohort fixed effect ($\delta_{.t}$ and $\zeta_{.t}$) and a Northern Ireland fixed effect ($\mu_{NI.}$ and $\gamma_{NI.}$). The binary instrument $Z_{i} \in \{0,1\}$ equals one if unit $i$ is exposed to the policy change. The covariates $X_{i}$ include survey year fixed effects and a quartic in age.\footnote{Although \cite{Oreopoulos2006-bn} lists a gender dummy among the controls, the replication code distributed with \cite{Oreopoulos2008} does not include it. We therefore treat survey year fixed effects and a quartic in age as controls.}\par
 The staggered introduction of the school reform can be viewed as a natural experiment, but is not randomized across regions in reality; \cite{Oreopoulos2006-bn} notes that the reform was implemented with political support, taking into account costs and benefits. This indicates that \cite{Oreopoulos2006-bn} implicitly relies on a staggered DID-IV identification strategy instead of exploiting the random variation of the policy change. Indeed, \cite{Oreopoulos2006-bn} presents plots of British and Northern Irish average education attainment and average log earnings by cohort, illustrating the evolution of these variables before and after the policy shock.\par

\subsection{Assessing the identifying assumptions in \cite{Oreopoulos2006-bn}}\label{sec5.2}
We first discuss the validity of the staggered DID-IV identification strategy in \cite{Oreopoulos2006-bn}. In the author's setting, our target parameter is the cohort specific average causal response on the treated (CACRT), as education attainment is a non-binary, ordered treatment. This parameter captures the effect of an additional year of schooling on log annual earnings among the compliers, that is, the units who increased their schooling in response to the reform.
 \subparagraph{Exclusion restriction.}
This assumption requires that the school reform affected log annual earnings only through education attainment. It may be violated if the reform also changed school quality, credentials, or other school inputs in ways that affected later earnings conditional on years of schooling. In our context, however, \cite{Devereux} and \cite{CLARK2023104902} document that the reform was accompanied by additional school resources, so that pupil--teacher ratios and class sizes remained broadly stable; \cite{Devereux} also note that few affected students stayed in school until age 16 to obtain credentials. These facts support the plausibility of this assumption.\par

 \subparagraph{Monotonicity assumption.}
This assumption requires that there are no units who reduce their schooling in response to the school reform. It is automatically satisfied in the author's setting, as the school reform raised the minimum school-leaving age from $14$ to $15$. Indeed, consistent with its testable implication in Section \ref{sec4.5}, the estimated first stage effects are all positive in the post-exposed periods (Panel (a) of Figure \ref{sec5figure3}).
 \subparagraph{No anticipation in the first stage.}
This assumption requires that education attainment was not affected by future exposure to the school reform before its implementation. In the setting of \cite{Oreopoulos2006-bn}, anticipation is a potential concern because the increase in the school-leaving age was legislated by the 1944 Education Act but was not implemented until 1947. In Appendix~\ref{AppendixC}, we therefore conduct a robustness check, replacing this assumption with the limited anticipation assumption in Section~\ref{sec3}.

\vskip\baselineskip
Next, we assess the validity of the parallel trends assumptions in the treatment and the outcome, using the DID estimator of \cite{Callaway2021-wl}. Here, the control group $U$ is Northern Ireland and the reference period is $t=1946$. For inference, we apply the multiplier bootstrap procedure of \cite{Callaway2021-wl} to construct $95\%$ simultaneous confidence bands. The results are shown in Figure \ref{sec5figure3}.\par

 \begin{figure}[t]
    \centering
    \includegraphics[width=\columnwidth]{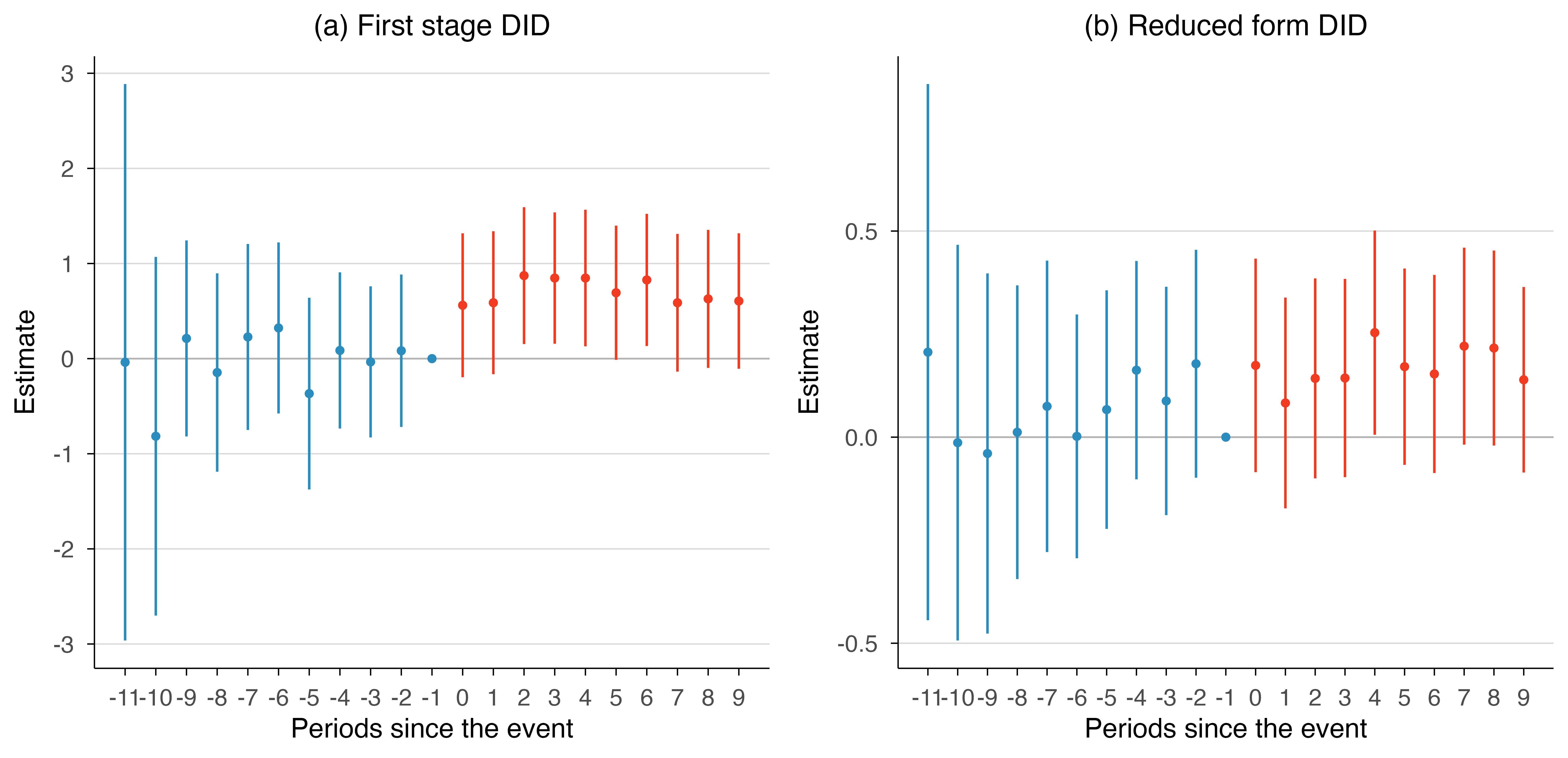}
  \caption{The effect of the instrument in the first stage and reduced form in the setting of \cite{Oreopoulos2006-bn}.
\textit{Notes}: This figure presents the estimated effect of the school reform on education attainment (Panel (a)) and on log annual earnings (Panel (b)) under the staggered DID-IV identification strategy. The control group $U$ is Northern Ireland, and the reference period is $t=1946$. The vertical bars represent 95\% simultaneous confidence bands from the multiplier bootstrap of \cite{Callaway2021-wl}, computed separately for Panels (a) and (b). Blue lines represent estimates for pre-exposed periods. Under the null hypothesis that the parallel trends assumptions in the treatment and outcome hold, these estimates should be zero. Red lines represent estimates for post-exposed periods.}
\label{sec5figure3}
\end{figure}

\subparagraph{Parallel trends assumptions in the treatment and the outcome.}
These assumptions require that the expectations of education attainment and log annual earnings would have followed the same path between Britain and Northern Ireland across cohorts in the absence of the school reform. In Figure \ref{sec5figure3}, the pre-exposed estimates in the first stage (Panel (a)) are not statistically different from zero. While those in the reduced form (Panel (b)) exhibit an upward pattern, they are also not statistically different from zero. These results support the plausibility of these assumptions.\par

\subsection{Illustrating our estimation method}\label{sec5.3}
We start by estimating the TWFEIV regression in the author's setting. Here, we do not apply the author's weights to our regression because we use the full micro dataset. We also exclude the covariates to clearly illustrate the pitfall of the TWFEIV regression. We return to the role of covariates in the next subsection.\par
The result is shown in Table \ref{tab:summary_returns}. The TWFEIV estimate is $-0.009$ and is not statistically significant at the $5\%$ level. Taken at face value, this result suggests that the return to schooling in the UK is nearly zero.\footnote{\cite{young2024nearly} revisits \cite{Oreopoulos2006-bn} and shows that the TSLS estimates can be numerically unstable due to near collinearity among the included instruments, such as the polynomials in age and birth cohort. Following \cite{young2024nearly}, we conduct sensitivity checks based on variable and observation ordering, and find that our TWFEIV estimate does not suffer from such numerical instability.} However, this may be a misleading conclusion. We cannot interpret the TWFEIV estimand as capturing a properly weighted average of all possible cohort specific returns to schooling in Britain if the effect of the school reform on education attainment or log annual earnings is not stable across cohorts (\cite{Miyaji2023-tw}).\par
Indeed, the TWFEIV regression performs ``bad comparisons'' (cf.\ \cite{Goodman-Bacon2021-ej}) in the author's setting: when we run the TWFEIV regression, we treat the set of already exposed units in Northern Ireland as controls between $1957$ and $1965$, as both regions are already exposed to the policy shock during these periods. In Appendix \ref{ApeD}, we quantify the bias terms of the TWFEIV estimand arising from the bad comparisons, and show that it is negatively biased in \cite{Oreopoulos2006-bn}: all the bias terms are positive, while all the assigned weights are negative (see Figure \ref{ApeDfigure1}).
\begin{table}[t]
\centering
\footnotesize
\caption{Returns to schooling in \cite{Oreopoulos2006-bn}}
\renewcommand{\arraystretch}{1.3}

\begin{tabular*}{0.98\textwidth}{@{\extracolsep{\fill}}lccc@{}}
\hline
& TWFEIV & $\theta_{sel}^{IV}(e)$ & $\theta_{sel}^{IV,age}(e)$ \\
\hline
Estimate & -0.009 & 0.240 & 0.164 \\
SE & 0.041 & 0.099 & 0.090 \\
95\% CI & [-0.090, 0.072] & [0.047, 0.434] & [-0.012, 0.340] \\
First stage $F$ & 68.6 & 11.2 & 13.6 \\
$F>10$ & \checkmark & \checkmark & \checkmark \\
AR 95\% CI & [-0.099, 0.068] & [0.066, 0.568] & [-0.021, 0.399] \\
Sample size & 82,790 & 29,077 & 18,711 \\
\hline
\end{tabular*}

\vspace{0.3em}

\begin{minipage}{0.98\textwidth}
\footnotesize
\textit{Notes}: This table reports the TWFEIV estimate and estimates of the summary causal parameters $\theta_{sel}^{IV}(e)$ and $\theta_{sel}^{IV,age}(e)$. For the TWFEIV estimate, the SE is the heteroskedasticity-robust standard error, and the first stage $F$ statistic is the heteroskedasticity-robust first stage $F$ statistic for the excluded instrument. For the summary parameters, the SEs, the first stage $F$ statistics, and the AR confidence sets are computed from the corresponding influence functions, as described in Section~\ref{sec5.5}. 95\% CIs are Wald intervals based on these standard errors. The row $F>10$ reports whether the statistic exceeds the rule-of-thumb threshold of 10. The AR 95\% CI row reports the Anderson--Rubin confidence sets. Sample size reports the number of observations used for each estimator.
\end{minipage}

\label{tab:summary_returns}
\end{table}

\begin{figure}[t]
    \centering
    \includegraphics[width=\columnwidth]{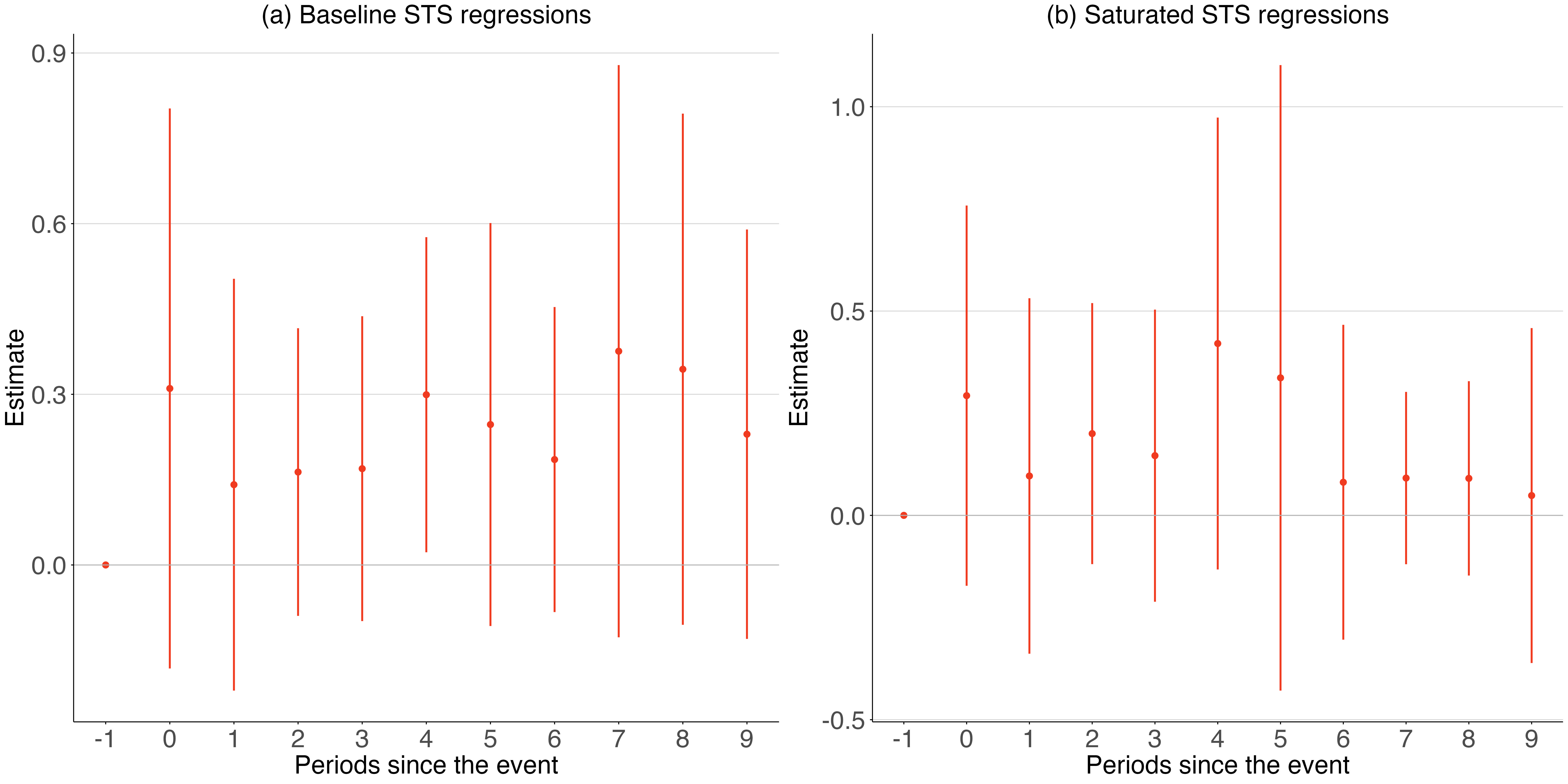}
    \caption{The cohort specific returns to schooling in \cite{Oreopoulos2006-bn}. \textit{Notes}: This figure shows the results for returns to schooling under the staggered DID-IV identification strategy. Panel (a) reports the estimates from the baseline STS regressions, and Panel (b) reports the estimates from the saturated STS regressions with five-year age-bin indicators. The control group $U$ is Northern Ireland and the reference period is $t=1946$. The red lines represent the STS estimates with $95\%$ simultaneous confidence bands constructed by the multiplier bootstrap for post-exposed periods.}
    \label{sec5figure4}
\end{figure}
We now employ the proposed method to estimate each cohort specific return to schooling in Britain (CACRT). First, we create data sets by cohort $t$ ($t=1947,\dots,1956$). Each data set contains only units of cohort $t$ and cohort $1946$ in Britain and Northern Ireland. We define Northern Ireland as the control group $U$ and cohort $1946$ as the reference period. Here, we discard the data for cohorts between $1957$ and $1965$ to avoid the bad comparisons. We then run the following STS regression in each data set:
\begin{align*}
Y_i
  = \alpha+\gamma G_i+\lambda T_{i}
     +\beta^{IV}D_i+u_{i},
\end{align*}
where the region indicator $G_i \in \{0,1\}$ equals one for Britain, the cohort indicator $T_{i} \in \{0,1\}$ equals one for cohort $t$ and zero for cohort $1946$, and their interaction $Z_{i}=G_iT_{i}$ is the excluded instrument. We construct a $95\%$ simultaneous confidence band for the STS estimators over the post-exposed periods, using the multiplier bootstrap described in Section~\ref{sec4.1}. The underlying influence function is given in equation~\eqref{ApeAinf_repeated} of Online Appendix Section~\ref{ApeB2}.\par
We also estimate the summary causal measure, the weighted average of these returns over the post-reform cohorts in Britain ($e=1947$):
\begin{align*}
\theta_{sel}^{IV}(e)=\sum_{t=1947}^{1956}\frac{CAET^{1}_{e,t}}{\sum_{s=1947}^{1956}CAET^{1}_{e,s}}CACRT_{e,t}.
\end{align*}
Here, each weight represents the relative size of the effect of the reform on education attainment in the corresponding period. We compute the standard error of the estimator, using its influence function (see Appendix \ref{OnlineWeights}).\par
The results are shown in Panel (a) of Figure~\ref{sec5figure4} and Table \ref{tab:summary_returns}. The estimates for each cohort specific return to schooling range from $14\%$ to $38\%$ with wide confidence intervals, and $1$ out of $10$ estimates is statistically significant at the $5\%$ level. The summary estimate is $0.24$ and is also statistically significant. These results indicate large returns to schooling in the UK. Notably, the estimates obtained from our proposed method and weighting scheme are significantly different from the TWFEIV estimate: our STS estimates and their summary measure are all positive, whereas the TWFEIV estimate is strictly negative.\par

\subsection{Saturated STS regression}\label{sec5.4}
While the previous subsection points to sizable returns to schooling in the UK, the summary estimate of $0.24$ is large relative to conventional estimates, which are typically around $0.06$--$0.15$ in developed countries (\cite{GUNDERSON202039}). One possible explanation for this large estimate is the difference in age composition between Britain and Northern Ireland: later cohorts are observed at older ages in Britain than in Northern Ireland.\footnote{See footnote $13$ in \cite{Oreopoulos2006-bn}. The difference arises from differences in the retained survey waves after excluding observations with missing earnings or education attainment; see footnote~\ref{fn:survey_coverage} in this paper.} Because earnings generally rise over the life cycle, the Britain--Northern Ireland earnings gap can then grow from cohort to cohort even in the absence of the reform, thereby inflating our estimates. Indeed, \cite{Oreopoulos2008} adds age controls (a quartic in age) to the TWFEIV regression and reports a return to schooling of $0.108$.\footnote{The corresponding estimate in \citet{Oreopoulos2006-bn} is $0.158$.}\par

Motivated by this observation, we next estimate the saturated STS regression with age-bin indicators as discrete covariates (see Section \ref{sec3}). Here, we use the saturated STS regression rather than the covariate-adjusted TWFEIV regression as in \cite{Oreopoulos2006-bn}, because the latter does not generally permit a causal interpretation under DID-IV designs (\cite{Miyaji2023-tw}).\par
Specifically, we run the following IV regression in each two-period/two-group data set and obtain the IV estimator \(\widehat{\beta}^{IV}\):
\begin{align}
\label{section6_saturation_1}
Y_i
  &= \beta^{IV} D_i
  + \sum_{x\in\mathcal X}\1\{X_i=x\}
    \left(\alpha_{x}+\gamma_{x}G_i+\lambda_{x}T_i\right)
  + u_{i},\\
\label{section6_saturation_2}
  D_i
  &=\pi Z_i+\sum_{x\in\mathcal X}\1\{X_i=x\}
    \left(\mu_{x}+\eta_{x}G_i+\kappa_{x}T_i\right)
  + v_{i}.
\end{align}
Here, \(X_i\in\mathcal X\) denotes the five-year age bin to which individual \(i\)'s age belongs, and \(\mathcal X\) is the
set of age bins in the estimation sample.\footnote{Age bins are defined using conventional five-year age groups (25--29, 30--34, \ldots, 60--64); because the sample is restricted to ages 28--64, the first bin contains only ages 28--29. We use five-year rather than single-year age bins because the latter require overlap at each exact age and therefore discard many observations.} In each data set, we restrict the sample to age bins
with observations in all four \((G_i,T_i)\) cells.\footnote{We do not include survey year fixed effects as in \cite{Oreopoulos2006-bn} for two reasons. First, the survey year is determined by the birth cohort and the age at survey; once we condition on the cohort and the age bins, the survey year fixed effects therefore add little independent variation. Second, including the survey year fixed effects can substantially reduce the sample size due to the common support condition across the four cells.} From Theorem \ref{theorem8_saturate}, the IV coefficient $\beta^{IV}$ can be interpreted as a weighted average of the age-bin-specific returns to schooling, provided that the identifying assumptions hold within each age bin.\footnote{Theorem \ref{theorem8_saturate} also requires the direction of the monotonicity to be common across the age bins. This condition is satisfied in our setting, as the monotonicity assumption holds unconditionally (see Section \ref{sec5.2}).}\par
We also estimate the corresponding summary causal parameter, aggregating the age-adjusted returns to schooling ($CACRT^{age}_{e,t}$) over the post-reform cohorts:
\begin{align*}
\theta_{sel}^{IV,age}(e)=\sum_{t=1947}^{1956}\frac{CAET^{1,age}_{e,t}}{\sum_{s=1947}^{1956}CAET^{1,age}_{e,s}}CACRT^{age}_{e,t}.
\end{align*}
Here, we estimate $CACRT^{age}_{e,t}$ by the saturated STS regression (the IV coefficient $\beta^{IV}$ in \eqref{section6_saturation_1}), and estimate $CAET^{1,age}_{e,t}$ by the corresponding first stage regression (the first stage coefficient $\pi$ in \eqref{section6_saturation_2}) in each two-period/two-group data set. We again compute the standard error of the estimator, using its influence function.\par

We report the results in Panel (b) of Figure~\ref{sec5figure4} and Table~\ref{tab:summary_returns}. The estimates for each age-adjusted cohort specific return to schooling range from $5\%$ to $42\%$ with wide confidence intervals, and none of the $10$ estimates is statistically significant at the $5\%$ level. The estimate of $\theta_{sel}^{IV,age}(e)$ is $0.164$ and is statistically significant at the $10\%$ level, but not at the $5\%$ level. Compared with the baseline estimate of $0.24$, the age-adjusted summary estimate is smaller and more in line with the conventional estimates of the return to schooling. This pattern is consistent with the potential upward bias induced by the difference in age composition.\par

\subsection{Weak IV diagnostics}\label{sec5.5}
Finally, we conduct weak instrument diagnostics and check whether our results are driven by weak identification. Following Remark~\ref{remark:weakIV} in Section~\ref{sec4.1}, we report the heteroskedasticity-robust first stage $F$ statistic for the excluded instrument and the Anderson--Rubin (AR) confidence set for each estimate in Table~\ref{tab:summary_returns}. Here, we derive the AR confidence set for each estimate analytically by inverting the AR test (\cite{andrews2019weak}).\footnote{For the summary causal parameter $\theta_{sel}^{IV}(e)$, the $F$ statistic is the robust $t$ statistic squared for the aggregated first stage effect $\sum_{t}\widehat{CAET}^{1}_{e,t}$, and the AR statistic for $H_0:\theta=\beta_0$ is the robust $t$ statistic squared for $\sum_{t}(\widehat{RF}_{e,t}-\beta_0\,\widehat{CAET}^{1}_{e,t})$, where $\widehat{RF}_{e,t}$ is the reduced form DID estimate. The same applies to $\theta_{sel}^{IV,age}(e)$, with the first stage DID estimate in each cohort being the coefficient $\pi$ in \eqref{section6_saturation_2} and the reduced form DID estimate the coefficient on $Z_i$ in the regression that replaces $D_i$ in \eqref{section6_saturation_2} with $Y_i$. The variances are computed from the corresponding influence functions.} The corresponding diagnostics for the cohort specific estimates are reported in Table~\ref{tab:period_specific_cacrt_panels} in Online Appendix Section~\ref{OnlineAppTables}, where we adjust the AR confidence sets for multiple testing over the post-reform cohorts, using the Bonferroni method.\footnote{Specifically, we invert each period specific AR test at level $\alpha/K$, where $K=10$.}\par
The results suggest that weak instrument concerns are not severe enough to overturn our conclusions. While some of the first stage $F$ statistics fall below the rule-of-thumb threshold of $10$ in the period specific STS regressions, the statistic is $68.6$ for the TWFEIV regression, $11.2$ for $\theta_{sel}^{IV}(e)$, and $13.6$ for $\theta_{sel}^{IV,age}(e)$, all exceeding the threshold. Similarly, while the Bonferroni-adjusted AR confidence sets are often wide or unbounded in the period specific STS regressions, the AR $95\%$ confidence intervals for the summary parameters lead to the same conclusions as the Wald intervals in Table~\ref{tab:summary_returns}: zero is excluded for $\theta_{sel}^{IV}(e)$, and included for the TWFEIV estimate and $\theta_{sel}^{IV,age}(e)$ (the AR $90\%$ confidence interval for the latter excludes zero).\footnote{The AR 90\% confidence interval is $[0.014,0.346]$.}\par
Overall, the results in this section suggest that the return to schooling is substantial in the UK and that the choice of estimation method matters in staggered DID-IV designs in practice.

\section{Comparing DID-IV with Fuzzy DID}\label{sec6}
Before concluding, we compare DID-IV with Fuzzy DID considered in \cite{De_Chaisemartin2018-xe} (dCDH). Our staggered DID-IV designs extend the $2 \times 2$ DID-IV designs studied in \cite{chasemartin2010-ch} and \cite{Hudson2017-tm}, while dCDH propose Fuzzy DID as an alternative framework for the same $2 \times 2$ setting.\footnote{In dCDH, the group variable $G$ plays the role of the instrument $Z$ in this paper, as it enters their treatment participation equation, $D = \mathbf{1}\{V \geq v_{GT}\}$ (see Assumption~3 in dCDH). We therefore refer to the groups $G=1$ and $G=0$ as the exposed and unexposed groups, respectively (see Online Appendix Section~\ref{ApeC.2}).} Both frameworks consider the same Wald-DID estimand, but differ in their target parameters and identifying assumptions.\par
In this section, we compare the two designs along three dimensions: (i) the restrictions on treatment adoption behavior, (ii) the interpretation of the target parameter, and (iii) the use of the Wald-DID estimand. These distinctions matter in practice because the first helps assess the credibility of each design, the second clarifies the causal question each design answers, and the third guides the choice of estimand. At the end of this section, we illustrate these considerations in the setting of \cite{Oreopoulos2006-bn}.\footnote{Although \cite{Oreopoulos2006-bn} considers multiple cohorts, focusing on one pre-reform and one post-reform cohort yields a $2 \times 2$ DID-IV setting.} Throughout, we consider the two-period, two-group setting, as in dCDH, in which our target parameter reduces to the LATET (see \cite{chasemartin2010-ch}). The detailed discussion and proofs underlying this section are provided in Online Appendix Section~\ref{ApeC}.\par
\subsubsection*{Restrictions on treatment adoption behavior}
In DID-IV designs, the only restriction on treatment adoption behavior is the monotonicity assumption (Assumption~\ref{sec3as4}). By contrast, the identifying assumptions in Fuzzy DID jointly restrict adoption behavior both in response to the instrument and over time: the treatment rate must strictly increase in the exposed group and remain stable in the unexposed group, while individual treatment status must evolve monotonically over time within each group. These assumptions therefore imply the monotonicity assumption in DID-IV while imposing additional restrictions, some of which are asymmetric between the exposed and unexposed groups. In Online Appendix Section~\ref{ApeC.3}, we make this comparison precise by extending the compliance types in Section~\ref{sec2.1} (e.g., compliers and defiers) to treatment adoption over time and characterizing the latent adoption types that each design rules out (Tables~\ref{ApeC.3.2.table1}--\ref{ApeC.3.3.table2}).\par

\subsubsection*{The interpretation of the target parameter}
In DID-IV designs, our target parameter is the LATET, which measures the
treatment effects for the compliers, a population defined by its
counterfactual response to the instrument in an exposed group. By contrast,
the target parameter in Fuzzy DID is the switcher local average treatment
effect on the treated (SLATET), which measures the treatment effects for the
switchers, a population defined by its transition from non-treatment to
treatment over time in an exposed group.\par 
Because the compliers and the switchers are defined by different sources of treatment variation, the two parameters need not measure the treatment effects for the same subpopulation, and answer different causal questions. In Online Appendix Section~\ref{ApeC.4}, we make this distinction precise by decomposing the SLATET under the identifying assumptions in Fuzzy DID. We show that, depending on the directions of dCDH's two monotonicity conditions (on treatment adoption in response to the instrument and over time), the switcher population can consist of one or both of two types: those who would switch into treatment over time even without the instrument and those induced into treatment by the instrument (see Theorem~\ref{ApeC.4.Theorem1}). The latter form a subset of the compliers underlying the LATET.\par
Our decomposition result has two implications for empirical practice. First, when researchers adopt Fuzzy DID and wish to know which latent types' causal effects the SLATET identifies, they need to specify the directions of dCDH's two monotonicity conditions. Second, researchers can choose between the LATET and the SLATET as their target parameter, depending on which type of treatment variation (instrument-induced, time-induced, or both) is relevant to their research question (\cite{Heckman2001-ur}).\footnote{When only the instrument-induced variation is relevant, the LATET is the natural target: in this case, the SLATET identifies the treatment effects for a subset of the compliers underlying the LATET, under stronger restrictions on treatment adoption behavior.} 

\subsubsection*{The use of the Wald-DID estimand}
In DID-IV designs, we adopt the Wald-DID estimand as the natural estimand for
identifying the LATET. By contrast, dCDH point out that in Fuzzy DID, the
Wald-DID estimand requires the stable treatment effect assumption to identify
the SLATET, and propose alternative estimands that avoid this restriction.\par 
In
Online Appendix Section~\ref{ApeC.5}, we show that this requirement does not
arise from the Wald-DID estimand itself, but from the parallel trends assumption
on untreated outcomes imposed in Fuzzy DID. Under this assumption, identifying
the SLATET with the Wald-DID estimand requires imputing the evolution of
untreated potential outcomes for units treated in both periods. Because the
untreated outcomes of these units are unobserved, this imputation requires their
treatment effects to remain stable over time. By contrast, DID-IV designs do not require the stable treatment effect assumption, because they impose the parallel trends assumption on unexposed outcomes and thus involve no such imputation; the Wald-DID estimand therefore identifies the LATET without restricting the evolution of treatment effects.\par
Whether the Wald-DID estimand is the natural estimand thus depends on
which parallel trends assumption researchers have in mind. In Online
Appendix Section~\ref{ApeC.5.2}, we discuss two advantages of the
parallel trends assumption on unexposed outcomes in DID-IV settings.
First, it aligns with the source of identifying variation, namely
exposure to the instrument rather than treatment adoption. Second,
its plausibility can be assessed indirectly using data from
pre-exposed periods (Section~\ref{sec4.5}).\par

\subsubsection*{Implications for empirical practice}
In sum, choosing between DID-IV and Fuzzy DID involves three considerations: whether the restrictions on treatment adoption behavior are credible in the setting at hand, which target population---the compliers or the switchers---is of substantive interest, and which parallel trends assumption can be defended in the application.\par
We illustrate these considerations in the setting of
\cite{Oreopoulos2006-bn}. In this application, our target parameter measures the returns
to schooling for the compliers, the students in the post-reform cohort whose education attainment increased in
response to the reform. By contrast, the target parameter in dCDH measures the returns to schooling
for the switchers, the students in the post-reform cohort whose schooling
was higher than it would have been had they belonged to the pre-reform
cohort. Along this dimension, researchers can choose between the two designs based
on which of these populations is of substantive interest. The choice,
however, also depends on the other two considerations.\par
The remaining two considerations suggest that DID-IV is well suited to this application. First, regarding the
restrictions on treatment adoption behavior, the monotonicity assumption in
DID-IV appears plausible, as discussed in Section~\ref{sec5.2}. By contrast, when treatment is
ordered, as with education attainment, the restrictions in
Fuzzy DID require the distribution of education attainment to remain stable
across cohorts in Northern Ireland (Online Appendix
Section~\ref{ApeC.7.1}), which may be restrictive in light of the secular increase in schooling over this period.\par
Next, the parallel trends assumptions in DID-IV also appear
plausible in this application. They concern the evolution of
education attainment and earnings across cohorts in the absence
of the reform, and we assess their plausibility through the
pre-trends tests in Section~\ref{sec5.2}. By contrast, Fuzzy DID
imposes parallel trends on untreated outcomes, which here
correspond to log earnings in the absence of any schooling.
Since virtually no students leave school without any education,
these trends are difficult to assess using observed data.\par
Note that once DID-IV is selected on these grounds, the Wald-DID estimand is the natural
estimand for identifying our target parameter. Unlike under Fuzzy DID, it does not require the returns to schooling to remain stable across cohorts (the stable treatment effect assumption), because DID-IV imposes parallel trends on unexposed outcomes rather than on untreated outcomes.\par
Finally, we note that, while we adopt DID-IV in this application, Fuzzy DID can also be a useful framework in other settings. For instance, suppose a government introduces a subsidy for job training in some regions, and that some workers in these regions would have enrolled over this period even without the subsidy. In this setting, DID-IV identifies the effects of training on earnings for the workers induced to enroll by the subsidy. Under Fuzzy DID, the target parameter can also capture the effects for the workers who would have enrolled without the subsidy, that is, those whose treatment is induced by time, and thus covers a potentially broader population than the LATET. If researchers are interested in these effects as well, for example because their question concerns the returns to training itself rather than the effect of the subsidy, Fuzzy DID is an attractive option.
\section{Conclusion}\label{sec7}
In this paper, we formalize an instrumented difference-in-differences (DID-IV) in multiple period settings with the staggered adoption of the instrument across units. We call this a staggered DID-IV design, and establish the target parameter and the identifying assumptions. We also propose an estimation method in this design that does not require strong restrictions on treatment effect heterogeneity. We illustrate the empirical relevance of our findings in the setting of \cite{Oreopoulos2006-bn}, estimating returns to schooling in the UK. In this application, we show that the choice of estimation method matters in staggered DID-IV designs in practice. Finally, we clarify the differences between DID-IV and Fuzzy DID designs considered in \cite{De_Chaisemartin2018-xe} and discuss their implications for empirical practice.\par
Overall, this paper provides a new econometric framework for estimating causal effects when treatment adoption is potentially endogenous over time, but researchers can exploit variation in policy adoption timing as an instrument for treatment. Further developing alternative estimation methods and diagnostic tools will be a promising avenue for future research, enhancing the credibility of DID-IV designs in practice.\par
\appendix
\renewcommand{\thesection}{\Alph{section}}
\renewcommand{\thesubsection}{\thesection.\arabic{subsection}}
\section{Proofs}
\paragraph{Proof of Theorem \ref{sec3.3.theorem1}}
\begin{proof}
    We prove part (i). Fix $e \in \mathcal{E}$ and $t \in \{e,\dots,T\}$. Note that Assumption \ref{sec3asrelevance} guarantees that the Wald-DID estimand $w_{e,t}^{DID}$ is well defined.\par
    First, we consider the numerator of $w_{e,t}^{DID}$. We have:
        \begin{align}
        &\E[Y_{i,t}-Y_{i,e-1}\mid E_i=e]
        -\E[Y_{i,t}-Y_{i,e-1}\mid E_i=\infty]\notag\\
        =&\E[Y_{i,t}(D^{e}_{i,t})-Y_{i,e-1}(D^{e}_{i,e-1})\mid E_i=e]-\E[Y_{i,t}(D^{\infty}_{i,t})-Y_{i,e-1}(D^{\infty}_{i,e-1})\mid E_i=\infty]\notag\\
        =&\E[Y_{i,t}(D^{e}_{i,t})-Y_{i,e-1}(D^{\infty}_{i,e-1})\mid E_i=e]-\E[Y_{i,t}(D^{\infty}_{i,t})-Y_{i,e-1}(D^{\infty}_{i,e-1})\mid E_i=\infty]\notag\\
        =&\E[Y_{i,t}(D^{e}_{i,t})-Y_{i,t}(D^{\infty}_{i,t})\mid E_i=e]\notag\\
        =&\E[(D^{e}_{i,t}-D^{\infty}_{i,t})\cdot (Y_{i,t}(1)-Y_{i,t}(0))\mid E_i=e]\notag\\
        \label{theorem2_eq1}
        =&CLATT_{e,t}\Prb(CM_{e,t}\mid E_i=e).
\end{align}
    The first equality follows from Assumptions~\ref{sec3as1}--\ref{sec3as3}. The second equality follows from Assumption \ref{sec3as5}. The third equality follows from adding and subtracting $\E[Y_{i,t}(D^{\infty}_{i,t})\mid E_i=e]$ and applying Assumption \ref{sec3as7}. The final equality follows from Assumption \ref{sec3as4}.\par
    Next, we consider the denominator of $w_{e,t}^{DID}$. We have:
        \begin{align}
        &\E[D_{i,t}-D_{i,e-1}\mid E_i=e]
        -\E[D_{i,t}-D_{i,e-1}\mid E_i=\infty]\notag\\
        =&\E[D^{e}_{i,t}-D^{e}_{i,e-1}\mid E_i=e]-\E[D^{\infty}_{i,t}-D^{\infty}_{i,e-1}\mid E_i=\infty]\notag\\
        =&\E[D^{e}_{i,t}-D^{\infty}_{i,e-1}\mid E_i=e]-\E[D^{\infty}_{i,t}-D^{\infty}_{i,e-1}\mid E_i=\infty]\notag\\
        =&\E[D^{e}_{i,t}-D^{\infty}_{i,t}\mid E_i=e]\notag\\
        \label{theorem2_eq2}
        =&\Prb(CM_{e,t}\mid E_i=e).
\end{align}

    The first equality follows from Assumption \ref{sec3as1}. The second equality follows from Assumption \ref{sec3as5}. The third equality follows from adding and subtracting $\E[D^{\infty}_{i,t}\mid E_i=e]$ and applying Assumption \ref{sec3as6}. The final equality follows from Assumption \ref{sec3as4}. Combining \eqref{theorem2_eq1} with \eqref{theorem2_eq2}, we obtain the desired result.\par
    The proof of part (ii) is analogous: the only additional step is that, by Assumption \ref{sec3as5}, the observed outcomes and treatments of the not-yet-exposed cohorts in periods $e-1$ and $t$ coincide with their never exposed counterparts. We therefore omit the details.
\end{proof}

\paragraph{Proof of Theorem \ref{sec4.1thm1}}
\begin{proof}
We prove that our STS estimator is asymptotically normal, deriving its influence function. The derivation does not depend on whether we are in case (i) or case (ii). We define $\theta_{e,t}$ to be
\begin{align*}
\theta_{e,t}&\equiv \frac{\E[Y_{i,t}-Y_{i,e-1}|E_i=e]-\E[Y_{i,t}-Y_{i,e-1}|E_i \in U]}{\E[D_{i,t}-D_{i,e-1}|E_i=e]-\E[D_{i,t}-D_{i,e-1}|E_i \in U]}\equiv \frac{\alpha_{e,t}^1-\alpha_{e,t}^2}{\pi_{e,t}^1-\pi_{e,t}^2}.
\end{align*}
Here, $\theta_{e,t}=CLATT_{e,t}$ holds by Theorem \ref{sec3.3.theorem1}.\par
The following fact, also found in \cite{De_Chaisemartin2018-xe}, is repeatedly used in the derivation of the influence function: if $\sqrt{N}(\hat{A}-A)=\frac{1}{\sqrt{N}}\sum_{i=1}^{N} a_i+o_p(1)$ and $\sqrt{N}(\hat{B}-B)=\frac{1}{\sqrt{N}}\sum_{i=1}^{N} b_i+o_p(1)$, we have $\sqrt{N}\bigl(\hat{A}/\hat{B}-A/B\bigr)=\frac{1}{\sqrt{N}}\sum_{i=1}^{N}\bigl\{a_i-(A/B)b_i\bigr\}/B+o_p(1)$.
Then, by some calculations, we obtain \eqref{sec4.1thm1_expansion}. The influence function $\psi_{i,e,t}$ is:
\footnotesize
\begin{align}
\label{influence_eq1}
\psi_{i,e,t}
=\frac{1}{\pi_{e,t}^1-\pi_{e,t}^2}\bigg(\frac{\mathbf{1}\{E_i=e\}\left[\delta_{i,e,t}^{p}-\E[\delta_{i,e,t}^{p}|E_i=e]\right]}{\Prb(E_i=e)}-\frac{\mathbf{1}\{E_i \in U\}\left[\delta_{i,e,t}^{p}-\E[\delta_{i,e,t}^{p}|E_i \in U]\right]}{\Prb(E_i \in U)}\bigg),
\end{align}
\normalsize
where we define $\delta_{i,e,t}^{p}=(Y_{i,t}-Y_{i,e-1})-\theta_{e,t}\cdot(D_{i,t}-D_{i,e-1})$.\par
Finally, since $\{\Psi_{i,t \geq e}\}_{i=1}^{N}$ are i.i.d.\ with mean zero and finite second moments, the multivariate central limit theorem yields \eqref{sec4.1thm1_clt}.
\end{proof}
\paragraph{Proof of Theorem \ref{bootstrap_coverage}}
\begin{proof}
    We first establish the conditional convergence
    \begin{align*}
        \sqrt{N}(\widehat{CLATT}^{\ast}_{t \geq e}-\widehat{CLATT}_{t \geq e}) \xrightarrow[\ast]{d} \mathcal{N}(0,\Sigma),
    \end{align*}
    where $\xrightarrow[\ast]{d}$ denotes weak convergence of the bootstrap law in probability, i.e., conditional on the original sample $\{(Y_{i,t},D_{i,t},Z_{i,t})_{t=1}^{T}\}_{i=1}^{N}$. By the multiplier central limit theorem (see, e.g., Lemma 2.9.5 in \cite{Vaart1996}), we have $\frac{1}{\sqrt{N}}\sum_{i=1}^{N}V_i\Psi_{i,t \geq e} \xrightarrow[\ast]{d} \mathcal{N}(0,\Sigma)$. Therefore, it suffices to show:
     \begin{align*}
        \frac{1}{\sqrt{N}}\sum_{i=1}^{N}V_i (\hat{\Psi}_{i,t \geq e}-\Psi_{i,t \geq e})=o_p^{\ast}(1),
    \end{align*}
    where $o_p^{\ast}(1)$ denotes a term that converges to zero in probability, conditional on the original sample.\par
Fix \((e,t)\). Let
\(A_i=Y_{i,t}-Y_{i,e-1}\), \(B_i=D_{i,t}-D_{i,e-1}\),
\(S_i=1\{E_i=e\}\), and \(C_i=1\{E_i\in U\}\). Let
\(\pi_{e,t}=\E[B_i\mid E_i=e]-\E[B_i\mid E_i\in U]\), \(p_e=\E[S_i]\), \(p_U=\E[C_i]\),
\(m_e=\E[A_i-\theta_{e,t}B_i\mid E_i=e]\), and
\(m_U=\E[A_i-\theta_{e,t}B_i\mid E_i\in U]\). By the expression for the influence
function in \eqref{influence_eq1}, we can write \(\psi_{i,e,t}=\lambda_{e,t}'R_{i,e,t}\), where
\begin{align*}
R_{i,e,t}
&=(S_iA_i,\ S_iB_i,\ S_i,\ C_iA_i,\ C_iB_i,\ C_i)',\\
\lambda_{e,t}
&=
\left(
\frac{1}{\pi_{e,t}p_e},
-\frac{\theta_{e,t}}{\pi_{e,t}p_e},
-\frac{m_e}{\pi_{e,t}p_e},
-\frac{1}{\pi_{e,t}p_U},
\frac{\theta_{e,t}}{\pi_{e,t}p_U},
\frac{m_U}{\pi_{e,t}p_U}
\right)'.
\end{align*}
Note that the plug-in influence function satisfies
\(\widehat\psi_{i,e,t}=\widehat\lambda_{e,t}'R_{i,e,t}\), where
\(\widehat\lambda_{e,t}\) is the sample analog of \(\lambda_{e,t}\). By the law of large numbers and the continuous mapping
theorem, \(\widehat\lambda_{e,t}\to_p\lambda_{e,t}\). Moreover, since $\E[Y_{i,s}^2], \E[D_{i,s}^2]<\infty$ for $s\in\{e-1,t\}$, we have
$N^{-1}\sum_{i=1}^{N} R_{i,e,t}R_{i,e,t}'=O_p(1)$, and hence
$N^{-1}\sum_{i=1}^{N}(\widehat\psi_{i,e,t}-\psi_{i,e,t})^2=o_p(1)$.
Conditional on the data,
\(N^{-1/2}\sum_{i=1}^{N}V_i(\widehat\psi_{i,e,t}-\psi_{i,e,t})\)
has mean zero and variance
\(N^{-1}\sum_{i=1}^{N}(\widehat\psi_{i,e,t}-\psi_{i,e,t})^2=o_p(1)\), so
Chebyshev's inequality yields
\(N^{-1/2}\sum_{i=1}^{N}V_i(\widehat\psi_{i,e,t}-\psi_{i,e,t})=o_p^{\ast}(1)\).
Since \(\Psi_{i,t\ge e}\) is finite-dimensional, we obtain the conditional convergence.\par
Given Theorem \ref{sec4.1thm1} and this conditional convergence, the bootstrap consistently approximates the distribution of $\sqrt{N}(\widehat{CLATT}_{t \geq e}-CLATT_{t \geq e})$. Since the interquartile range and the maximum are continuous functionals (continuous at the limit because \(\E[\psi_{i,e,t}^{2}]>0\), with the empirical quantiles converging to the bootstrap quantiles as \(B\to\infty\)), $\widehat{\sigma}_{e,t}$ and $\widehat{c}_{1-\alpha}$ in Algorithm~\ref{bootstrap_algorithm} are consistent by the continuous mapping theorem, and the coverage statement follows.
\end{proof}
\section{Summary causal parameters}\label{OnlineWeights}
This section provides additional results on the summary causal parameters in staggered DID-IV designs, introduced in Section \ref{sec4.4}. Table \ref{table5} presents the specific expressions of the weighting functions for a variety of summary causal parameters, corresponding to those in \cite{Callaway2021-wl}. Corollary \ref{sec6lemma} below derives the influence function of $\hat{\theta}^{IV}$.
\begin{Corollary}
\label{sec6lemma}
Suppose the assumptions of Theorem \ref{sec4.1thm1} hold. Assume also that each $\hat{w}(e,t)$ admits the asymptotically linear representation $\sqrt{N}(\hat{w}(e,t)-w(e,t))=\frac{1}{\sqrt{N}}\sum_{i=1}^{N} \zeta_{i,e,t}+o_p(1)$, where $\zeta_{i,e,t}$ satisfies $\E[\zeta_{i,e,t}]=0$ and $\E[(\zeta_{i,e,t})^2] < \infty$.\footnote{This condition holds for the sample analogs of the weighting functions in Table \ref{table5}.} Then, we have $\sqrt{N}(\hat{\theta}^{IV}-\theta^{IV})=\frac{1}{\sqrt{N}}\sum_{i=1}^{N}l_{i}+o_p(1)$, where
\begin{align*}
l_{i}=\sum_{e}\sum_{t=1}^T\left( w(e,t)\cdot \psi_{i,e,t}+\zeta_{i,e,t}\cdot{CLATT}_{e,t} \right).
\end{align*}
\end{Corollary}
\begin{table}[t]
\centering
\caption{Weights for a variety of summary causal parameters}
\label{table5}
\footnotesize
\renewcommand{\arraystretch}{1.6}
\begin{tabular}{ccc}
\hline
Target Parameter & $w(e,t)$\\
\hline
$\theta_{es}^{IV}(l)$ & $\mathbf{1}\{e+l \leq T\}\mathbf{1}\{t=e+l\}\displaystyle\frac{\Prb(E=e|E+l \leq T)\,CAET_{e,e+l}^{1}}{\sum_{e': e'+l \leq T}\Prb(E=e'|E+l \leq T)\,CAET_{e',e'+l}^{1}}$\\
$\theta_{es}^{bal,IV}(l,l')$ & $\mathbf{1}\{e+l' \leq T\}\mathbf{1}\{t=e+l\}\displaystyle\frac{\Prb(E=e|E+l' \leq T)\,CAET_{e,e+l}^{1}}{\sum_{e': e'+l' \leq T}\Prb(E=e'|E+l' \leq T)\,CAET_{e',e'+l}^{1}}$\\
$\theta_{sel}^{IV}(\Tilde{e})$ & $\mathbf{1}\{t \geq e\}\mathbf{1}\{e=\Tilde{e}\}\displaystyle\frac{CAET_{\Tilde{e},t}^{1}}{\sum_{t'=\Tilde{e}}^T CAET_{\Tilde{e},t'}^{1}}$\\
$\theta_{c}^{IV}(\Tilde{t})$ & $\mathbf{1}\{t \geq e\}\mathbf{1}\{t=\Tilde{t}\}\displaystyle\frac{\Prb(E=e|E \leq \Tilde{t})\,CAET_{e,\Tilde{t}}^{1}}{\sum_{e' \leq \Tilde{t}}\Prb(E=e'|E \leq \Tilde{t})\,CAET_{e',\Tilde{t}}^{1}}$\\
$\theta_{c}^{cumu,IV}(\Tilde{t})$ & $\mathbf{1}\{t \geq e\}\mathbf{1}\{t \leq \Tilde{t}\}\displaystyle\frac{\Prb(E=e|E \leq t)\,CAET_{e,t}^{1}}{\sum_{e' \leq t}\Prb(E=e'|E \leq t)\,CAET_{e',t}^{1}}$\\
$\theta_{W}^{o,IV}$ & $\displaystyle\frac{\mathbf{1}\{t \geq e\}\,\Prb(E=e|E \leq T)\,CAET_{e,t}^{1}}{\sum_{e' \in \mathcal{S}(E)}\sum_{t'=e'}^{T}\Prb(E=e'|E \leq T)\,CAET_{e',t'}^{1}}$\\
$\theta_{sel}^{o,IV}$ & $\mathbf{1}\{t \geq e\}\,\Prb(E=e|E \leq T)\displaystyle\frac{CAET_{e,t}^{1}}{\sum_{t'=e}^T CAET_{e,t'}^{1}}$\\
\hline
\end{tabular}
\\[5pt]
\begin{minipage}{0.95\textwidth}
\footnotesize
\textit{Notes}: This table represents the specific expressions for the weights on each $CLATT(e,t)$ or each $CACRT(e,t)$ in a variety of summary causal parameters. Here, $l=t-e$ denotes the relative period from the initial exposure to the instrument. The symbols $\Tilde{e}$, $\Tilde{t}$, and $l'$ denote a fixed cohort, a fixed calendar period, and a fixed relative period, respectively. Each target parameter without the superscript IV is defined in \cite{Callaway2021-wl}. Under the monotonicity assumption (Assumption \ref{sec3as4}), the weights are non-negative and sum to one in all cases except for $\theta_{c}^{cumu,IV}(\Tilde{t})$, whose weights sum to the number of the aggregated periods, reflecting that it is a cumulative measure.
\end{minipage}
\end{table}

\section{Limited anticipation in \cite{Oreopoulos2006-bn}}\label{AppendixC}
In this appendix, we replace the no anticipation assumption with the limited anticipation assumption described in
Section~\ref{sec3}, and report the summary estimates for \(\delta=1,2,3\).\footnote{The choice $\delta \le 3$ reflects the three-year gap between the 1944 Education Act and its 1947 implementation (see \cite{Oreopoulos2006-bn}).} To keep the estimand comparable
to the baseline \(\theta_{sel}^{IV}(e)\), we aggregate over post-reform cohorts \(t\ge 1947\) for which Northern Ireland
remains a valid comparison group under \(\delta\)-period anticipation.\footnote{If
we instead aggregate over the effective post-exposure cohorts, \(t\ge 1947-\delta\),
the estimates are quantitatively similar and lead to the same qualitative conclusion.}
We also report the corresponding estimates from saturated STS regressions with
five-year age-bin controls.\par
Table~\ref{tab:limited_anticipation_oreopoulos} shows
that the estimates under limited anticipation are generally less precise than
the baseline estimates and sensitive to the choice of \(\delta\). In Panel A, the
point estimates are close to zero for \(\delta=1\) and \(\delta=3\), and positive
for \(\delta=2\). In Panel B,
the estimates are negative for all three values of \(\delta\). In both panels,
however, the conventional and AR confidence intervals include zero in every case.\footnote{The smaller and imprecise estimates under limited anticipation are not fully consistent with the positive return under the no anticipation assumption, but they are broadly in line with
\cite{Devereux} and \cite{CLARK2023104902}, who revisit the same 1947
British reform and find smaller or near-zero returns to schooling.}\par
\begin{table}[t]
\centering
\footnotesize
\caption{Summary causal parameter under limited anticipation}
\label{tab:limited_anticipation_oreopoulos}
\renewcommand{\arraystretch}{1.3}

\begin{tabular*}{0.98\textwidth}{@{\extracolsep{\fill}}lccc@{}}
\hline
 & $\delta=1$ & $\delta=2$ & $\delta=3$ \\
\hline
\multicolumn{4}{@{}l}{\textit{Panel A: Baseline}}\\
Estimate & -0.008 & 0.105 & -0.004 \\
SE & 0.115 & 0.093 & 0.105 \\
95\% CI & $[-0.233, 0.217]$ & $[-0.078, 0.288]$ & $[-0.209, 0.202]$ \\
First stage $F$ & 9.5 & 13.4 & 9.4 \\
$F>10$ & -- & $\checkmark$ & -- \\
AR 95\% CI & $[-0.403,0.207]$ & $[-0.111,0.322]$ & $[-0.326,0.218]$ \\
Sample size & 24,747 & 21,091 & 17,672 \\
\hline
\multicolumn{4}{@{}l}{\textit{Panel B: With age-bin controls}}\\
Estimate & -0.143 & -0.041 & -0.161 \\
SE & 0.123 & 0.095 & 0.126 \\
95\% CI & $[-0.384, 0.097]$ & $[-0.227, 0.144]$ & $[-0.409, 0.086]$ \\
First stage $F$ & 12.5 & 15.1 & 9.5 \\
$F>10$ & $\checkmark$ & $\checkmark$ & -- \\
AR 95\% CI & $[-0.572,0.051]$ & $[-0.311,0.130]$ & $[-0.682,0.036]$ \\
Sample size & 8,807 & 8,417 & 8,085 \\
\hline
\end{tabular*}

\vspace{0.3em}

\begin{minipage}{0.98\textwidth}
\footnotesize
\textit{Notes}: This table reports estimates of the summary causal parameter under limited anticipation in the first stage. Panel A reports the baseline estimates, and Panel B reports the estimates with the five-year age-bin controls. For each $\delta$, the reference cohort is $1947-\delta-1$. The SEs, the first stage $F$ statistics, and the AR confidence sets are computed from the corresponding influence functions, as described in Section~\ref{sec5.5}. 95\% CIs are Wald intervals based on these standard errors. The row $F>10$ reports whether the statistic exceeds the rule-of-thumb threshold of 10. The AR 95\% CI row reports the Anderson--Rubin confidence sets. Sample size reports the number of observations used for each estimator.
\end{minipage}

\end{table}

\section{Decomposing the TWFEIV estimand}\label{ApeD}
In this appendix, we quantify the bias terms in the TWFEIV estimand, arising from the bad comparisons in \cite{Oreopoulos2006-bn}. Let $e=1947$ and $U=1957$ denote Britain and Northern Ireland, respectively. Let $T_i$ denote the cohort of individual $i$, and let $N_{g,t}$ denote the sample size in region $g \in \{e,U\}$ and cohort $t$. For a generic random variable $R_{i}$, let $R_{g,t} \equiv \frac{1}{N_{g,t}}\sum_{i:\,E_i=g,\,T_i=t}R_{i}$ denote its sample mean in region $g$ and cohort $t$.\par
\begin{figure}[t]
    \centering
    \begin{adjustbox}{width=\columnwidth, height=0.6\columnwidth, keepaspectratio}
        \includegraphics{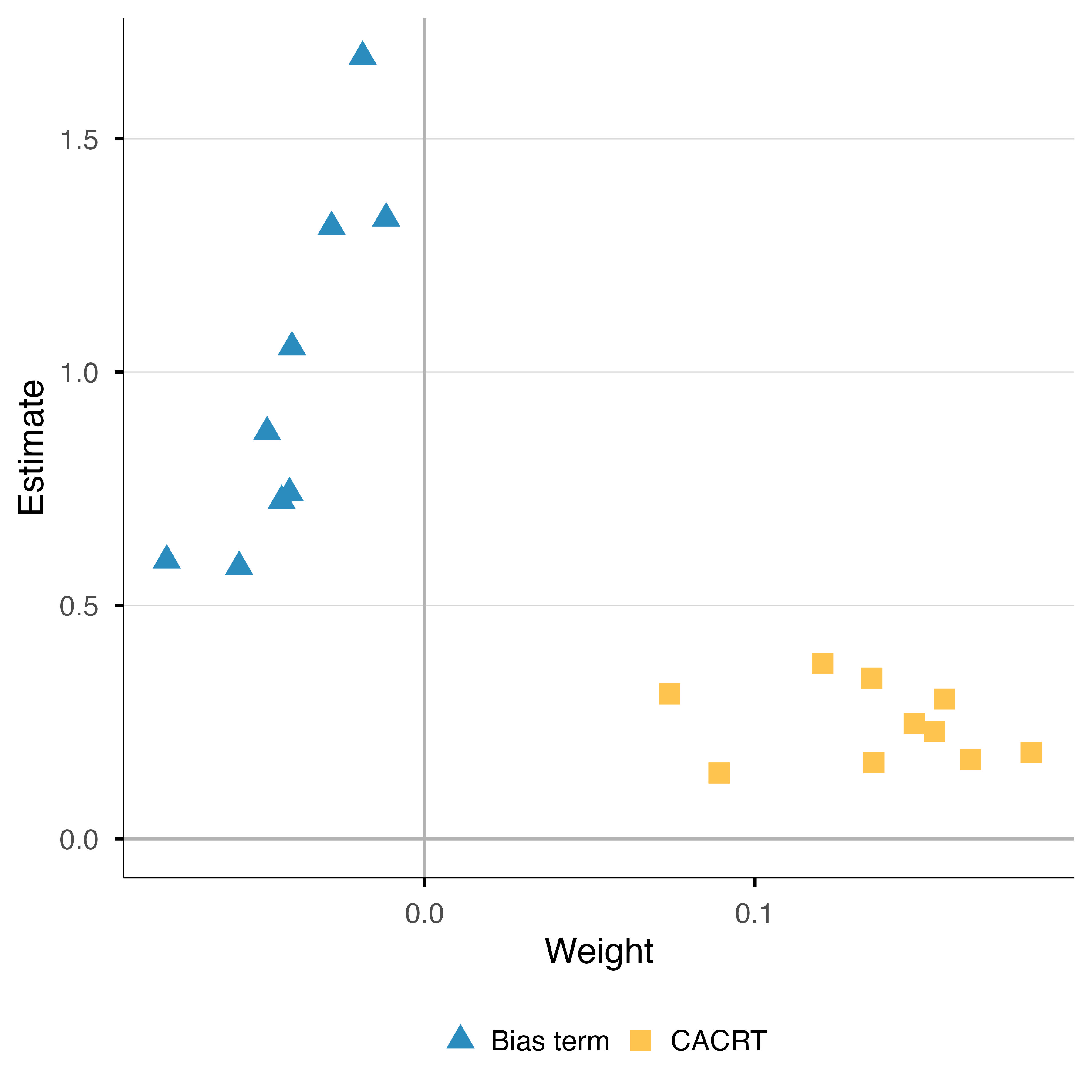}
    \end{adjustbox}
    \caption{Decomposition result for the TWFEIV estimand in \cite{Oreopoulos2006-bn}. \textit{Notes}: This figure plots the estimated weights and the corresponding estimates for each $CACRT_{e,t}$ and $\Delta_{e,t}$ (bias term). The closed squares are the weights and the corresponding estimates for each $CACRT_{e,t}$. The closed triangles are the weights and the corresponding estimates for each $\Delta_{e,t}$ (bias term).}
    \label{ApeDfigure1}
\end{figure}
First, we decompose the TWFEIV estimator $\hat{\beta}_{IV}$ as follows:
\begin{align*}
\hat{\beta}_{IV}=\frac{\sum_{t}N_{e,t}\hat{Z}_{e,t}\left[D_{e,t}-D_{e,1946}-(D_{U,t}-D_{U,1946})\right]\cdot \hat{w}^{DID}_{e,t}}{\sum_{t}N_{e,t}\hat{Z}_{e,t}\left[D_{e,t}-D_{e,1946}-(D_{U,t}-D_{U,1946})\right]},
\end{align*}
where $\hat{w}^{DID}_{e,t} \equiv \displaystyle\frac{\left[Y_{e,t}-Y_{e,1946}-(Y_{U,t}-Y_{U,1946})\right]}{\left[D_{e,t}-D_{e,1946}-(D_{U,t}-D_{U,1946})\right]}$ and $\hat{Z}_{e,t}$ is the sample mean of the residual $\hat{Z}_{i}$ from regressing $Z_{i}$ on the cohort and Northern Ireland fixed effects.\par
Here, to clarify the interpretation of the TWFEIV estimand, we assume the parallel trends assumptions in the treatment and the outcome over the entire sample period.
Then, by using the similar arguments in the proof of Theorem \ref{OnlineA.1.thm2} in the Online Appendix repeatedly, we obtain the following decomposition result for the TWFEIV estimand:
\begin{align}
\label{decomposition1}
\beta_{IV}=&\sum\limits_{U-1 \geq t \geq e}w_{e,t}^{1}\cdot CACRT_{e,t}+\sum\limits_{t \geq U}w_{e,t}^{2} \cdot \Delta_{e,t},
\end{align}
where we define $CAET^{1}_{e,t}=\E[D_{i}-D_{i}^{\infty}|T_i=t,E_i=e]$ and
\begin{align*}
\Delta_{e,t}=\frac{CAET^{1}_{e,t}\cdot CACRT_{e,t}-CAET^{1}_{U,t}\cdot CACRT_{U,t}}{CAET^{1}_{e,t}-CAET^{1}_{U,t}},
\end{align*}
and the weights $w^{1}_{e,t}$ and $w^{2}_{e,t}$ are:
\footnotesize
\begin{align*}
&w_{e,t}^1
=\frac{z_{e,t}\cdot n_{e,t} \cdot CAET^1_{e,t}}{\sum\limits_{U-1 \geq t' \geq e} z_{e,t'}\cdot n_{e,t'}\cdot CAET^1_{e,t'}+\sum\limits_{t' \geq U} z_{e,t'}\cdot n_{e,t'}\cdot (CAET^{1}_{e,t'}-CAET^{1}_{U,t'})},\\
&w_{e,t}^{2}
=\frac{z_{e,t}\cdot n_{e,t}\cdot(CAET^{1}_{e,t}-CAET^{1}_{U,t})}{\sum\limits_{U-1 \geq t' \geq e} z_{e,t'}\cdot n_{e,t'}\cdot CAET^1_{e,t'}+\sum\limits_{t' \geq U} z_{e,t'}\cdot n_{e,t'}\cdot (CAET^{1}_{e,t'}-CAET^{1}_{U,t'})}.
\end{align*}
\normalsize
Here, $z_{e,t}=\E[\hat{Z}_{i}|E_i=e,T_i=t]$ and $n_{e,t}$ represent the population counterparts of the cell-mean residual $\hat{Z}_{e,t}$ and the sample share $N_{e,t}/N$ for Britain in cohort $t$, respectively.\par
In \cite{Oreopoulos2006-bn}, we can only identify $CACRT_{e,t}$ between $1947$ and $1956$ because the policy change occurred in Britain in $1947$ and in Northern Ireland in $1957$. This implies that each $\Delta_{e,t}$ in equation \eqref{decomposition1} is the bias term arising from the bad comparisons between $1957$ and $1965$ performed by the TWFEIV regression.\par
Figure \ref{ApeDfigure1} plots the estimated weight and the corresponding estimate for each $CACRT_{e,t}$ and $\Delta_{e,t}$ (bias term).\footnote{The consistent estimators for each $w^{1}_{e,t}$, $w^{2}_{e,t}$, $CACRT_{e,t}$, and $\Delta_{e,t}$ are constructed as the sample analogs, using $\hat{w}^{DID}_{e,t}$ and $(D_{e,t}-D_{e,1946})-(D_{U,t}-D_{U,1946})$.} From this figure, we see that the TWFEIV estimand is negatively biased in the setting of \cite{Oreopoulos2006-bn}: all the CACRTs are positive and positively weighted, whereas all the bias terms are positive and negatively weighted.\par
\bibliographystyle{econ-econometrica.bst}
\bibliography{reference}
\newpage
\appendix
\interfootnotelinepenalty=100
\makeatletter
\gdef\@thanks{}  
\makeatother
\title{Supplement to ``Instrumented Difference-in-Differences with Heterogeneous Treatment Effects''}
\author{Sho Miyaji}

\date{\today}

\onehalfspacing
\maketitle
\begin{abstract}
In this supplement, we first provide additional results for the extensions discussed in the main text. We then investigate the connections between DID-IV and Fuzzy DID designs considered in \citeonline{De_Chaisemartin2018-xe}.\par
\end{abstract}
\renewcommand{\thesection}{\Alph{section}}
\section{Extensions}\label{ApeB}
This section contains the results for the four extensions: limited anticipation in the first stage, ordered treatments, repeated cross sections, and saturated STS regressions.
\subsection{Limited anticipation}\label{ApeB0}
In this section, we present the identification results under limited anticipation in the first stage (Assumption \ref{sec4as1} in Section \ref{sec3}). Let $\mathcal{E}_{\delta} = \mathcal{E} \cap \{2+\delta,\dots,T\}$, where $\delta$ is the anticipation horizon defined in Assumption \ref{sec4as1}. For notational convenience, we interpret \(\bar e-\delta\) as \(\infty\) when \(\bar e=\infty\).\par
To identify each $CLATT_{e,t}$, we consider the following Wald-DID estimand, in which the pre-exposed period is now $e-\delta-1$:
\begin{align*}
    w^{DID}_{e,t}=\frac{\E[Y_{i,t}-Y_{i,e-\delta-1}|E_i=e]-\E[Y_{i,t}-Y_{i,e-\delta-1}|E_i \in U]}{\E[D_{i,t}-D_{i,e-\delta-1}|E_i=e]-\E[D_{i,t}-D_{i,e-\delta-1}|E_i \in U]},
\end{align*}
where $U \subseteq \mathcal{S}(E_i)$ denotes the control set. When the control group is the never exposed cohort, we set $U=\{\infty\}$. When the control group is the not-yet-exposed cohorts, we set $U=\{e' \in \mathcal{S}(E_i): e'>t+\delta\}$ in each period $t$, that is, the units with $Z_{i,t+\delta}=0$.\par
We consider the following identifying assumptions, which extend the monotonicity assumption, the parallel trends assumptions, and the relevance condition in Section \ref{sec2.3} to the anticipation periods $t \geq e-\delta$.
\begin{Assumption}[Monotonicity assumption in multiple time periods]
\label{sec4as2}
For each $e\in \mathcal{E}$ and $t \in \{2,\dots,T\}$ such that $t \geq e-\delta$, $\Prb(D_{i,t}^{e} \geq D_{i,t}^{\infty})=1$ or $\Prb(D_{i,t}^{e} \leq D_{i,t}^{\infty})=1$.
\end{Assumption}
As in Section \ref{sec2.3}, we hereafter consider the type of the monotonicity assumption that rules out the defiers $DF_{e,t}$ for all $t \geq e-\delta$ in any cohort $e \in \mathcal{E}$.\par
\begin{Assumption}[Parallel trends assumption in the treatment]
\label{sec4as3}
For each $e\in \mathcal{E}$ and $t\in \{2,\dots,T\}$ such that $t \geq e-\delta$, $\E[D_{i,t}^{\infty}-D_{i,t-1}^{\infty}|E_i=e]=\E[D_{i,t}^{\infty}-D_{i,t-1}^{\infty}|E_i=\infty]$.
\end{Assumption}
\begin{Assumption}[Parallel trends assumption in the outcome]
\label{sec4as4}
For each $e\in \mathcal{E}$ and $t\in \{2,\dots,T\}$ such that $t \geq e-\delta$,
\begin{align*}
\E[Y_{i,t}(D_{i,t}^{\infty})-Y_{i,t-1}(D_{i,t-1}^{\infty})|E_i=e]=\E[Y_{i,t}(D_{i,t}^{\infty})-Y_{i,t-1}(D_{i,t-1}^{\infty})|E_i=\infty].
\end{align*}
\end{Assumption}
\begin{Assumption}[Relevance condition]
\label{sec4as5}
For each $e\in \mathcal{E}_{\delta}$ and $t\in \{2,\dots,T\}$ such that $t \geq e-\delta$, $\E[D_{i,t}-D_{i,e-\delta-1}|E_i=e]-\E[D_{i,t}-D_{i,e-\delta-1}|E_i=\infty] > 0$.
\end{Assumption}
\begin{Assumption}[Parallel trends assumption in the treatment based on not-yet-exposed cohorts]
\label{sec4as3ny}
For each $e\in \mathcal{E}$ and each $(s,t) \in \{2,\dots,T\} \times \{2,\dots,T\}$ such that $t \geq e-\delta$ and $t+\delta \leq s < \bar{e}$, $\E[D_{i,t}^{\infty}-D_{i,t-1}^{\infty}|E_i=e]=\E[D_{i,t}^{\infty}-D_{i,t-1}^{\infty}|Z_{i,s}=0]$.
\end{Assumption}
\begin{Assumption}[Parallel trends assumption in the outcome based on not-yet-exposed cohorts]
\label{sec4as4ny}
For each $e\in \mathcal{E}$ and each $(s,t) \in \{2,\dots,T\} \times \{2,\dots,T\}$ such that $t \geq e-\delta$ and $t+\delta \leq s < \bar{e}$,
\begin{align*}
\E[Y_{i,t}(D_{i,t}^{\infty})-Y_{i,t-1}(D_{i,t-1}^{\infty})|E_i=e]=\E[Y_{i,t}(D_{i,t}^{\infty})-Y_{i,t-1}(D_{i,t-1}^{\infty})|Z_{i,s}=0].
\end{align*}
\end{Assumption}
\begin{Assumption}[Relevance condition based on not-yet-exposed cohorts]
\label{sec4as5ny}
For each $e\in \mathcal{E}_{\delta}$ and $t\in \{2,\dots,T-\delta\}$ such that $e-\delta \leq t < \bar{e}-\delta$, $\E[D_{i,t}-D_{i,e-\delta-1}|E_i=e]-\E[D_{i,t}-D_{i,e-\delta-1}|Z_{i,t+\delta}=0] > 0$.
\end{Assumption}
The following theorem extends Theorem \ref{sec3.3.theorem1} to the limited anticipation setting. When $\delta=0$, Theorem \ref{sec4_anticipation_theorem1} reduces to Theorem \ref{sec3.3.theorem1}.\footnote{Theorem~\ref{sec4_anticipation_theorem1}(i) identifies \(CLATT_{e,t}\) only up to \(t=T-\delta\) because units in the never exposed cohort may anticipate exposure that occurs after period \(T\). In practice, however, if some units are known to remain unexposed to the instrument after period \(T\), these units can serve as a valid comparison group, and \(CLATT_{e,t}\) can be identified up to \(t=T\), provided that the corresponding parallel trends and relevance conditions hold.}\par
\begin{Theorem}
    \label{sec4_anticipation_theorem1}
    Let Assumptions~\ref{sec3as1}--\ref{sec3as3} hold.
    \begin{itemize}
        \item[(i)] Suppose Assumptions~\ref{sec4as1}--\ref{sec4as5} hold and $U=\{\infty\}$. Then, for all $e \in \mathcal{E}_{\delta}$ and $t \in \{2,\dots,T-\delta\}$ such that $t \geq e-\delta$, we have $w^{DID}_{e,t}=CLATT_{e,t}$.
        \item[(ii)] Suppose Assumptions~\ref{sec4as1}--\ref{sec4as2} and \ref{sec4as3ny}--\ref{sec4as5ny} hold and $U=\{e' \in \mathcal{S}(E_i): e'>t+\delta\}$. Then, for all $e \in \mathcal{E}_{\delta}$ and $t \in \{2,\dots,T-\delta\}$ such that $e-\delta \leq t < \bar{e}-\delta$, we have $w^{DID}_{e,t}=CLATT_{e,t}$.
    \end{itemize}
\end{Theorem}\par
\paragraph{Proof of Theorem \ref{sec4_anticipation_theorem1}}
\begin{proof}
The proof follows the same argument as that of
Theorem~\ref{sec3.3.theorem1}, replacing the pre-exposed period $e-1$
with $e-\delta-1$ and extending the monotonicity and parallel trends
assumptions to the anticipation periods $t \geq e-\delta$. For the
not-yet-exposed control group, Assumption~\ref{sec4as1} ensures that
the observed outcomes and treatments of the units with
$Z_{i,t+\delta}=0$ coincide with their never exposed counterparts in
periods $e-\delta-1$ and $t$. We therefore omit the details.
\end{proof}
\subsection{Non-binary, ordered treatment}\label{ApeB1}
In this section, we present the identification results with non-binary, ordered treatment. Theorem \ref{OnlineA.1.thm2} below shows that the Wald-DID estimands in Section \ref{sec2.3} identify the corresponding $CACRT_{e,t}$ under the same assumptions as in Theorem \ref{sec3.3.theorem1}.
\begin{Theorem}
\label{OnlineA.1.thm2}
Suppose that the treatment is non-binary and ordered: $D_{i,t} \in \{0,1,\dots,J\}$.
\begin{itemize}
\item[(i)] Under the assumptions of Theorem \ref{sec3.3.theorem1}(i), for each $e \in \mathcal{E}$ and $t \in \{e,\dots,T\}$, we have $w^{DID}_{e,t}=CACRT_{e,t}$.
\item[(ii)] Under the assumptions of Theorem \ref{sec3.3.theorem1}(ii), for each $e \in \mathcal{E}$ and $t \in \{e,\dots,T\}$ such that $t < \bar{e}$, we have $w^{DID,ny}_{e,t}=CACRT_{e,t}$.
\end{itemize}
\end{Theorem}
\begin{proof}
We prove part (i); the proof of part (ii) is analogous. Fix $e \in \mathcal{E}$ and $t \in \{e,\dots,T\}$. First, we consider the numerator of $w^{DID}_{e,t}$. By the same argument as in the proof of Theorem \ref{sec3.3.theorem1}, we have
\begin{align*}
&\E[Y_{i,t}-Y_{i,e-1}\mid E_i=e]-\E[Y_{i,t}-Y_{i,e-1}\mid E_i=\infty]\\
&\quad=\E[Y_{i,t}(D_{i,t}^{e})-Y_{i,t}(D_{i,t}^{\infty})\mid E_i=e].
\end{align*}
Since Assumption \ref{sec3as4} rules out the defiers, we have $D_{i,t}^{e} \geq D_{i,t}^{\infty}$ almost surely, and hence:
\begin{align}
&\E[Y_{i,t}(D_{i,t}^{e})-Y_{i,t}(D_{i,t}^{\infty})\mid E_i=e]\notag\\
&=\sum_{j=1}^J \E[Y_{i,t}(j)-Y_{i,t}(j-1)\mid D_{i,t}^{e} \geq j > D_{i,t}^{\infty},E_i=e]\notag\\
\label{ApeAmulti3}
&\qquad\times \Prb(D_{i,t}^{e} \geq j > D_{i,t}^{\infty}\mid E_i=e).
\end{align}\par
Next, we consider the denominator of $w^{DID}_{e,t}$. By the same argument as in the proof of Theorem \ref{sec3.3.theorem1}, together with $D_{i,t}^{e}-D_{i,t}^{\infty}=\sum_{j=1}^J \mathbf{1}\{D_{i,t}^{e}\geq j>D_{i,t}^{\infty}\}$, we have
\begin{align}
&\E[D_{i,t}-D_{i,e-1}\mid E_i=e]-\E[D_{i,t}-D_{i,e-1}\mid E_i=\infty]\notag\\
\label{ApeAmulti4}
=&\E[D_{i,t}^{e}-D_{i,t}^{\infty}\mid E_i=e]=\sum_{j=1}^J \Prb(D_{i,t}^{e} \geq j > D_{i,t}^{\infty}\mid E_i=e).
\end{align}
Combining \eqref{ApeAmulti4} with \eqref{ApeAmulti3}, we obtain the desired result.
\end{proof}

\subsection{Repeated cross sections}\label{ApeB2}
In this section, we extend our staggered DID-IV framework to repeated cross
section settings. We first introduce the notation. Let $T_i \in \{1,\dots,T\}$ denote the time
period in which unit $i$ is observed: unlike in the panel case, each unit is
observed only once. For each unit $i$, let $Y_{i}$ denote
the outcome and $D_{i}$ denote the treatment. Let $Z_{i,t} \in \{0,1\}$
denote the instrument status of unit $i$ in period $t$, and let
$Z_i=(Z_{i,1},\dots,Z_{i,T})$ denote the instrument path. Although unit $i$
is observed only in period $T_i$, we assume that its instrument path is
known to researchers.\footnote{This requirement is satisfied, for instance,
when the policy is adopted at the regional level and the data record each
unit's region.} For each unit $i$, we observe $(Y_i,D_i,Z_i,T_i)$, and we
assume that $\{(Y_i,D_i,Z_i,T_i)\}_{i=1}^{N}$ are i.i.d.\par
We maintain Assumption \ref{sec3as1}; that is, we assume
the staggered adoption of the instrument across units. We define the cohort variable $E_i$, the largest cohort $\bar{e}$, and the set $\mathcal{E}$ as in Section \ref{sec2} of the main text.\par
For
$d \in \mathcal{S}(D)$ and $z \in \mathcal{S}(Z)$, let $Y_{i}(d,z)$ and $D_{i}(z)$ denote
the potential outcome and the potential treatment status for unit $i$. As in Section \ref{sec2}, we write $D_{i}^{e}$ for the potential
treatment status when unit $i$ belongs to cohort $e$, and $D_{i}^{\infty}$
when it is never exposed. For all $t \in \{1,\dots,T\}$ and $e \in \mathcal{S}(E_i)$, let $D_{i,t}^{e}$ denote the potential treatment status of unit $i$ in period $t$ when it belongs to cohort $e$. Note that we have $D_{i}^{e}=\sum_{t=1}^{T}\mathbf{1}\{T_i=t\}D_{i,t}^{e}$.\par
We impose the following exclusion restriction.
\begin{Assumption}[Exclusion restriction in repeated cross sections]
\label{ApeB2as0}
For all $z \in \mathcal{S}(Z)$ and $d \in \mathcal{S}(D)$, $Y_{i}(d,z)=Y_{i}(d)$.
\end{Assumption}
In repeated cross section settings with a binary treatment, our target parameter is the $CLATT_{e,t}$ in period $T_i=t$:
\begin{align*}
CLATT_{e,t} \equiv \E[Y_{i}(1)-Y_{i}(0)|T_i=t,E_i=e,CM_{e,t}],
\end{align*}
where $CM_{e,t} = \{D_{i,t}^{e} > D_{i,t}^{\infty}\}$. Our estimand is the Wald-DID estimand:
\begin{align*}
    w^{DID,r}_{e,t}
    =
    \frac{\Delta^{t}_{Y,E=e}-\Delta^{t}_{Y,U}}
    {\Delta^{t}_{D,E=e}-\Delta^{t}_{D,U}},
\end{align*}
for $e \in \mathcal{E}$ and $t \in \{e,\dots,T\}$, where $U \subseteq \mathcal{S}(E_i)$ denotes the control set and, for any $e \in \mathcal{S}(E_i)$ and any $U$,
\begin{align*}
    \Delta^{t}_{Y,E=e}
    &=
    \E[Y_{i}\mid E_i=e,T_i=t]-\E[Y_{i}\mid E_i=e,T_i=e-1],\\
    \Delta^{t}_{Y,U}
    &=
    \E[Y_{i}\mid E_i \in U,T_i=t]-\E[Y_{i}\mid E_i \in U,T_i=e-1],
\end{align*}
with $\Delta^{t}_{D,E=e}$ and $\Delta^{t}_{D,U}$ defined analogously. We consider two choices of the control set. When the control group is the never exposed cohort, we set $U=\{\infty\}$ and assume $\bar{e}=\infty$. When the control group is the not-yet-exposed cohorts, we set $U=\{e' \in \mathcal{S}(E_i): e'>t\}$ in each period $t$, that is, the units with $Z_{i,t}=0$; in this case, the estimand is defined for $t < \bar{e}$.\par
We now present the identifying assumptions under which the estimand captures the corresponding $CLATT_{e,t}$. 
\begin{Assumption}[Monotonicity assumption in repeated cross sections]
\label{ApeB2as1}
For each $e\in \mathcal{E}$ and $t \in \{e,\dots,T\}$, $\Prb(D_{i,t}^{e} \geq D_{i,t}^{\infty})=1$ or $\Prb(D_{i,t}^{e} \leq D_{i,t}^{\infty})=1$.
\end{Assumption}
As in Section \ref{sec2.3}, we hereafter consider the type of the monotonicity assumption that rules out the defiers.
\begin{Assumption}[No anticipation in the first stage in repeated cross sections]
\label{ApeB2as2}
For each $e\in \mathcal{S}(E_i)\setminus\{\infty\}$ and $t \in \{1,\dots,T\}$ such that $t < e$, $D_{i,t}^{e}=D_{i,t}^{\infty}$ a.s. for all units $i$ with $E_i=e$.
\end{Assumption}
\begin{Assumption}[Parallel trends assumption in the treatment]
\label{secA.2.2as1}
For each $e\in \mathcal{E}$ and $t\in \{2,\dots,T\}$ such that $t \geq e$,
\begin{gather*}
\E[D_{i}^{\infty}|E_i=e,T_i=t]-\E[D_{i}^{\infty}|E_i=e,T_i=t-1]\\
=\\
\E[D_{i}^{\infty}|E_i=\infty,T_i=t]-\E[D_{i}^{\infty}|E_i=\infty,T_i=t-1].
\end{gather*}
\end{Assumption}
\begin{Assumption}[Parallel trends assumption in the outcome]
\label{secA.2.2as2}
For each $e\in \mathcal{E}$ and $t\in \{2,\dots,T\}$ such that $t \geq e$,
\begin{gather*}
\E[Y_i(D_i^{\infty})|E_i=e,T_i=t]-\E[Y_i(D_i^{\infty})|E_i=e,T_i=t-1]\\
=\\
\E[Y_i(D_i^{\infty})|E_i=\infty,T_i=t]-\E[Y_i(D_i^{\infty})|E_i=\infty,T_i=t-1].
\end{gather*}
\end{Assumption}
\begin{Assumption}[Relevance condition]
\label{secA.2.2as3}
For each $e \in \mathcal{E}$ and $t\in \{e,\dots,T\}$, $\Delta^{t}_{D,E=e}-\Delta^{t}_{D,U} > 0$ with $U=\{\infty\}$.
\end{Assumption}
\begin{Assumption}[Parallel trends assumption in the treatment based on not-yet-exposed cohorts]
\label{secA.2.2as1ny}
For each $e\in \mathcal{E}$ and each $(s,t) \in \{2,\dots,T\} \times \{2,\dots,T\}$ such that $e \leq t \leq s < \bar{e}$,
\begin{gather*}
\E[D_{i}^{\infty}|E_i=e,T_i=t]-\E[D_{i}^{\infty}|E_i=e,T_i=t-1]\\
=\\
\E[D_{i}^{\infty}|Z_{i,s}=0,T_i=t]-\E[D_{i}^{\infty}|Z_{i,s}=0,T_i=t-1].
\end{gather*}
\end{Assumption}
\begin{Assumption}[Parallel trends assumption in the outcome based on not-yet-exposed cohorts]
\label{secA.2.2as2ny}
For each $e\in \mathcal{E}$ and each $(s,t) \in \{2,\dots,T\} \times \{2,\dots,T\}$ such that $e \leq t \leq s < \bar{e}$,
\begin{gather*}
\E[Y_i(D_i^{\infty})|E_i=e,T_i=t]-\E[Y_i(D_i^{\infty})|E_i=e,T_i=t-1]\\
=\\
\E[Y_i(D_i^{\infty})|Z_{i,s}=0,T_i=t]-\E[Y_i(D_i^{\infty})|Z_{i,s}=0,T_i=t-1].
\end{gather*}
\end{Assumption}
\begin{Assumption}[Relevance condition based on not-yet-exposed cohorts]
\label{secA.2.2as3ny}
For each $e \in \mathcal{E}$ and $t\in \{e,\dots,T\}$ such that $t < \bar{e}$, $\Delta^{t}_{D,E=e}-\Delta^{t}_{D,U} > 0$ with $U=\{e' \in \mathcal{S}(E_i): e'>t\}$.
\end{Assumption}
\begin{Theorem}
\label{ApeB2thm}
Suppose Assumption \ref{sec3as1} in Section \ref{sec2} and Assumptions \ref{ApeB2as0}--\ref{ApeB2as2} hold.
\begin{itemize}
\item[(i)] Suppose Assumptions \ref{secA.2.2as1}--\ref{secA.2.2as3} hold and $U=\{\infty\}$. Then, for each $e \in \mathcal{E}$ and $t \in \{e,\dots,T\}$, we have $w^{DID,r}_{e,t}=CLATT_{e,t}$.
\item[(ii)] Suppose Assumptions \ref{secA.2.2as1ny}--\ref{secA.2.2as3ny} hold and $U=\{e' \in \mathcal{S}(E_i): e'>t\}$. Then, for each $e \in \mathcal{E}$ and $t \in \{e,\dots,T\}$ such that $t < \bar{e}$, we have $w^{DID,r}_{e,t}=CLATT_{e,t}$.
\end{itemize}
\end{Theorem}
The proof is analogous to that of Theorem \ref{sec3.3.theorem1}, with all the arguments conditioning on the observation period $T_i$; we therefore omit it.\par
\begin{Remark}
The results in this section extend to an ordered, non-binary treatment
$D_{i} \in \{0,1,\dots,J\}$ and to the limited anticipation setting in
Online Appendix Section \ref{ApeB0}, as in the panel case. In the former
case, our target parameter is the $CACRT_{e,t}$, defined as in Section
\ref{sec3} with all the expectations and probabilities conditional on
$T_i=t$; Theorem \ref{ApeB2thm} continues to hold with $CLATT_{e,t}$
replaced by $CACRT_{e,t}$, by the same arguments as in the proof of
Theorem \ref{OnlineA.1.thm2}. In the latter case, Theorem \ref{ApeB2thm}
continues to hold with the ranges of $t$ adjusted as in Theorem
\ref{sec4_anticipation_theorem1}.
\end{Remark}
Finally, we establish the counterpart of Theorem \ref{sec4.1thm1} in this setting. In repeated cross section settings, our STS estimator $\hat{\theta}_{e,t}^{r}$ is the sample analog of the Wald-DID estimand $w^{DID,r}_{e,t}$. Let $CLATT_{t \geq e}$ and $\hat{\theta}^{r}_{t \geq e}$ denote the vectors that stack $CLATT_{e,t}$ and $\hat{\theta}^{r}_{e,t}$ over $(e,t) \in \mathcal{I}_{+}$ when $U=\{\infty\}$, and over $(e,t) \in \mathcal{I}_{+}^{ny}$ when $U=\{e' \in \mathcal{S}(E_i): e'>t\}$.
\begin{Theorem}
\label{ApeAthm}
(i) Suppose Assumption \ref{sec3as1}, Assumptions \ref{ApeB2as0}--\ref{ApeB2as2}, and Assumptions \ref{secA.2.2as1}--\ref{secA.2.2as3} hold with $U=\{\infty\}$. Assume also that \(\E[Y_{i}^2]<\infty\) and \(\E[D_{i}^2]<\infty\). Then, for each $(e,t) \in \mathcal{I}_{+}$, we have:
\begin{align*}
\sqrt{N}(\hat{\theta}^{r}_{e,t}-CLATT_{e,t})=\frac{1}{\sqrt{N}}\sum_{i=1}^{N}\psi^{r}_{i,e,t}+o_p(1),
\end{align*}
where the influence function $\psi^{r}_{i,e,t}$ is given in equation \eqref{ApeAinf_repeated} below. Furthermore, as $N \rightarrow \infty$, we have:
\begin{align*}
\sqrt{N}(\hat{\theta}^{r}_{t \geq e}-CLATT_{t \geq e}) \xrightarrow{d} \mathcal{N}(0,\Sigma^{r}),
\end{align*}
where $\Sigma^{r}=\E[\Psi^{r}_{i,t \geq e}\Psi^{r\prime}_{i,t \geq e}]$ and $\Psi^{r}_{i,t \geq e}$ is the vector that stacks $\psi^{r}_{i,e,t}$ over $(e,t) \in \mathcal{I}_{+}$.\par
(ii) Suppose Assumption \ref{sec3as1}, Assumptions \ref{ApeB2as0}--\ref{ApeB2as2}, and the assumptions of Theorem \ref{ApeB2thm}(ii) hold, $U=\{e' \in \mathcal{S}(E_i): e'>t\}$, and the moment conditions in (i) hold. Then, the same results hold with $\mathcal{I}_{+}$ replaced by $\mathcal{I}_{+}^{ny}$.
\begin{proof}
The proof is analogous to that of Theorem \ref{sec4.1thm1} and therefore omitted.
\end{proof}
\end{Theorem}
The influence function $\psi^{r}_{i,e,t}$ is:
\footnotesize
\begin{align}
\label{ApeAinf_repeated}
&\psi_{i,e,t}^{r}=\frac{1}{\Delta^{t}_{D,E=e}-\Delta^{t}_{D,U}}\notag\\
&\times\Bigg(\frac{\mathbf{1}_{e,t}\cdot\left[\delta^{r}_{i,e,t}-\E[\delta^{r}_{i,e,t}|E_i=e,T_i=t]\right]}{\Prb(E_i=e,T_i=t)}-\frac{\mathbf{1}_{e,e-1}\cdot\left[\delta^{r}_{i,e,t}-\E[\delta^{r}_{i,e,t}|E_i=e,T_i=e-1]\right]}{\Prb(E_i=e,T_i=e-1)}\notag\\&-\frac{\mathbf{1}_{U,t}\cdot\left[\delta^{r}_{i,e,t}-\E[\delta^{r}_{i,e,t}|E_i \in U,T_i=t]\right]}{\Prb(E_i \in U,T_i=t)}+\frac{\mathbf{1}_{U,e-1}\cdot\left[\delta^{r}_{i,e,t}-\E[\delta^{r}_{i,e,t}|E_i \in U,T_i=e-1]\right]}{\Prb(E_i \in U,T_i=e-1)}\Bigg),
\end{align}
\normalsize
where $\mathbf{1}_{A}$ is the indicator function and takes one if $A$ is true, and $\delta_{i,e,t}^{r}=Y_i-w^{DID,r}_{e,t}\cdot D_i$.\par

\subsection{Saturated STS regressions}\label{OnlineApeSat}
In this section, we first prove Theorem \ref{theorem8_saturate} in the main text. We then establish the counterpart of Theorem \ref{sec4.1thm1} for the saturated STS regression.\par
\paragraph{Proof of Theorem \ref{theorem8_saturate}}
\begin{proof}
Throughout the proof, let $S_i=\1\{E_i \in \{e\} \cup U\}$ as in Theorem \ref{theorem8_saturate}. By the FWL theorem, we have $\beta^{e,t}
    =\frac{\E[R_{i,s}Y_{i,s}\mid S_i=1]/\E[R_{i,s}^{2}\mid S_i=1]}{\E[R_{i,s}D_{i,s}\mid S_i=1]/\E[R_{i,s}^{2}\mid S_i=1]}.$
We first show that, for a generic random variable $A_{i,s}$ with $s \in \{e-1,t\}$,
\begin{align}
\label{ApeB4generic}
  \frac{\E[R_{i,s} A_{i,s}\mid S_i=1]}{\E[R_{i,s}^2\mid S_i=1]}
  =
  \frac{\sum_{x\in\mathcal X}\Prb(X_i=x \mid S_i=1)\,\omega_x\Delta_A(x)}
       {\sum_{x\in\mathcal X}\Prb(X_i=x \mid S_i=1)\,\omega_x}.
\end{align}
Conditioning on $X_i$, we have:
\begin{align*}
        \frac{\E[R_{i,s} A_{i,s}\mid S_i=1]}{\E[R_{i,s}^2\mid S_i=1]}=\frac{\sum_{x\in\mathcal X}\Prb(X_i=x \mid S_i=1)\,\omega_x\frac{\E[R_{i,s} A_{i,s}\mid X_i=x, S_i=1]}{\E[R_{i,s}^2\mid X_i=x, S_i=1]}}{\sum_{x\in\mathcal X}\Prb(X_i=x \mid S_i=1)\,\omega_x}.
    \end{align*}\par
Next, we show $\frac{\E[R_{i,s} A_{i,s}\mid X_i=x, S_i=1]}{\E[R_{i,s}^2\mid X_i=x, S_i=1]}=\Delta_A(x)$. Let $R_{i,s}^{x}$ be the residual from regressing $\1\{E_i=e\}\1\{s=t\}$ on $(1,\1\{E_i=e\},\1\{s=t\})$ conditional on $X_i=x$ and $S_i=1$. Since the controls are saturated in $X_i$, we have $R_{i,s}=R_{i,s}^{x}$ on $X_i=x$. Therefore, $\frac{\E[R_{i,s} A_{i,s}\mid X_i=x, S_i=1]}{\E[R_{i,s}^2\mid X_i=x, S_i=1]}$ is the coefficient $\tau_A(x)$ on the interaction $\1\{E_i=e\}\1\{s=t\}$ in the regression
\begin{align*}
    A_{i,s}=\alpha_x+\gamma_x \1\{E_i=e\}+\lambda_x \1\{s=t\}+\tau_A(x)\,\1\{E_i=e\}\1\{s=t\} +e_{i,s},
\end{align*}
conditional on $X_i=x$ and $S_i=1$. This is the DID regression with two periods and two groups conditional on $X_i=x$ and $S_i=1$, and hence $\tau_A(x)=\Delta_A(x)$.\par
Applying \eqref{ApeB4generic} to $A_{i,s}=Y_{i,s}$ and $A_{i,s}=D_{i,s}$, together with $\Delta_Y(x)=\Delta_D(x)\beta(x)$, we obtain the desired result.
\end{proof}
Next, we establish the asymptotic properties of the saturated STS estimator $\hat{\beta}^{e,t}$. Let $\beta_{t \geq e}$ and $\hat{\beta}_{t \geq e}$ denote the vectors that stack $\beta^{e,t}$ and $\hat{\beta}^{e,t}$ over $(e,t) \in \mathcal{I}_{+}$ when $U=\{\infty\}$, and over $(e,t) \in \mathcal{I}_{+}^{ny}$ when $U=\{e' \in \mathcal{S}(E_i): e'>t\}$.
\begin{Theorem}
\label{OnlineSatThm}
Suppose the assumptions of Theorem \ref{theorem8_saturate} hold for each $(e,t) \in \mathcal{I}_{+}$ and $U=\{\infty\}$. Assume also that \(\E[Y_{i,s}^2]<\infty\) and \(\E[D_{i,s}^2]<\infty\) for all $s \in \{1,\dots,T\}$. Then, for each $(e,t) \in \mathcal{I}_{+}$, we have $\sqrt{N}(\hat{\beta}^{e,t}-\beta^{e,t})=\frac{1}{\sqrt{N}}\sum_{i=1}^{N}\psi^{sat}_{i,e,t}+o_p(1),$ where the influence function $\psi^{sat}_{i,e,t}$ is given below.\par 
Furthermore, as $N \rightarrow \infty$, we have $\sqrt{N}(\hat{\beta}_{t \geq e}-\beta_{t \geq e}) \xrightarrow{d} \mathcal{N}(0,\Sigma^{sat}),$
where $\Sigma^{sat}=\E[\Psi^{sat}_{i,t \geq e}\Psi^{sat\prime}_{i,t \geq e}]$ and $\Psi^{sat}_{i,t \geq e}$ is the vector that stacks $\psi^{sat}_{i,e,t}$ over $(e,t) \in \mathcal{I}_{+}$. The same results hold with $\mathcal{I}_{+}$ replaced by $\mathcal{I}_{+}^{ny}$.
\begin{proof}
The proof is analogous to that of Theorem \ref{sec4.1thm1} and therefore omitted.
\end{proof}
\end{Theorem}
The influence function $\psi^{sat}_{i,e,t}$ is:
\begin{align*}
\psi^{sat}_{i,e,t}=\frac{S_i\sum_{s \in \{e-1,t\}} R_{i,s}\,u_{i,s}}{\E\left[S_i\sum_{s \in \{e-1,t\}}R_{i,s}D_{i,s}\right]},
\end{align*}
where $S_i$ and $R_{i,s}$ are defined as in Theorem \ref{theorem8_saturate} and $u_{i,s}$ is the error term in the saturated IV regression.

\section{Comparing DID-IV with Fuzzy DID}\label{ApeC}
This section provides the detailed discussion and proofs underlying
Section~\ref{sec6} of the main text, which compares DID-IV with the
Fuzzy DID proposed in \citeonline{De_Chaisemartin2018-xe} (henceforth,
``dCDH''). We first introduce the $2 \times 2$
DID-IV design under repeated cross section settings and clarify the
similarities and differences between DID-IV and Fuzzy DID. We then compare the restrictions
on treatment adoption behavior and the target
parameters under the two designs. We next explain
why the common Wald-DID estimand requires the stable treatment effect
assumption under Fuzzy DID but not under DID-IV,
and discuss which parallel trends assumption is suitable for the DID-IV
setting. Finally, we extend the comparison to
non-binary, ordered treatments and discuss the alternative estimands
proposed by dCDH.

\subsection{Preparation: $2 \times 2$ DID-IV designs}\label{ApeC.1}
In preparation for the comparison, we first introduce the $2 \times 2$ DID-IV design considered in \citeonline{chasemartin2010-ch} and \citeonline{Hudson2017-tm}.\par
Let $T \in \{0,1\}$ denote the time period in which a unit is observed. Let $Y$ and $D \in \{0,1\}$ denote the outcome and the treatment. Let $Z_{t} \in \{0,1\}$ denote the instrument status in period $t$, and let $Z=(Z_{0},Z_{1})$ denote the instrument path. We assume that the instrument is assigned only in the second period: $Z_{0}=0$ for all units. We therefore write $Z \equiv Z_{1}$ in what follows: $Z=1$ if a unit is exposed to the instrument in the second period and $Z=0$ otherwise. We assume that $Z$ is known to researchers even for units observed in period $0$.\par 
For all $z \in \{0,1\}$, let $Y(0,z)$, $Y(1,z)$, and $D(z)$ denote the potential outcomes and the potential treatment status. For all $t \in \{0,1\}$, let $D_{t}(z)$ denote the potential treatment status in period $t$.\footnote{For example, when $T$ represents the birth cohort, $D_{t}(z)$ is the treatment status of a unit had it been born in cohort $t$.} Note that we have $D(z)=(1-T)D_{0}(z)+TD_{1}(z)$.\par
In the two-period, two-group setting, our target parameter is the LATET in period $1$:
\begin{align*}
LATET \equiv \E[Y(1)-Y(0)|Z=1,T=1,CM^{Z}],
\end{align*}
where $CM^{Z} = \{D_{1}(1) > D_{1}(0)\}$.
Our estimand is the Wald-DID estimand:
\begin{align*}
w^{DID}=\frac{\Delta_{Y,Z=1}-\Delta_{Y,Z=0}}{\Delta_{D,Z=1}-\Delta_{D,Z=0}},
\end{align*}
where $\Delta_{Y,Z=z}=\E[Y|Z=z,T=1]-\E[Y|Z=z,T=0]$ and $\Delta_{D,Z=z}$ is defined analogously.\par
\citeonline{chasemartin2010-ch} shows that the Wald-DID estimand identifies the LATET under the following assumptions (\citeonline{Hudson2017-tm} consider a non-binary, ordered treatment).\footnote{More precisely, \citeonline{chasemartin2010-ch} and \citeonline{Hudson2017-tm} implicitly impose the no anticipation assumption (Assumption~\ref{ApeBas3}) by writing the observed treatment choice in
period $0$ as the potential treatment under non-exposure. In addition, while \citeonline{chasemartin2010-ch} imposes the monotonicity on $D(z)$, we impose it on $D_{1}(z)$; both assumptions are sufficient for identifying the LATET.}

\begin{Assumption}[Exclusion restriction]
\label{ApeBas1}
For all $z \in \{0,1\}$ and $d \in \{0,1\}$, $Y(d,z)=Y(d)$.
\end{Assumption}

\begin{Assumption}[IV Monotonicity]
\label{ApeBas2}
$\Prb(D_{1}(1) \geq D_{1}(0))=1$ or $\Prb(D_{1}(1) \leq D_{1}(0))=1$.
\end{Assumption}
Here, we consider the monotonicity assumption that rules out the defiers.
\begin{Assumption}[No anticipation in the first stage]
\label{ApeBas3}
$D_{0}(1)=D_{0}(0)$ a.s.\ conditional on $Z=1$.
\end{Assumption}

\begin{Assumption}[Parallel trends assumption in the treatment]
\label{ApeBas5}
\begin{align*}
&\E[D(0)|Z=1,T=1]-\E[D(0)|Z=1,T=0]\\
=&\E[D(0)|Z=0,T=1]-\E[D(0)|Z=0,T=0].
\end{align*}
\end{Assumption}

\begin{Assumption}[Parallel trends assumption in the outcome]
\label{ApeBas6}
\begin{align*}
&\E[Y(D(0))|Z=1,T=1]-\E[Y(D(0))|Z=1,T=0]\\
=&\E[Y(D(0))|Z=0,T=1]-\E[Y(D(0))|Z=0,T=0].
\end{align*}
\end{Assumption}

\begin{Assumption}[Relevance condition]
\label{ApeBas7}
$\Delta_{D,Z=1}-\Delta_{D,Z=0} > 0$.
\end{Assumption}

\subsection{Similarities and differences between DID-IV and Fuzzy DID}\label{ApeC.2}
In this section, we first clarify the similarities and differences between DID-IV and Fuzzy DID with respect to the setting, the target parameter, the estimand, and the identifying assumptions.
\subsubsection*{The setting}\label{ApeC.2.1}
The setting introduced in Section~\ref{ApeC.1} has two key features: (i)
the instrument is assigned only in the second period (sharp assignment of
the instrument), and (ii) units are partitioned into two groups based on exposure to the instrument. We refer to this setup as the ``DID-IV'' setting.\par
Under this definition, dCDH also consider the DID-IV setting.\footnote{Here, it is useful to clarify the distinction between the setting and the design in dCDH. In dCDH, the Fuzzy DID design is described as a situation in which the treatment rate in one group increases more than that in another group (see Assumption 1 in dCDH). However, this characterization corresponds to an identifying assumption, rather than a primitive feature of the empirical setting. Independently of Assumption 1, dCDH also define the Fuzzy DID setting as $D \neq G \times T$. This definition corresponds to the DID-IV setting once combined with the restriction $v_{10}=v_{00}$.} In particular, the group variable $G$ in dCDH plays the role of the instrument $Z$ in DID-IV designs. This follows from the fact that $G$ is included in their treatment participation equation $D = \mathbf{1}\{V \geq v_{GT}\}$ (see Assumption~3 in dCDH), which corresponds to condition (ii) of the DID-IV setting. Condition (i) is also assumed in dCDH: they impose $v_{10} = v_{00}$ throughout their proofs (see page~1019 of dCDH), implying that the treatment participation thresholds at time $T=0$ are identical across groups $G=0$ and $G=1$.\par
Based on the above discussion, in what follows, we replace $G$ with $Z$ and refer to the groups $Z=1$ and $Z=0$ as the exposed and unexposed groups, respectively.
\subsubsection*{The target parameter}\label{ApeC.2.2}
In DID-IV designs, our target parameter is the LATET. By contrast, the target parameter in Fuzzy DID is the switcher local average treatment effect on the treated (SLATET):
\begin{align*}
SLATET \equiv \E[Y(1)-Y(0)|Z=1,T=1,SW],
\end{align*}
where $SW=\{D_{0}(0) < D_{1}(1)\}$. Here, $D_t(Z) = \mathbf{1}\{V \geq v_{Zt}\}$ denotes the treatment status
at time $T=t$ for group $Z$ ($D(t)$ in the notation of dCDH), and we use
the restriction $v_{10}=v_{00}$ in dCDH, which implies
$D_{0}(1)=D_{0}(0)$.\par
This parameter measures the treatment effects for units who belong to the exposed group ($Z=1$) and switch into treatment at time $T=1$ (switchers). In Section~\ref{ApeC.4}, we clarify the differences between the SLATET and the LATET.
\subsubsection*{The estimand}\label{ApeC.2.3}
Both designs consider the same Wald-DID estimand $w^{DID}$ defined in Section~\ref{ApeC.1}. However, the role of the Wald-DID estimand differs across the two designs. In DID-IV designs, we adopt it as the natural estimand for identifying the LATET. By contrast, dCDH point out that the Wald-DID estimand identifies the SLATET only if the treatment effect is stable over time, and propose alternative estimands that avoid this restriction. We return to this point in Section~\ref{ApeC.5}.
\subsubsection*{The identifying assumptions}\label{ApeC.2.4}
The Fuzzy DID design imposes the following identifying assumptions for the Wald-DID estimand to identify the SLATET. For detailed discussions of each assumption, see dCDH.
\begin{Assumption}[Fuzzy design]
\label{chas1}
\begin{align}
\label{cheq1}
&\E[D|Z=1,T=1]>\E[D|Z=1,T=0],\\
\label{cheq2}
&\E[D|Z=1,T=1]-\E[D|Z=1,T=0]>\E[D|Z=0,T=1]-\E[D|Z=0,T=0].
\end{align}
\end{Assumption}

\begin{Assumption}[Stable percentage of treated units in an unexposed group]
\label{chas2}
$0<\E[D|Z=0,T=1]=\E[D|Z=0,T=0]<1$.
\end{Assumption}

\begin{Assumption}[Treatment participation equation]
\label{chas3}
$D=\mathbf{1}\{V \geq v_{ZT}\}$, with $V \indep T|Z$.
\end{Assumption}

\begin{Assumption}[Parallel trends assumption in untreated outcomes]
\label{chas4}
\begin{align*}
&\E[Y(0)|Z=1,T=1]-\E[Y(0)|Z=1,T=0]\\
=&\E[Y(0)|Z=0,T=1]-\E[Y(0)|Z=0,T=0].
\end{align*}
\end{Assumption}

\begin{Assumption}[Stable treatment effect over time]
\label{chas5}
$\E[Y(1)-Y(0)|Z,T=1,D_0(Z)=1]=\E[Y(1)-Y(0)|Z,T=0,D_0(Z)=1]$.
\end{Assumption}
We now state the detailed connections between the identifying assumptions in Fuzzy DID and those in DID-IV. First, both designs impose condition~\eqref{cheq2} in Assumption~\ref{chas1}, which corresponds to the standard relevance condition (Assumption~\ref{ApeBas7}). Next, both designs impose Assumption~\ref{ApeBas2}: Assumption~\ref{chas3} (treatment participation equation) implies Assumption~\ref{ApeBas2} (\citeonline{Vytlacil2002-fj}).\par
Moreover, dCDH implicitly impose Assumptions~\ref{ApeBas1} and~\ref{ApeBas3}. Specifically, dCDH define potential outcomes as $Y(D)$, which implies Assumption~\ref{ApeBas1}, and their treatment participation equation depends only on the current instrument status (Assumption~\ref{chas3}), which, combined with the restriction $v_{10}=v_{00}$, implies Assumption~\ref{ApeBas3}.\par
The differences between the identifying assumptions in Fuzzy DID and those in DID-IV are as follows. First, Fuzzy DID imposes the parallel trends assumption in untreated outcomes (Assumption~\ref{chas4}), whereas DID-IV imposes the parallel trends assumption in unexposed outcomes (Assumption~\ref{ApeBas6}).\footnote{\citeonline{dechaisemartin2024differenceindifferencesestimatorstreatmentscontinuously} note that the parallel trends assumption in unexposed outcomes imposes restrictions on treatment effect heterogeneity when the standard parallel trends assumption holds between exposed and unexposed groups.} Second, Fuzzy DID assumes stable treatment rates in the unexposed group (Assumption~\ref{chas2}), whereas DID-IV imposes the parallel trends assumption in the treatment between exposed and unexposed groups (Assumption~\ref{ApeBas5}). Third, Fuzzy DID imposes an independence assumption between unobserved heterogeneity $V$ and time $T$ conditional on $Z$ in Assumption~\ref{chas3}, whereas DID-IV does not impose such a restriction. Fourth, the treatment participation equation in Assumption~\ref{chas3} is stronger than Assumption~\ref{ApeBas2} in DID-IV: it imposes the monotonicity with respect to time $T$, in addition to the monotonicity with respect to instrument $Z$ (see Lemma~\ref{ApeC.3.3.lemma1} in Section~\ref{ApeC.3}).
Finally, dCDH impose condition~\eqref{cheq1} in Assumption~\ref{chas1}, which requires that the treatment rate in the exposed group increases between the two periods, whereas DID-IV does not impose this condition.\par
\begin{Remark}
When Assumption~\ref{chas2} is violated, dCDH consider the following assumption for the Wald-DID estimand to capture the SLATET.
\begin{Assumption}[Homogeneous treatment effect between exposed and unexposed groups]
\label{Homogeneous1}
$SLATET=SLATET'$, where $SLATET'=\E[Y(1)-Y(0)|Z=0,T=1,D_{0}(0) \neq D_{1}(0)]$.
\end{Assumption}
This assumption requires the treatment effects among the switchers to be the same between exposed and unexposed groups, which dCDH call the homogeneous treatment effect assumption. In earlier work, \citeonline{Blundell565} make a similar argument in Subsection~F.2 of Section~IV. In Section~\ref{ApeC.4}, we present the decomposition result for the SLATET, and show that we cannot straightforwardly interpret this assumption as requiring the homogeneous treatment effect across groups. In the following discussion, we therefore treat Assumptions~\ref{chas1}--\ref{chas5} as the main identifying assumptions in Fuzzy DID.\par
\end{Remark}
\subsection{Restrictions on treatment adoption behavior}\label{ApeC.3}
In this section, we compare the restrictions on treatment adoption behavior between the two designs. Since both designs impose Assumption \ref{ApeBas3}, we summarize the potential treatment choices by $(D_{0}(0),D_{1}(0), D_{1}(1))$ in the following.\par
To characterize the heterogeneity in treatment adoption behavior, we introduce two group variables. Let $G^{Z}=(D_{1}(0),D_{1}(1))$ denote the group variable that represents the treatment adoption behavior in period $1$ in response to the instrument. We define $G^{Z}=(0,0) \equiv NT^{Z}$ as never-takers, $G^{Z}=(1,1) \equiv AT^{Z}$ as always-takers, $G^{Z}=(0,1) \equiv CM^{Z}$ as compliers, and $G^{Z}=(1,0) \equiv DF^{Z}$ as defiers. Similarly, let $G^{T}=(D_{0}(0),D_{1}(0))$ denote the group variable that represents the treatment adoption process across the two periods when the instrument is $Z=0$ in the second period. We define $G^{T}=(0,0) \equiv NT^{T}$ as time never-takers, $G^{T}=(1,1) \equiv AT^{T}$ as time always-takers, $G^{T}=(0,1) \equiv CM^{T}$ as time compliers, and $G^{T}=(1,0) \equiv DF^{T}$ as time defiers. Using the group variables $G^{Z}$ and $G^{T}$, we can partition units into eight mutually exclusive and exhaustive latent types within each group (see Tables~\ref{ApeC.3.2.table1}--\ref{ApeC.3.2.table2}). In the following, we examine which types are excluded by the identifying assumptions under DID-IV and Fuzzy DID, respectively.

\subsubsection{DID-IV}\label{ApeC.3.2}
In DID-IV designs, the only restriction on treatment adoption behavior is the monotonicity
assumption (Assumption~\ref{ApeBas2}),
which rules out either $DF^{Z}$ or $CM^{Z}$.
If we rule out $DF^{Z}$, we obtain the following tables.
\begin{table}[H]
\centering
\renewcommand{\arraystretch}{1.2}\setlength{\tabcolsep}{3pt}
\begin{minipage}[t]{0.48\textwidth}
\centering \footnotesize
\caption{Exposed group ($z=1$)}
\label{ApeC.3.2.table1}
\begin{tabular}{ccc}
\hline \hline
observed & \multicolumn{2}{c}{counterfactual} \\
$(D_0(0),\, D_1(1))$ & $D_1(0)=1$ & $D_1(0)=0$ \\
\hline
$(1,1)$ & $AT^Z\!\land\! AT^T$ & $CM^Z\!\land\! DF^T$ \\
$(1,0)$ & \cellcolor[gray]{0.8}$DF^Z\!\land\! AT^T$ & $NT^Z\!\land\! DF^T$ \\
$(0,1)$ & $AT^Z\!\land\! CM^T$ & $CM^Z\!\land\! NT^T$ \\
$(0,0)$ & \cellcolor[gray]{0.8}$DF^Z\!\land\! CM^T$ & $NT^Z\!\land\! NT^T$ \\
\hline
\end{tabular}
\end{minipage}\hfill
\begin{minipage}[t]{0.48\textwidth}
\centering \footnotesize
\caption{Unexposed group ($z=0$)}
\label{ApeC.3.2.table2}
\begin{tabular}{ccc}
\hline \hline
observed & \multicolumn{2}{c}{counterfactual} \\
$(D_0(0),\, D_1(0))$ & $D_1(1)=1$ & $D_1(1)=0$ \\
\hline
$(1,1)$ & $AT^Z\!\land\! AT^T$ & \cellcolor[gray]{0.8}$DF^Z\!\land\! AT^T$ \\
$(1,0)$ & $CM^Z\!\land\! DF^T$ & $NT^Z\!\land\! DF^T$ \\
$(0,1)$ & $AT^Z\!\land\! CM^T$ & \cellcolor[gray]{0.8}$DF^Z\!\land\! CM^T$ \\
$(0,0)$ & $CM^Z\!\land\! NT^T$ & $NT^Z\!\land\! NT^T$ \\
\hline
\end{tabular}
\end{minipage}
\\[5pt]
\begin{minipage}{0.95\textwidth}
\footnotesize
\textit{Notes}: These tables represent mutually exclusive and exhaustive types under DID-IV designs. The types painted in gray color are excluded by Assumption~\ref{ApeBas2}.
\end{minipage}
\end{table}\par

\subsubsection{Fuzzy DID}\label{ApeC.3.3}
Next, we consider the treatment adoption behavior under Fuzzy DID. We first prove the following lemma.
\begin{Lemma}[Monotonicity conditions implied by Fuzzy DID]\label{ApeC.3.3.lemma1}
Assumption~\ref{chas3} implies the IV monotonicity
(Assumption~\ref{ApeBas2}) and the following time monotonicity:
\begin{Assumption}[Time monotonicity for each $z \in \{0,1\}$]
\label{as33}
$\Prb(D_0(z) \geq D_1(z))=1$ or $\Prb(D_0(z) \leq D_1(z))=1$ for each $z \in \{0,1\}$.
\end{Assumption}
\end{Lemma}
\begin{proof}
Comparing $v_{11}$ with $v_{01}$ yields Assumption~\ref{ApeBas2}. Similarly,
comparing $v_{z0}$ with $v_{z1}$ for each $z\in\{0,1\}$ yields
Assumption~\ref{as33}.
\end{proof}\par
Next, we prove the following lemma.
\begin{Lemma}[Restrictions implied by Fuzzy DID]\label{ApeC.3.3.lemma2}
The following results hold.
\begin{itemize}
\item[(i)] Assumptions \ref{chas2} and \ref{chas3} rule out both time compliers $CM^{T}$ and time defiers $DF^{T}$ in an unexposed group $Z=0$.
\item[(ii)] The condition \eqref{cheq1} in Assumption \ref{chas1} and Assumption \ref{chas3} imply $\Prb(D_0(0) \leq D_1(1)|Z=1)=1$ and $\Prb(D_0(0) < D_1(1)|Z=1) > 0$.\footnote{Note that the latter condition ensures the existence of the switchers in an exposed group; that is, the SLATET is well defined.}
\item[(iii)] The condition \eqref{cheq1} in Assumption \ref{chas1} and Assumption \ref{chas3} imply $\Prb(D_0(0) \leq D_1(1)|Z=0)=1$. In particular, the type $(DF^{Z}\land AT^{T})$ is ruled out in an unexposed group.
\end{itemize}
\end{Lemma}
\begin{proof}
(i) Note that Assumption \ref{chas3} also implies the time monotonicity given $Z=0$:
\begin{align*}
\Prb(D_0(0) \geq D_1(0)|Z=0)=1\hspace{2mm}\text{or}\hspace{2mm}\Prb(D_0(0) \leq D_1(0)|Z=0)=1.
\end{align*}
This implies that at most one of the two switching types ($CM^T$ or $DF^T$) can exist in an unexposed group $Z=0$. Without loss of generality, suppose that $DF^T$ is ruled out in the unexposed group.\par
Then, Assumption \ref{chas2} implies:
\begin{align}
&\E[D|T=1,Z=0]=\E[D|T=0,Z=0]\notag\\
\iff & \Prb(D_1(0) > D_0(0)|Z=0)=0\notag\\
\label{commenteq1}
\iff & \Prb(CM^T|Z=0)=0.
\end{align}
The first equivalence follows from $V \indep T|Z$ in Assumption \ref{chas3}. Equation \eqref{commenteq1} shows that there are no time compliers $CM^T$ in an unexposed group $Z=0$.\par
(ii) Note that Assumption \ref{chas3} also implies the time monotonicity given $Z=1$:
\begin{align*}
\Prb(D_0(1) \geq D_1(1)|Z=1)=1\hspace{2mm}\text{or}\hspace{2mm}\Prb(D_0(1) \leq D_1(1)|Z=1)=1.
\end{align*}
In what follows, we replace $D_0(1)$ with $D_0(0)$.\par
From the condition \eqref{cheq1} in Assumption \ref{chas1}, we have
\begin{align}
&\E[D|Z=1,T=1] > \E[D|Z=1,T=0]\notag\\
\iff & \E[D_1(1)-D_0(0)|Z=1] > 0\notag\\
\label{commentlemma3}
\iff & \Prb(D_1(1) > D_0(0)|Z=1)- \Prb(D_0(0) > D_1(1)|Z=1) > 0,
\end{align}
where the first equivalence follows from $V \indep T|Z$ in Assumption \ref{chas3} and $D_0(1)=D_0(0)$. Suppose that $\Prb(D_0(0) \geq D_1(1)|Z=1)=1$ holds. Then, we have $\Prb(D_1(1) > D_0(0)|Z=1)=0$. This contradicts the condition \eqref{commentlemma3}. Therefore, we should have $\Prb(D_0(0) \leq D_1(1)|Z=1)=1$. From the condition \eqref{commentlemma3}, we also have $\Prb(D_0(0) < D_1(1)|Z=1) > 0$.\par
(iii) By Assumption \ref{chas3}, we have either $v_{11} \leq v_{00}$ or $v_{11} > v_{00}$. In the latter case, $D_1(1) \leq D_0(0)$ holds almost surely, which contradicts $\Prb(D_0(0) < D_1(1)|Z=1) > 0$ in part (ii). Therefore, we should have $v_{11} \leq v_{00}$. Hence, $D_0(0) \leq D_1(1)$ holds almost surely, and therefore $\Prb(D_0(0) \leq D_1(1)|Z=0)=1$.\par
\end{proof}\par
 From Lemmas \ref{ApeC.3.3.lemma1} and \ref{ApeC.3.3.lemma2}, we obtain Tables \ref{ApeC.3.3.table1}--\ref{ApeC.3.3.table2}, which summarize the restrictions on treatment adoption behavior under Fuzzy DID. These tables show that four types are excluded in an exposed group and five types are excluded in an unexposed group.\footnote{If we exclude $DF^{Z}$ and $CM^{T}$ by Assumption \ref{ApeBas2} and Assumption \ref{as33} (under $z=0$), four types and five types are again excluded in exposed and unexposed groups, respectively. If we exclude $CM^{Z}$ and $DF^{T}$ by Assumptions \ref{ApeBas2} and \ref{as33} (under $z=0$), four types are excluded in an exposed group and six types are excluded in an unexposed group.} Here, we display the case in which Assumptions \ref{ApeBas2} and \ref{as33} (under $z=0$) rule out $DF^{Z}$ and $DF^{T}$, respectively.\footnote{Note that we cannot impose Assumptions \ref{ApeBas2} and \ref{as33} that rule out both $CM^{Z}$ and $CM^{T}$, which contradicts $\Prb(D_0(0) < D_1(1)|Z=1) > 0$ in Lemma \ref{ApeC.3.3.lemma2}(ii).} Since these restrictions depend on the direction of the two monotonicity conditions, we paint these types in blue. The remaining exclusions do not depend on the direction of the two monotonicity conditions, and we paint these types in gray: Lemma \ref{ApeC.3.3.lemma2}(ii) rules out $DF^{Z} \land AT^{T}$ and $NT^{Z} \land DF^{T}$ in an exposed group, and Lemma \ref{ApeC.3.3.lemma2}(i) and (iii) rule out $DF^{T}$, $CM^{T}$ and $DF^{Z} \land AT^{T}$ in an unexposed group.\par
 \begin{table}[H]
\centering
\renewcommand{\arraystretch}{1.2}\setlength{\tabcolsep}{3pt}
\begin{minipage}[t]{0.48\textwidth}
\centering \footnotesize
\caption{Exposed group ($z=1$)}
\label{ApeC.3.3.table1}
\begin{tabular}{ccc}
\hline \hline
observed & \multicolumn{2}{c}{counterfactual} \\
$(D_0(0),\, D_1(1))$ & $D_1(0)=1$ & $D_1(0)=0$ \\
\hline
$(1,1)$ & $AT^Z\!\land\! AT^T$ & \cellcolor{blue!20}$CM^Z\!\land\! DF^T$ \\
$(1,0)$ & \cellcolor[gray]{0.8}$DF^Z\!\land\! AT^T$ & \cellcolor[gray]{0.8}$NT^Z\!\land\! DF^T$ \\
$(0,1)$ & $AT^Z\!\land\! CM^T$ & $CM^Z\!\land\! NT^T$ \\
$(0,0)$ & \cellcolor{blue!20}$DF^Z\!\land\! CM^T$ & $NT^Z\!\land\! NT^T$ \\
\hline
\end{tabular}
\end{minipage}\hfill
\begin{minipage}[t]{0.48\textwidth}
\centering \footnotesize
\caption{Unexposed group ($z=0$)}
\label{ApeC.3.3.table2}
\begin{tabular}{ccc}
\hline \hline
observed & \multicolumn{2}{c}{counterfactual} \\
$(D_0(0),\, D_1(0))$ & $D_1(1)=1$ & $D_1(1)=0$ \\
\hline
$(1,1)$ & $AT^Z\!\land\! AT^T$ & \cellcolor[gray]{0.8}$DF^Z\!\land\! AT^T$ \\
$(1,0)$ & \cellcolor[gray]{0.8}$CM^Z\!\land\! DF^T$ & \cellcolor[gray]{0.8}$NT^Z\!\land\! DF^T$ \\
$(0,1)$ & \cellcolor[gray]{0.8}$AT^Z\!\land\! CM^T$ & \cellcolor[gray]{0.8}$DF^Z\!\land\! CM^T$ \\
$(0,0)$ & $CM^Z\!\land\! NT^T$ & $NT^Z\!\land\! NT^T$ \\
\hline
\end{tabular}
\end{minipage}
\\[5pt]
\begin{minipage}{0.95\textwidth}
\footnotesize
\textit{Notes}: These tables represent mutually exclusive and exhaustive types under Fuzzy DID. The types painted in gray color are excluded by Lemma \ref{ApeC.3.3.lemma2}. The types painted in blue color are excluded by Assumptions \ref{ApeBas2} and \ref{as33} (under $z=0$) when their directions rule out $DF^{Z}$ and $DF^{T}$, respectively.
\end{minipage}
\end{table}\par
\par

\subsubsection{Comparison}\label{ApeC.3.4}
By comparing Tables \ref{ApeC.3.2.table1}--\ref{ApeC.3.2.table2} with Tables \ref{ApeC.3.3.table1}--\ref{ApeC.3.3.table2}, we obtain two implications. First, the restrictions imposed under Fuzzy DID are stronger than those under DID-IV. Second, while the restrictions under DID-IV are symmetric between exposed and unexposed groups, they are asymmetric under Fuzzy DID.\par
Note that these differences are driven by the fact that Fuzzy DID jointly restricts adoption behavior both in response to the instrument and over time. Specifically, Assumption \ref{chas3}, which is stronger than Assumption \ref{ApeBas2}, imposes time monotonicity and thus rules out one of the two time-switching types ($CM^T$ or $DF^T$). The condition \eqref{cheq1} in Assumption \ref{chas1}, which DID-IV does not impose, rules out the types $DF^{Z} \land AT^{T}$ and $NT^{Z} \land DF^{T}$ in the exposed group and the type $DF^{Z} \land AT^{T}$ in the unexposed group. Finally, Assumption \ref{chas2}, combined with the time monotonicity (Assumption \ref{as33} under $z=0$), eliminates both $DF^{T}$ and $CM^{T}$ in the unexposed group under Fuzzy DID.
\subsection{The difference in target parameter}\label{ApeC.4}
In this section, building on Tables \ref{ApeC.3.3.table1}--\ref{ApeC.3.3.table2}, we next clarify the difference between the SLATET and the LATET. Specifically, we decompose the SLATET under the identifying assumptions in Fuzzy DID and discuss its implications.\par
We prove the following theorem.
\begin{Theorem}
\label{ApeC.4.Theorem1}
Suppose Assumptions \ref{chas1}--\ref{chas5} hold. Then, we can decompose the SLATET as follows, depending on the direction of the two monotonicity conditions (Assumptions \ref{ApeBas2} and \ref{as33} under $z=0$):
\begin{itemize}
\item[(i)] under the monotonicity assumptions that exclude $DF^{Z}$ and $DF^{T}$, we have
\begin{align}
\footnotesize
&SLATET \notag\\
&=\E[Y(1)-Y(0)|Z=1,T=1,AT^{Z} \land CM^{T}]\Prb(AT^{Z} \land CM^{T}|Z=1,SW)\notag\\
\label{ApeC.4.eq1}
&+\E[Y(1)-Y(0)|Z=1,T=1,CM^{Z} \land NT^{T}]\Prb(CM^{Z} \land NT^{T}|Z=1,SW).
\end{align}
\item[(ii)] {\emergencystretch=3em under the monotonicity assumptions that exclude $DF^{Z}$ and $CM^{T}$, we have $SLATET=\E[Y(1)-Y(0)|Z=1,T=1,CM^{Z} \land NT^{T}]$.\par}
\item[(iii)] {\emergencystretch=3em under the monotonicity assumptions that exclude $CM^{Z}$ and $DF^{T}$, we have $SLATET=\E[Y(1)-Y(0)|Z=1,T=1,AT^{Z} \land CM^{T}]$.\par}
\end{itemize}
\end{Theorem}
\begin{proof}
The case (i) directly follows from Table \ref{ApeC.3.3.table1}. The tables for the cases (ii) and (iii) can be constructed analogously, which ensures the corresponding decompositions.
\end{proof}
Theorem \ref{ApeC.4.Theorem1} has several implications. First, the interpretation of the SLATET depends on the direction of the two monotonicity conditions \ref{ApeBas2} and \ref{as33} (under $z=0$). This implies that, if empirical researchers wish to understand which latent type's causal effect is identified by the SLATET, they need to specify the direction of these monotonicity conditions.\par
Second, the SLATET need not be defined solely by treatment variation induced by the instrument. In case (i), it also includes treatment effects for the type $AT^{Z} \land CM^{T}$, whose treatment status changes over time but not in response to the instrument; in case (iii), it is defined entirely for this type. When identifying the SLATET in Fuzzy DID, researchers should therefore ask not only whether the instrument represents the policy change of interest, but also whether the treatment variation induced by time is relevant to their policy question (\citeonline{Heckman2001-ur}).\par
Third, the SLATET may be defined on a smaller population than the LATET. In particular, under case (ii), the SLATET identifies the causal effect among the type $CM^{Z} \land NT^{T}$, which is a subpopulation of $CM^{Z}$. In this case, Fuzzy DID thus identifies the treatment effects for a smaller population than DID-IV, under stronger restrictions on treatment adoption behavior.\par

\begin{Remark}\label{ApeC.4.Remark1}
Using the decomposition result for the SLATET, we can also clarify the interpretation of Assumption \ref{Homogeneous1}, which dCDH introduce as an alternative to Assumption \ref{chas2}. If we do not impose Assumption \ref{chas2}, the time-switching types ($CM^{T}$ or $DF^{T}$) may exist in an unexposed group. Under the monotonicity assumptions that exclude $DF^{Z}$ and $DF^{T}$, the only remaining time-switching type in an unexposed group is $AT^{Z} \land CM^{T}$. We can therefore interpret the SLATET' (see Assumption \ref{Homogeneous1}) as the treatment effects among this type:
\begin{align}
\label{ApeC.4.eq2}
SLATET'=\E[Y(1)-Y(0)|Z=0,T=1,AT^{Z} \land CM^{T}].
\end{align}
Combining \eqref{ApeC.4.eq2} with \eqref{ApeC.4.eq1}, we see that Assumption \ref{Homogeneous1} ($SLATET = SLATET'$) does not require homogeneous treatment effects for the same latent type across exposed and unexposed groups. Rather, it equates a weighted average of treatment effects among types $AT^{Z} \land CM^{T}$ and $CM^{Z} \land NT^{T}$ in the exposed group with the treatment effect for type $AT^{Z} \land CM^{T}$ in the unexposed group.
\end{Remark}
\subsection{The use of the Wald-DID estimand}\label{ApeC.5}
In this section, we explain why the Wald-DID estimand requires the stable treatment effect assumption (Assumption \ref{chas5}) under Fuzzy DID but not under DID-IV. We show that this difference arises because the two designs rely on different types of the parallel trends assumption. We also discuss which parallel trends assumption is suitable for the DID-IV setting.\par
\subsubsection{Why the Wald-DID estimand requires Assumption \ref{chas5} under Fuzzy DID}\label{ApeC.5.1}
We first clarify Assumption \ref{chas5} using latent treatment adoption types. This assumption is in fact conditional on $D_0(Z)=1$ in addition to $Z$ and $T$, but it is not immediately clear from dCDH why such a strong restriction is imposed only on this subpopulation.\par
From Tables \ref{ApeC.3.3.table1}--\ref{ApeC.3.3.table2}, we can reinterpret Assumption \ref{chas5} as follows:\footnote{Under the identifying assumptions in Fuzzy DID, the type $DF^{Z} \land AT^{T}$ is eliminated in both groups. Therefore, we write the type $AT^{Z} \land AT^{T}$ as $AT^{T}$.}
\begin{align}
\label{ApeC.5.1.eq1}
&\E[Y(1)-Y(0)|Z=z,T=1,D_0(z)=1]=\E[Y(1)-Y(0)|Z=z,T=0,D_0(z)=1]\notag\\
&\iff \E[Y(1)-Y(0)|Z=z,T=1,AT^{T}]=\E[Y(1)-Y(0)|Z=z,T=0,AT^{T}]
\end{align}
for each $z \in \{0,1\}$.
Here, we use $D_0(0)=D_0(1)$, which follows from the restriction $v_{10}=v_{00}$ in dCDH. Equation \eqref{ApeC.5.1.eq1} shows that Assumption \ref{chas5} actually requires the stable treatment effect only for the time always-takers $AT^{T}$ within each group.\par
We now show that Assumption \ref{chas5} is not required by the Wald-DID estimand per se,
but arises because dCDH impose the parallel trends assumption
on untreated outcomes. The key point is that imposing the parallel trends assumption on untreated outcomes under the DID-IV setting forces us to impute the counterfactual untreated potential outcome for the time always-takers $AT^{T}$, which in turn requires imposing
Assumption \ref{chas5}.\par
To see this formally, we decompose each term in the numerator of the Wald-DID estimand as follows.
\begin{align*}
&\E[Y|Z=0,T=1]-\E[Y|Z=0,T=0]\\
=&(\E[Y(1)-Y(0)|Z=0,T=1,AT^{T}]-\E[Y(1)-Y(0)|Z=0,T=0,AT^{T}])\\
\times & \Prb(AT^{T}|Z=0)+\{\E[Y(0)|Z=0,T=1]-\E[Y(0)|Z=0,T=0]\}.
\end{align*}
\begin{align*}
&\E[Y|Z=1,T=1]-\E[Y|Z=1,T=0]\\
=&(\E[Y(1)-Y(0)|Z=1,T=1,AT^{T}]-\E[Y(1)-Y(0)|Z=1,T=0,AT^{T}])\\
\times & \Prb(AT^{T}|Z=1)+\{\E[Y(0)|Z=1,T=1]-\E[Y(0)|Z=1,T=0]\}\\
+&\E[Y(1)-Y(0)|SW,Z=1,T=1]\Prb(SW|Z=1).
\end{align*}\par
These decompositions clarify why Assumption \ref{chas5} becomes necessary.
Imposing parallel trends on untreated outcomes requires constructing
the average time trends of the untreated potential outcome in both groups (the second term in each conditional expectation).
Since these trends are not observed for time always-takers $AT^{T}$,
this construction necessarily introduces the first term
in the above decompositions, which captures changes
in the average treatment effect among $AT^{T}$.
For the Wald-DID estimand to identify the SLATET,
one must therefore impose stable treatment effects for $AT^{T}$,
as in Assumption \ref{chas5}.\par
In contrast to Fuzzy DID, this requirement does not arise under DID-IV. The key difference is that, in DID-IV designs, the parallel trends assumption
is imposed on unexposed outcomes rather than untreated potential outcomes. As a result, the Wald-DID estimand does not require imputing the counterfactual untreated potential outcomes for time always-takers,
and hence does not require Assumption \ref{chas5}.\par
\subsubsection{Which parallel trends assumption is suitable for the DID-IV setting}\label{ApeC.5.2}
Given that the type of the parallel trends assumption determines whether the Wald-DID estimand requires Assumption \ref{chas5}, a natural question is which parallel trends assumption is suitable for the DID-IV setting. In this section, we discuss two advantages of the parallel trends assumption on unexposed outcomes over that on untreated potential outcomes. First, it is more consistent with the source of identifying variation in DID-IV settings: the parallel trends assumption on untreated potential outcomes relies on variation in treatment, while the parallel trends assumption on unexposed outcomes exploits variation in the instrument.\footnote{Recall that in canonical DID designs, we rely on the parallel trends assumption in untreated outcomes because units are partitioned into treatment and control groups based on the exposure to the treatment in the second period.}\par
Second, the parallel trends assumption on unexposed outcomes is indirectly testable using pre-exposed period data, as we discussed in Section~\ref{sec4.5}. By contrast, the parallel trends assumption on untreated potential outcomes is not testable in general. This is because some units may have already adopted the treatment before period $0$, which prevents a comparison of pre-trends in untreated outcomes between exposed and unexposed groups.\footnote{dCDH propose a placebo test in their appendix, but its scope for
assessing the parallel trends assumption in untreated potential outcomes is somewhat
limited: the test applies only when the share of treated units changes in
neither group between the two pre-exposed periods, and even then it
provides only a necessary condition for Assumptions \ref{chas4} and
\ref{chas5} to hold jointly, not a direct assessment of the trends in
untreated outcomes.}\par
\subsection{Other comparisons}\label{ApeC.7}
In this section, we document the relationship between DID-IV and Fuzzy DID in settings with non-binary, ordered treatments.\footnote{In multiple group settings, the DID-IV design considered in this paper is not directly compatible with the Fuzzy DID framework, for two reasons. First, the group variable $G$ takes a finite number of ordered values: $G \in \{0,1,\dots,\bar{g}\}$. Second, dCDH construct the ``super groups'' from the group structure in the data and use the Wald-DID estimand to identify the average treatment effect among the switchers pooled across all groups.} We also discuss the alternative estimands proposed in dCDH that avoid Assumption \ref{chas5}.
\subsubsection{Non-binary, ordered treatment}\label{ApeC.7.1}
In DID-IV designs, our target parameter is the average causal response on the treated (ACRT), the conditional version of the ACR in \citeonline{Angrist1995-ij} (see also \citeonline{Hudson2017-tm}). By contrast, dCDH consider the average causal response among the switchers, the units whose treatment status increases over time (see Section~3.3 in dCDH).\footnote{In the case of non-binary, ordered treatment, the condition (ii) of the DID-IV setting is satisfied in dCDH because the group variable $G$ is included in their ordered treatment equation (see Assumption 3$'$ in dCDH), while, to the best of our knowledge, the condition (i) is not imposed: the threshold $v^{d}_{GT}$ is allowed to depend on $G$ in period $T=0$. Nevertheless, based on their empirical application (a revisit of \citeonline{Duflo2001-nh}), it is reasonable to infer that dCDH have the DID-IV setting in mind. Given this, the absence of condition (i) can be interpreted as a relaxation of the restrictions on treatment adoption behavior in period $T=0$.}
As in the binary case, their target parameter includes the treatment variation induced by time, and thus differs from the ACRT, which is defined solely by the treatment variation induced by the instrument.\par    
Regarding the identifying assumptions, in DID-IV designs, we consider the same assumptions as in the binary case (see \citeonline{Hudson2017-tm}). By contrast, dCDH maintain Assumptions \ref{chas1} and \ref{chas4} and replace Assumption \ref{chas3} with the following one.
\begin{Assumption}[Ordered treatment equation]
\label{dCDHordered1}
\begin{align*}
D=\sum_{j=1}^{J}\mathbf{1}\{V \geq v^{j}_{ZT}\},\hspace{2mm}\text{with}\hspace{2mm}-\infty=v^{0}_{ZT}<v^{1}_{ZT}<\dots<v^{J+1}_{ZT}=\infty\hspace{2mm}\text{and}\hspace{2mm}V \indep T|Z.
\end{align*}
\end{Assumption}
In addition, they replace Assumption \ref{chas2} with the following condition:
\begin{align}
\label{ch_stochastic1}
D_{11} \succeq D_{10}\hspace{2mm}\text{and}\hspace{2mm}D_{01} \sim D_{00}.
\end{align}
Here, $D_{ZT}$ denotes the treatment conditional on $Z$ and $T$, $\succeq$ denotes stochastic dominance between two random variables, and $\sim$ denotes equality in distribution.\par
The condition \eqref{ch_stochastic1} restricts treatment adoption behavior: as dCDH point out in their proof (see page 1025 in dCDH), this condition combined with Assumption \ref{dCDHordered1} implies $v^{j}_{01}=v^{j}_{00}$ and $v^{j}_{11} \leq v^{j}_{10}$ for all $j \in \{1,\dots,J\}$. That is, the treatment status is stable over time for all units in the absence of exposure to the instrument, and weakly increasing over time for all units under exposure to the instrument. By contrast, in DID-IV designs, we restrict the treatment adoption behavior only through the monotonicity with respect to the instrument in the second period (Assumption \ref{ApeBas2}).\par
Finally, dCDH point out that the Wald-DID estimand identifies their target parameter only if the stable treatment effect assumption is satisfied. As in the binary treatment case, one can show that this requirement arises not from the Wald-DID estimand itself, but from imposing the parallel trends assumption on untreated outcomes.\par
\subsubsection{Alternative estimands}\label{ApeC.7.3}
As an alternative to the Wald-DID estimand, dCDH propose the time-corrected Wald estimand and the changes-in-changes Wald estimand, which do not rely on Assumptions \ref{chas4} and \ref{chas5}. Under both estimands, however, dCDH continue to impose Assumptions \ref{chas1}--\ref{chas3} to identify the SLATET, so the discussions in Sections~\ref{ApeC.3}--\ref{ApeC.4} apply to these estimands as well. These estimands do not require Assumption \ref{chas5} because Assumption \ref{chas4} is not imposed for the time always-takers $AT^{T}$.

\section{Additional tables}\label{OnlineAppTables}
\begin{table}[H]
\centering
\footnotesize
\setlength{\tabcolsep}{2.5pt}
\renewcommand{\arraystretch}{0.97}
\caption{CACRT estimates}
\label{tab:period_specific_cacrt_panels}

\begin{tabular*}{\textwidth}{@{\extracolsep{\fill}}lccccc@{}}
\toprule
\multicolumn{6}{@{}l}{\textit{Panel A: Baseline}}\\
\midrule
& 1947 & 1948 & 1949 & 1950 & 1951 \\
\midrule
Estimate & 0.310 & 0.141 & 0.163 & 0.169 & 0.299 \\
SE & 0.185 & 0.136 & 0.095 & 0.101 & 0.104 \\
95\% CI & $[-0.182, 0.802]$ & $[-0.221, 0.503]$ & $[-0.089, 0.416]$ & $[-0.099, 0.437]$ & $[0.022, 0.576]$ \\
First stage $F$ & 4.4 & 4.7 & 11.6 & 11.5 & 10.9 \\
$F>10$ & -- & -- & $\checkmark$ & $\checkmark$ & $\checkmark$ \\
AR 95\% set & $(-\infty,\infty)$ & $(-\infty,\infty)$ & $[-0.175,0.803]$ & $[-0.177,0.839]$ & $[0.025,1.611]$ \\
Sample size & 3,281 & 3,459 & 3,675 & 3,914 & 4,190 \\
\midrule
& 1952 & 1953 & 1954 & 1955 & 1956 \\
\midrule
Estimate & 0.247 & 0.185 & 0.376 & 0.344 & 0.230 \\
SE & 0.133 & 0.101 & 0.189 & 0.169 & 0.136 \\
95\% CI & $[-0.107, 0.601]$ & $[-0.083, 0.453]$ & $[-0.127, 0.879]$ & $[-0.105, 0.793]$ & $[-0.130, 0.590]$ \\
First stage $F$ & 7.6 & 10.9 & 5.2 & 5.6 & 5.8 \\
$F>10$ & -- & $\checkmark$ & -- & -- & -- \\
AR 95\% set & $\begin{array}{@{}c@{}}
(-\infty,-6.747] \\
{}\cup[-0.162,\infty)
\end{array}$ & $[-0.165,0.960]$ & $\begin{array}{@{}c@{}}
(-\infty,-0.809] \\
{}\cup[-0.039,\infty)
\end{array}$ & $\begin{array}{@{}c@{}}
(-\infty,-1.013] \\
{}\cup[-0.033,\infty)
\end{array}$ & $(-\infty,\infty)$ \\
Sample size & 4,549 & 4,896 & 4,914 & 5,209 & 5,759 \\
\midrule
\multicolumn{6}{@{}l}{\textit{Panel B: With five-year age-bin controls}}\\
\midrule
& 1947 & 1948 & 1949 & 1950 & 1951 \\
\midrule
Estimate & 0.293 & 0.096 & 0.200 & 0.146 & 0.421 \\
SE & 0.166 & 0.155 & 0.114 & 0.127 & 0.197 \\
95\% CI & $[-0.172, 0.758]$ & $[-0.339, 0.531]$ & $[-0.119, 0.519]$ & $[-0.211, 0.504]$ & $[-0.132, 0.973]$ \\
First stage $F$ & 5.2 & 4.1 & 8.8 & 7.0 & 6.0 \\
$F>10$ & -- & -- & -- & -- & -- \\
AR 95\% set & $(-\infty,\infty)$ & $(-\infty,\infty)$ & $[-0.259,2.545]$ & $(-\infty,\infty)$ & $\begin{array}{@{}c@{}}
(-\infty,-1.757] \\
{}\cup[-0.057,\infty)
\end{array}$ \\
Sample size & 3,281 & 3,459 & 3,538 & 3,618 & 3,582 \\
\midrule
& 1952 & 1953 & 1954 & 1955 & 1956 \\
\midrule
Estimate & 0.337 & 0.081 & 0.091 & 0.091 & 0.049 \\
SE & 0.273 & 0.137 & 0.075 & 0.085 & 0.146 \\
95\% CI & $[-0.429, 1.102]$ & $[-0.304, 0.466]$ & $[-0.120, 0.302]$ & $[-0.147, 0.328]$ & $[-0.361, 0.458]$ \\
First stage $F$ & 2.3 & 5.3 & 14.0 & 12.2 & 4.5 \\
$F>10$ & -- & -- & $\checkmark$ & $\checkmark$ & -- \\
AR 95\% set & $(-\infty,\infty)$ & $(-\infty,\infty)$ & $[-0.253,0.419]$ & $[-0.336,0.486]$ & $(-\infty,\infty)$ \\
Sample size & 3,675 & 3,831 & 2,468 & 2,528 & 2,612 \\
\bottomrule
\end{tabular*}

\vspace{0.15em}
\begin{minipage}{\textwidth}
\footnotesize
\textit{Notes}: This table reports the CACRT estimates in each post-exposed cohort (the year in which individuals were aged 14). Panel A reports the estimates from the baseline STS regressions, and Panel B reports the estimates from the saturated STS regressions with the five-year age-bin controls. SEs are bootstrap standard errors, and the 95\% CIs are simultaneous confidence bands constructed by the multiplier bootstrap over the post-exposed cohorts. The first stage $F$ statistic and the row $F>10$ are defined as in Table \ref{tab:summary_returns} in the main text. The AR 95\% set row reports the Bonferroni-adjusted Anderson--Rubin confidence sets using level $\alpha/K$ in each period, with $K=10$. Sample size reports the number of observations used for estimating the corresponding CACRT.
\end{minipage}
\end{table}

\bibliographystyleonline{econ-econometrica.bst}
\bibliographyonline{reference}
\end{document}